%% file: ex_article.tex
\documentclass[hidelinks,onefignum,onetabnum]{siamart250211}

\input{ex_shared}
\usepackage{fullpage}
\ifpdf
\hypersetup{
  pdftitle={A Kinetic Theory Approach to Ordered Fluids},
  pdfauthor={J.~A.~Carrillo, P.~E.~Farrell, A.~Medaglia, U.~Zerbinati}
}
\fi

\usepackage{caption}

\begin{document}

\maketitle

\begin{abstract}
	We develop a unified kinetic theory for ordered fluids, which systematically extends the phase space with the appropriate generalized angular momenta.
	\rev{
		Our theory yields a mesoscopic model for any continuum with microstructure that is characterized by Capriz's order-parameter manifold.
	}
	We illustrate our theory with three running examples: liquids saturated with non-diffusive gas bubbles, liquids composed of calamitic (rod-like) molecules, and liquids composed of calamitic molecules with additional head-to-tail symmetry.
	We discuss the symmetries of the microscopic interactions via Noether's theorem, and use them to characterize the conserved quantities in mesoscopic dynamics.
	We derive the mesoscopic model for ordered fluids from a kinetic point of view, assuming that the microscopic interactions are of weak nature when it comes to the ordering of the fluid.
	Lastly, we discuss under which conditions an H-theorem result holds at the mesoscopic scale and for which Vlasov potential we can expect the emergence of collective behavior.
\end{abstract}



\section{Introduction}

Fluids with microscopic ordering commonly arise in physics, chemistry, and biology. Prominent examples include liquid crystals~\cite{virga,degennes}, liquids saturated with non-diffusive gas bubbles~\cite{carpiz, russoSmereka, russoSmereka2}, polymer mixtures like synovial fluid in joints~\cite{suli}, ferrofluids~\cite{carpiz,nochetto}, and superfluid helium-4 \cite{carpiz,Pavelka}. How does the microscopic ordering influence the macroscopic behavior of such fluids? Ericksen, for instance, pioneered the study of how nematic ordering in liquid crystals leads to anisotropic effects in the elastic properties of the macroscopic fluid \cite{ericksen, ericksen2}.

Several frameworks have been proposed to study this question. Capriz pioneered a unified approach from a continuum mechanics perspective, relating symmetries of the microscopic constituents to the macroscopic behavior of the fluid~\cite{carpiz}. Virga \& Sonnet developed a phenomenological framework deriving appropriate energy and dissipation functionals from macroscopic symmetries~\cite{virgaSonnet}. We also mention the framework of Beris \& Edwards, which exploits thermodynamic brackets to derive models from a continuum thermodynamics perspective~\cite{berisEdwards,GENERIC}.

While there have been kinetic approaches for specific fluids with microscopic ordering~\cite{russoSmereka,russoSmereka2,degond,curtissI,curtissII,curtissIII,curtissIV,curtissV,suli}, there has been no unified treatment of this question from a kinetic point of view. In this work we aim to develop a unified kinetic theory for ordered fluids.
Our theory will exploit the idea of the order parameter manifold introduced by Capriz, systematically extending the phase space with the appropriate generalized angular momenta.
\rev{
	Once appropriate interaction potentials are specified, our theory determines a mesoscopic model for any continuum with microstructure that is characterized by Capriz's order parameter manifold.
}
\rev{
	In particular, in our theory we will adopt a Hamiltonian description based on the canonical Poisson bracket, which poses no restriction on the choice of the order parameter manifold, but comes at the price of a larger phase space and does not encode the conservation laws associated with the symmetries of the order parameter manifold into the kinetic equations.
	In contrast the use of reduced Poisson brackets and the associated GENERIC framework would allow the encoding of the conservation laws associated with the symmetries of the order parameter manifold into the kinetic equations, but would require additional restrictions on the choice of the order parameter manifold, such as symplecticity \cite{GENERIC}.
}

Our theory does not address the question of the \emph{emergence} of microstructural ordering. However, numerous studies have addressed this on a case-by-case basis. For example, in the context of liquid crystals, Onsager explained the mechanisms behind the origin of nematic ordering \cite{onsager}.
\rev{In particular, due to the emergence of such ordering, the macroscopic behavior of liquid crystals is very different from that of simple fluids, presenting phenomena such as anisotropic elasticity, anisotropic viscosity and anisotropic acoustic properties \cite{virga,degennes, farrell}.}

\rev{We employ four running examples to illustrate our theory:
\begin{itemize}
\item Example A: liquids saturated with non-diffusive gas bubbles;
\item Example B: liquids composed of calamitic (rodlike) molecules in two dimensions;
\item Example C: liquids composed of calamitic molecules in two dimensions with additional head-to-tail symmetry;
\item Example D: liquids composed of calamitic molecules in three dimensions.
\end{itemize}}
\subsection*{Notation}
	Vectors and matrices are denoted as $\vec{x}$ and $\mat{A}$. Elements of affine spaces (where differences of elements lie in a vector space) are denoted as $\pt{x}$. Elements of an intrinsic manifold $\mathcal{M}$ are denoted without notation (e.g.~$\nu \in \mathcal{M}$); when $\mathcal{M}$ is considered to be embedded in a vector space, we denote the corresponding vector as $\vec{\nu}$, except for the particular case of the special orthogonal group $\text{SO}(d)$, whose elements we denote by $\mat{Q}$.
	We denote the dyadic product of two vectors $\vec{a}$ and $\vec{b}$ as $\vec{a}\otimes \vec{b}$.
	We denote the Lie algebra associated with the Lie group $\text{SO}(d)$, i.e.~the tangent space to the manifold $SO(d)$ at the identity, as $\mathfrak{so}(d)$.
    We denote the vector in $\mathbb{S}^d$ generating $\nu$ via the dyadic product as $\vec{\nu}$, i.e. $\nu = \vec{\nu}\otimes \vec{\nu}$.
	The mapping $\nu$ and $\vec{\nu}$ is not unique, yet as long as we are interested only in $\vec{\nu}\otimes \vec{\nu}$ the choice of the sign of $\vec{\nu}$ is irrelevant.
	We denote $\mathbb{S}^d$ the unit sphere in $\mathbb{R}^{d+1}$, and $\mathbb{RP}^d$ the real projective space of dimension $d$. 
	
\section{Order parameter and Lagrangian formulation}
In this section we focus our attention on the microscopic constituents of the fluid under consideration.
In particular, we will first introduce Capriz's concept of an order parameter manifold to describe the nature of the ordering of the fluid and the effect that rigid motions have on it.
We will then introduce the Lagrangian description of the dynamics of the microscopic constituents and study the symmetries of the interactions between the microscopic constituents via Noether's theorem~\cite{fasanoMarmi}.
Lastly, thanks to the symmetries of the interactions, we will characterise quantities conserved by the dynamics of the microscopic constituents.

\subsection{Order parameter manifold} 
\label{sec:manifold}
We introduce the concept of an order parameter manifold to describe the nature of the order parameter field describing the ordering of the fluid and the effect that rigid motions have on it.

Let $d\in\mathbb{N}$ be the dimension of Euclidean space $\mathbb{E}^d$ containing the fluid under consideration.
To describe the ordering of the fluid, it is common to introduce $\iota\in \mathbb{N}$ scalar fields $\nu_\alpha: \mathbb{E}^d\to \mathbb{R}$, $\alpha = 1, \dots, \iota$, called order parameters.
Clearly the range of the order parameters is not arbitrary, but rather dependent on the nature of the ordering; for example, an angle should lie within $[0, 2\pi)$.  For this reason we will think of the order parameters $\nu_\alpha$ as coordinates of an element $\nu$ belonging to a smooth compact manifold $\mathcal{M}$ of dimension $\iota$.

We fix a set of coordinate charts that fully describe the manifold.
For each chart, we fix a preferred frame in its domain, i.e.~we fix a parametrization of $\mathcal{M}$.
We assume that translating the fluid in space does not change its microstructure, i.e.~that any translation of $\mathbb{E}^d$ leaves the order parameter unchanged. However, rotating the fluid may change our representation of its microstructure. For example, in a calamitic fluid with rodlike constituents, rotating the fluid will rotate the alignment of the molecules.
To describe the effect of rotations, we consider the group $\text{SO}(d)$ acting on $\mathbb{E}^d$.
We then introduce a group action $\mathcal{A}: \text{SO}(d)\times \mathcal{M}\to \mathcal{M}$ describing the effect of rotations on the order parameter.
\begin{definition}
	\label{def:orderParameterManifold}
	We say that the tuple $(\mathcal{M}, \mathcal{A})$ is an order parameter manifold if $\mathcal{M}$ is a smooth compact manifold equipped with a fixed parametrization, and $\mathcal{A}$ is a Lie group action of $\text{SO}(d)$ on $\mathcal{M}$, i.e. the map $\mathcal{A}$ is smooth enough to be differentiable.
	Furthermore, we say that a field $\nu: \mathbb{E}^d\to \mathcal{M}$ is an order parameter field if $\Forall \vec{c} \in \mathbb{R}^d$ and $\Forall \mat{Q} \in \text{SO}(d)$ we have
	\begin{equation}
		\nu(\mat{Q}\pt{x} + \vec{c}) = \mathcal{A}(\mat{Q}, \nu(\pt{x})), \quad \Forall \pt{x} \in \mathbb{E}^d.
	\end{equation}
\end{definition}
\rev{Perhaps a more appropriate name for the concept introduced in the previous definition would be the one of ``microstructure manifold'', since it describes the nature of the microstructure of the fluid, yet following Capriz's choice of naming $\nu$ as order parameters we decided to opt for the name of order parameter manifold \cite{carpiz}.}
\begin{remark}
	The order parameter manifold introduced in Definition \ref{def:orderParameterManifold} is equivalent to a G-space \cite[Chapter 7, p.146]{lee}, which is a widely used concept in the context of Lie groups and Lie algebra actions on manifolds.
\end{remark}
\begin{exA}[Nondiffusing gas bubbles: Order parameter manifold]
	\label{ex:bubbles}
	We think of the order parameter as a scalar field $\nu$ that describes the volume fraction of gas in the liquid.
	The order parameter manifold is the interval $[0,1]$ together with the trivial group action of $\text{SO}(d)$, which leaves the order parameter unchanged.
    Notice that this parametrization excludes the possibility that a bubble might carry rotational energy.
\end{exA}
\begin{exB}[Calamitic molecules in two dimensions: Order parameter manifold]
	\label{ex:2D}
	We consider a calamitic liquid crystal occupying a two-dimensional domain.
	The liquid crystal is a fluid whose microscopic constituents are calamitic molecules of vanishing girth.
	We model such molecules as segments of uniform length in $\mathbb{R}^2$.
	Thus, we can imagine the state of a stationary molecule by the position of its center of mass $\pt{x}$ and a vector $\vec{\nu}$ pointing in the direction of the molecule.
	Since no additional information about the state of the molecule is provided by norm of $\vec{\nu}$, we normalize $\vec{\nu}$ and consider the unit circle $\mathcal{M} = \mathbb{S}^1$ as order parameter manifold together with the action of $\text{SO}(2)$, $\mathcal{A}(\mat{Q}, \vec{\nu}) = \mat{Q}\vec{\nu}$.
\end{exB}
\begin{exC}[Head-tail symmetric molecules in two dimensions: Order parameter manifold]
	\label{ex:head-tail}
	The previous example neglects a key feature of many rodlike molecules: their head-tail symmetry. Incorporating the fact that $\vec{\nu}$ and $-\vec{\nu}$ represent the same state leads us to consider the real projective space $\mathbb{RP}^1$ as order parameter manifold.
	Using angles $\theta$ to parametrize $\mathbb{S}^1$ induces a parametrization of the real projective space $\mathbb{RP}^1$ as 
	\begin{equation}
		\label{eq:head-tail-parametrization}
		\theta\mapsto
		\begin{bmatrix} \cos(\theta) \\ \sin(\theta) \end{bmatrix} \otimes \begin{bmatrix} \cos(\theta) \\ \sin(\theta) \end{bmatrix} = \begin{bmatrix} \cos^2(\theta) & \cos(\theta)\sin(\theta) \\ \cos(\theta)\sin(\theta) & \sin^2(\theta) \end{bmatrix},
	\end{equation}
    which is invariant under the map $\theta \mapsto \theta + \pi$.
	To characterize the order parameter manifold associated with head-tail symmetric rodlike molecules it remains to specify the action of $\text{SO}(2)$ on $\mathbb{RP}^1$, which we choose as
	\begin{equation}
		\mathcal{A}: \text{SO}(2)\times \mathbb{RP}^1\to \mathbb{RP}^1, \quad (\mat{Q}, \nu)\mapsto (\mat{Q}\nu)\otimes (\mat{Q}\nu).
	\end{equation}
\end{exC}
\begin{exD}[Calamitic molecules in three dimensions: Order parameter manifold]
	\label{ex:3D}
    We consider the three-dimensional analogue of \cref{ex:2D}.
	We model the molecules as segments of uniform length in $\mathbb{R}^3$.
	Hence, the order parameter manifold is the unit sphere $\mathcal{M} = \mathbb{S}^2$ together with the action of $\text{SO}(3)$, $\mathcal{A}(\mat{Q}, \vec{\nu}) = \mat{Q}\vec{\nu}$.
\end{exD}
\begin{remark}[Landau--de Gennes Q-tensor]
	\label{rem:Qtensor}
	One could also consider as order parameter manifold the set of symmetric traceless second order tensors.
	In the context of rigid bodies, this would be equivalent to endow each molecule with a description of its inertia tensor, or to be more precise its deviation from the inertia tensor of a sphere.
	The authors believe that this example would introduce unnecessary complications in the exposition of the theory, and thus we will not consider it in the rest of the paper.
	However, the development of the theory presented in this paper to the case of the Landau--de Gennes Q-tensor will be the subject of upcoming work.
\end{remark}
For a fixed order parameter $\nu\in \mathcal{M}$, the orbit of $\nu$ under the action of $\text{SO}(d)$
\begin{equation}
	\text{SO}(d)\nu \coloneqq \left\{\mathcal{A}(\mat{Q}, \nu) \;:\; \mat{Q}\in \text{SO}(d)\right\}
\end{equation}
is the set of all order parameters that can be obtained from $\nu$ via a change of frame in space. The orbits partition $\mathcal{M}$ into equivalence classes.
\begin{definition}
  The action $\mathcal{A}$ is said to be transitive if for any $\nu \in \mathcal{M}$ the orbit $\text{SO}(d)\nu$ is the whole manifold $\mathcal{M}$. \\
\end{definition}
\begin{remark}
	Transitivity of the action $\mathcal{A}$ is equivalent to the statement that all the possible order parameter configurations on the manifold $\mathcal{M}$ can be achieved by a change of frame. The action of \cref{ex:bubbles} is not transitive, while all other examples are transitive.
\end{remark}
For a fixed $\nu \in \mathcal{M}$, the orbit map
\begin{equation}
	\mathcal{A}_\nu: \text{SO}(d)\to \text{SO}(d)\nu, \quad \mat{Q}\mapsto \mathcal{A}(\mat{Q}, \nu),
\end{equation}
is differentiable at the identity \cite[Chapter 7, p.~146]{lee}. We will denote $A_\nu: \mathfrak{so}(d) \to T_\nu\mathcal{M}$ the differential of $\mathcal{A}_\nu$ evaluated at the identity.
The maps $A_\nu$ are the infinitesimal generators of the action of $\text{SO}(d)$ on the manifold $\mathcal{M}$, i.e.~they describe the effect of infinitesimal rotations on the order parameter $\nu$. Lastly we denote by $A$ the map $\nu\mapsto A_\nu$ which associates to any order parameter $\nu$ its infinitesimal generator $A_\nu$.\\

\begin{remark}
The map $A$ will be useful below in characterizing the angular momentum of the fluid in terms of the order parameter manifold.
\end{remark}

\begin{exD}[Calamitic molecules in three dimensions: Infinitesimal generators]
	\label{ex:infinitesimalGenerator3D}
	We continue Example \ref{ex:3D}. We wish to understand the group action associated with the order parameter manifold $\mathcal{A}(\mat{Q}, \vec{\nu}) = \mat{Q}\vec{\nu}$.
	We parametrize elements of $SO(3)$ via a rotation vector $\vec{q}\in \mathbb{R}^3$ whose norm is the angle of rotation and whose direction is the axis of rotation, i.e.\\
	\begin{equation}
		\label{eq:rodrigues}
		\mat{Q}(\norm{q},\norm{q}^{-1}\vec{q}) = e^{\mat{W}(\vec{q})} = \mat{I} + \sin(\|\vec{q}\|)\mat{W}(\vec{q}) + (1-\cos(\|\vec{q}\|))\mat{W}(\vec{q})^2,
	\end{equation}
	where $\mat{W}(\vec{q})$ is the skew-symmetric matrix associated with the vector $\vec{q}$ via the cross product, and $\mat{I}$ is the identity matrix. To obtain the last identity we used Rodrigues' formula \cite[Chapter 4, eq.~4.62]{goldstein}.
	Using this representation of the element of $SO(3)$ we can compute the differential of the orbit map as
	\begin{equation}
		A_{\vec{\nu}}: \mathbb{R}^3\to T_\nu\mathcal{M},\quad
		\vec{q} \mapsto \frac{\partial}{\partial{\norm{\vec{q}}}} \mat{Q}(\norm{\vec{q}}, \norm{q}^{-1}\vec{q}){\Big|_{_{\norm{q} = 0}}}\!\!\!\!\!\!\vec{\nu}= \mat{W}(\vec{q})\vec{\nu}= \vec{q}\times \vec{\nu}.
	\end{equation}
	By a slight abuse of notation we will denote both $A_{\vec{\nu}}: \mathfrak{so}(3)\to T_\nu\mathcal{M}$ and $A_{\vec{\nu}}: \mathbb{R}^3\to T_\nu\mathcal{M}$ with the same symbol.
\end{exD}
\begin{exB}[Calamitic molecules in two dimensions: Infinitesimal generators]
	\label{ex:infinitesimalGenerator2D}
	We continue \cref{ex:2D}. We wish to understand the group action associated with the order parameter manifold $\mathcal{A}(\mat{Q}, \vec{\nu}) = \mat{Q}\vec{\nu}$.
	Since the Lie algebra of $\text{SO}(2)$ is Abelian, its associated Lie bracket is zero.
	If we parametrize elements of $SO(2)$ via a rotation angle $\theta\in \mathbb{R}$, i.e.
	\begin{equation}
		\mat{Q}(\theta) = \begin{bmatrix}
			\cos(\theta) & -\sin(\theta)\\
			\sin(\theta) & \cos(\theta)
		\end{bmatrix},
	\end{equation}
	it is easy to see that the differential of the orbit map evaluated at $\theta = 0$ is constant, i.e.
	\begin{equation}
		A_{\vec{\nu}}: \mathbb{R}\to T_\nu\mathcal{M},\qquad 
		\theta \mapsto \partial_{\theta} \mat{Q}(\theta){\Big|_{_{\theta = 0}}}\!\!\!\!\!\!\vec{\nu}= \begin{bmatrix}
			0 & -1\\
			1 & 0
		\end{bmatrix}\vec{\nu}.
	\end{equation}
	It will be useful to observe that if we embed the unit circle $\mathbb{S}^1$ in $\mathbb{R}^3$ then we can consider the group $\text{SO}(2)$ as the subgroup of $\text{SO}(3)$ that leaves the $z$-axis invariant.
	Using the parametrization \eqref{eq:rodrigues} of the element of $SO(3)$ we can represent the element of $\text{SO}(2)$ as vector in $\mathbb{R}^3$, with third component equal to zero.
	Then we can compute the differential of the orbit map to obtain the infinitesimal generator of the action of $\text{SO}(3)$ on the unit circle $\mathbb{S}^1$ embedded in $\mathbb{R}^3$, as
	\begin{equation}
		A_{\vec{\nu}}: \mathbb{R}^3\to T_\nu\mathcal{M},\qquad 
		\vec{q} \mapsto \partial_{\norm{\vec{q}}} \mat{Q}(\norm{\vec{q}}, \norm{q}^{-1}\vec{q}){\Big|_{_{\norm{q} = 0}}}\!\!\!\!\!\!\vec{\nu}= \vec{q}\times \begin{bmatrix}
			\cos(\theta)\\
			\sin(\theta)\\
			0
		\end{bmatrix},
	\end{equation}
	where in this case $\theta$ is the angle of the vector $\vec{\nu}$ in $\mathbb{S}^1$.
\end{exB}
\begin{exC}[Head-tail symmetric molecules in two dimensions: Infinitesimal generators]
	\label{ex:infinitesimalGeneratorHeadTail}
    We continue \cref{ex:head-tail} and compute the infinitesimal generators of the action of $\text{SO}(2)$ on the real projective space $\mathbb{RP}^1$.
	Again considering $SO(2)$ as the subgroup of $SO(3)$ that leaves the $z$-axis invariant, we can parametrize element of $SO(2)$ as a vector in $\mathbb{R}^3$ with third component equal to zero, as in the previous Example.
	We then observe that if we denote by $\mathcal{A}$ the group action of $SO(2)$ on $\mathbb{S}^1$ introduced in \cref{ex:2D}, and by $\mathcal{A}^{\text{ht}}$ the group action of $SO(2)$ on $\mathbb{RP}^1$ introduced in \cref{ex:head-tail}, the following identity holds:
	\begin{equation}
		\mathcal{A}^{\text{ht}}:\mathbb{R}^3\times \mathcal{M}\to \mathcal{M}, \qquad (\vec{q}, \nu) \mapsto \mathcal{A}(\vec{q}, \vec{\nu})\otimes \mathcal{A}(\vec{q},\vec{\nu}).
	\end{equation}
	Thus, we see that the infinitesimal generators $A_{\nu}^{\text{ht}}$ of the action of $SO(2)$ on $\mathbb{RP}^1$ are the third order tensors obtained by
	\begin{equation}
		A_{\nu}^{\text{ht}}: \mathbb{R}^3 \to T_{\nu}\mathcal{M},\qquad
		\vec{q} \mapsto A_{\nu} \vec{q} \otimes\vec{\nu} + \vec{\nu} \otimes A_{\nu} \vec{q}.
	\end{equation}
	Notice that if we parametrize the element of $\mathbb{RP}^1$ as in \cref{ex:head-tail} we can rewrite the previous expression as
	\begin{equation}
		A_{\nu}^{\text{ht}}: \mathbb{R}^3\to T_{\nu}\mathcal{M},\quad \vec{q} \mapsto (\vec{q}\times \vec{\nu})\otimes \vec{\nu} + \vec{\nu}\otimes (\vec{q}\times \vec{\nu}),
	\end{equation}
	thus we can compute $A_{\nu}^{\text{ht}}\vec{q}$ as
	\begin{equation}
		A_{\nu}^{\text{ht}}\vec{q}=\left(\vec{\nu}\times \begin{bmatrix}
			\cos(\theta)\\
			\sin(\theta)\\
			0
		\end{bmatrix}\right)\otimes \begin{bmatrix}
			\cos(\theta)\\
			\sin(\theta)\\
			0
		\end{bmatrix} + \begin{bmatrix}
			\cos(\theta)\\
			\sin(\theta)\\
			0
		\end{bmatrix}\otimes \left(\vec{\nu}\times \begin{bmatrix}
			\cos(\theta)\\
			\sin(\theta)\\
			0
		\end{bmatrix}\right).
	\end{equation}
\end{exC}

Wherever orbits of groups arise, stabilizers arise soon after.
For a fixed order parameter $\nu\in \mathcal{M}$, the stabilizer of $\nu$ under the action of $\text{SO}(d)$ is the subgroup
\begin{equation}
	\text{SO}(d)_{\nu} \coloneqq \left\{{Q}\in \text{SO}(d) \;:\; \mathcal{A}(\mat{Q}, \nu) = \nu\right\}.
\end{equation}
\begin{definition}
We say that the action $\mathcal{A}$ is free if the stabilizer of any order parameter is the trivial subgroup generated by the identity.
\end{definition}
\begin{remark}
	\label{rem:free}
	If the action $\mathcal{A}$ is free, the infinitesimal generators $A_\nu$ are injective \cite[Chapter 7, p.~147]{lee} which means that all rotations have an effect on any order parameter.
\end{remark}

\subsection{Symmetries and conserved quantities}
\label{sec:symmetry}
In this section we investigate the symmetries of the microscopic interactions via Noether's theorem for a generic order parameter manifold $(\mathcal{M}, \mathcal{A})$.

The state of a single microscopic constituent of the fluid under investigation can be described via a Lagrangian description.
To this aim, we introduce the generalized coordinates $\pt{x}$ and $\nu$ to describe the position and the ordering state of the constituent and the respective generalized velocities $\vec{v}$ and $\dot{\vec{\nu}}$.
Under the hypothesis that any constraint acting on the constituent is holonomic, and any active force is conservative, the dynamics of the constituent can be described by the Euler--Lagrange equations
\begin{equation}
	\label{eq:EL}
	\frac{d}{dt}\left(\frac{\partial \mathcal{L}}{\partial \dot{q}}\right) - \frac{\partial \mathcal{L}}{\partial q} = {0}, \quad q = \left(\pt{x}, \nu\right),
\end{equation}
where $\mathcal{L}$ is the Lagrangian of the constituent. We focus our attention on the microscopic constituents for which $\mathcal{L}$ can be decomposed as
\begin{equation}
	\label{eq:decomposition}
	\mathcal{L} = \mathcal{L}_0(\vec{v}) + \mathcal{L}_1(\nu, \vec{\dot{\nu}}), \qquad \mathcal{L}_0(\vec{v}) \coloneqq \frac{1}{2}m(\vec{v}\cdot\vec{v}),\qquad \mathcal{L}_1(\nu, \vec{\dot{\nu}}) \coloneqq \frac{1}{2}\vec{\dot{\nu}}\cdot \mat{B}(\nu) \vec{\dot{\nu}},
\end{equation}
where $m$ is the mass of the constituent, and $\mat{B}(\nu)$ is the generalized inertia matrix, a positive definite matrix.
\rev{
	In particular, through the rest of the paper we will assume that if $\mat{B}$ depends on $\nu$, then $\mat{B}$ transforms objectively under the action of $SO(d)$, i.e.
	\begin{equation}
		\label{eq:covariant}
		\mat{B}(\mathcal{A}(\mat{Q}, \nu)) = \mat{Q}\mat{B}(\nu)\mat{Q}^T, \qquad \Forall \mat{Q}\in SO(d), \Forall \nu\in \mathcal{M}.
	\end{equation} 
}
We postulate that the Lagrangian satisfies the frame indifference constraint
\begin{equation}
	\label{eq:frameIndifference}
	\rev{
		\mathcal{L}(\mat{Q} \pt{x},
			\mat{Q}\vec{v},
			\mathcal{A}(\mat{Q},\nu),
			dA_{\mat{Q}}^{-1}(\mathcal{A}(\mat{Q},\nu))\dot{\nu}) = \mathcal{L}(\vec{v}, \dot{\nu}),
	}
\end{equation}
for any $\mat{Q}\in \text{SO}(d)$, {where $dA_{\mat{Q}}:T_\nu \mathcal{M}\to T_{\mathcal{A}(\mat{Q},\nu)}$ is the tangent map associated with the action $\mathcal{A}:SO(d)\times \mathcal{M}\to \mathcal{M}$ for a fixed $\mat{Q}\in SO(d)$.}
\begin{example}
\rev{In the setting of \Cref{ex:2D,ex:3D}, where the group action is given by $\mathcal{A}(\mat{Q}, \nu) = \mat{Q}\nu$, the frame indifference constraint \eqref{eq:frameIndifference} can be rewritten as
\begin{equation}
	\label{eq:frameIndifference2}
	\mathcal{L}(\mat{Q} \pt{x},
	\mat{Q}\vec{v},
	\mat{Q}\nu,
	\mat{Q}\vec{\dot{\nu}}) = \mathcal{L}(\pt{x},\vec{v},\nu, \dot{\nu}),
\end{equation}
since the tangent map associated with the action of $SO(d)$ on $\mathbb{S}^d$ is given by $dA_{\mat{Q}}(\nu) = \mat{Q}$, thus reducing the right hand side of \eqref{eq:frameIndifference} to the pull-back of Lagrangian by the action of $SO(d)$ on $\mathbb{R}\times \mathcal{M}$.
}
\end{example}
In particular, we observe that the frame indifference constraint is satisfied by $\mathcal{L}_0$, while for $\mathcal{L}_1$ Capriz expresses the frame indifference constraint as~\cite[(5.11)]{carpiz}
\begin{equation}
	\label{eq:frameIndifferenceLagrangian}
	\frac{\partial\mathcal{L}_1(\nu, \vec{\dot{\nu}})}{\partial \mat{Q}}\coloneqq {A^T_{\nu}}{\frac{\partial \mathcal{L}_1}{\partial \nu}} + {\left(\frac{\partial A_{\nu}}{\partial \nu}\vec{\dot{\nu}}\right)^T}{\frac{\partial \mathcal{L}_1}{\partial \vec{\dot{\nu}}}} = \vec{0}.
\end{equation}
\rev{Thus the frame indifference constraint for $\mathcal{L}_1$ has been expressed as the constraint that the Lie derivative of $\mathcal{L}_1$ along the infinitesimal generators of the action $\mathcal{A}$ is zero. An in-depth discussion of \eqref{eq:frameIndifferenceLagrangian} can be found in \cref{sec:frameIndifference} and in \cite{carpiz}.}
\begin{proposition}
	\label{prop:constant}
	If the group action $\mathcal{A}$ is transitive and $\mathcal{L}_1$ in \eqref{eq:decomposition} is frame indifferent, then the matrix $\mat{B}(\nu)$ is a constant.
\end{proposition}
\begin{proof}
	\rev{
	Let us consider two different order parameters $\nu_1, \nu_2\in \mathcal{M}$ and observe that since the group action $\mathcal{A}$ is transitive, there always exists a rotation $\mat{Q}\in \text{SO}(d)$ such that $\nu_2 = \mathcal{A}(\mat{Q}, \nu_1)$.
	Then, using the frame indifference constraint for $\mathcal{L}_1$ we can write
	\begin{equation}
		\mathcal{L}_1(\nu_2, \vec{\dot{\nu}}) = \mathcal{L}_1(\mathcal{A}(\mat{Q}, \nu_1), \vec{\dot{\nu}}) = \mathcal{L}_1(\nu_1, \vec{\dot{\nu}}).
	\end{equation}
	Thus expanding the definition of $\mathcal{L}_1$ we can write
	\begin{equation}
		\vec{\dot{\nu}} \cdot \mat{B}(\nu_1) \vec{\dot{\nu}} = \vec{\dot{\nu}} \cdot \mat{B}(\nu_2) \vec{\dot{\nu}}.
	\end{equation}
	Since we can choose $\vec{\dot{\nu}}$ arbitrarily, we conclude that $\mat{B}(\nu_1) = \mat{B}(\nu_2)$, thus proving the claim.
	}
\end{proof}
\rev{
	Notice that in the setting of \cref{ex:2D} and \cref{ex:3D}, the previous theorem applies because $\mat{B}(\nu)$ is a scalar multiple of the identity matrix and we chose as conjugate variable $\vec{\dot{\nu}}$ the time derivative of a vector.
	If we would have expressed the Lagrangian $\mathcal{L}_1$ in terms of the angular velocity $\vec{\omega}$ and the inertia tensor $\mat{\mathbb{I}}(\nu)$, the previous theorem would not have applied since keeping fixed $\omega$ and changing $\nu$ would have changed the value of the Lagrangian, hence we could not have applied the frame indifference constraint to conclude that $\mat{\mathbb{I}}(\nu)$ is constant.
}
\begin{corollary}
\label{cor:diagonal}
Under the hypotheses of Proposition \ref{prop:constant}, there exists a change of coordinates for the tangent space to the manifold $\mathcal{M}$ such that the matrix $\mat{B}(\nu)$ is constant and diagonal.
\end{corollary}
\begin{proof}
	We begin diagonalising the matrix $\mat{B}$, i.e.~we find a change of basis $\mat{V}$ for the tangent space to the manifold $\mathcal{M}$ such that $\mat{V}^T\mat{B}(\nu)\mat{V} = \mat{\Lambda}$, where $\mat{\Lambda}$ is a diagonal matrix.
	We conclude by observing that since the matrix is symmetric positive definite, the diagonal elements of $\mat{\Lambda}$ are positive and $\mat{V}$ is an orthogonal matrix, therefore it induces a well-defined change of coordinates $\cdot \mapsto V\cdot$.
\end{proof}

Let us now study a system formed by two constituents with the aim of understanding the microscopic two body interactions.
We describe the configuration of each constituent via the generalized coordinates $\pt{x}_i$, $\nu_i$ and the respective generalized velocities $\vec{v}_i$, $\vec{\dot{\nu}_i}$.
Furthermore, we will assume that the interactions between the two constituents are conservative and fully described by a potential of the form $\mathcal{W}(\abs{\pt{x}_1-\pt{x}_2}, \nu_1, \nu_2)$.
Once again under the hypothesis that all constraints acting on the constituents are holonomic, we can describe the dynamics of the system by the Euler--Lagrange equations \eqref{eq:EL} with $q_{1,2} = \left(\pt{x}_1, \pt{x}_2, \nu_1,\nu_2\right)$ and the Lagrangian
\begin{equation}
	\label{eq:lagrangian}
	\mathcal{L}(q_{1,2},{\dot{q}}_{1,2}) = -\mathcal{W}(\abs{\pt{x}_1-\pt{x}_2}, \nu_1, \nu_2)+\sum_{i=1}^2\left(\mathcal{L}_{0,i}(\vec{v}_i)+\mathcal{L}_{1,i}(\nu_i, \vec{\dot{\nu}}_i) \right) ,
\end{equation}
where
\begin{equation}
    \label{eq:afterlagrangian}
	\begin{alignedat}{2}
		\mathcal{L}_{0,1}(\vec{v}_1) &\coloneqq \frac{1}{2}m_1(\vec{v}_1\cdot\vec{v}_1),\qquad\qquad
		\mathcal{L}_{0,2}(\vec{v}_2) \coloneqq \frac{1}{2}m_2(\vec{v}_2\cdot\vec{v}_2),\\
		\mathcal{L}_{1,1}(\nu_1, \vec{\dot{\nu}}_1) &\coloneqq \frac{1}{2}\vec{\dot{\nu}}_1\cdot \mat{B}_1(\nu_1) \vec{\dot{\nu}}_1,\;\;\;\;\;\;
		\mathcal{L}_{1,2}(\nu_2, \vec{\dot{\nu}}_2) \coloneqq \frac{1}{2}\vec{\dot{\nu}}_2\cdot \mat{B}_2(\nu_2) \vec{\dot{\nu}}_2.
	\end{alignedat}
\end{equation}
We now proceed to study the effect of the rotational symmetry and the frame indifference constraint of the Lagrangian $\mathcal{L}$ on the dynamics of the system, via Noether's theorem.
Noether's theorem states that if a Lagrangian $\mathcal{L}$ is invariant under a group action with infinitesimal generators $G$, then
\begin{equation}
	\label{eq:Noether}
	\frac{d}{dt}\left(\frac{\partial \mathcal{L}}{\partial {\dot{q}}_{1,2}}\cdot {G}\right)=0, \qquad q_{1,2} = \left(\pt{x}_{1}, \pt{x}_{2}, \nu_{1}, \nu_{2}\right),
\end{equation}
where
\begin{equation}
\frac{\partial \mathcal{L}}{\partial {\dot{q}}_{1,2}} = (m_1 \vec{v}_1, m_2 \vec{v}_2, \mat{B}_1(\nu_1) \vec{\dot{\nu}}_1, \mat{B}_2(\nu_2) \vec{\dot{\nu}}_2).
\end{equation}
\begin{proposition}
	\label{prop:linear}
	The total linear momentum of the system composed of two constituents, defined as 
	\begin{equation}
		\label{eq:conservation_linear}
		\vec{p}_1 + \vec{p}_2, \quad \vec{p}_i \coloneqq m_i\vec{v}_i,
	\end{equation}
	is preserved by any two body interaction whose potential is of the form $\mathcal{W}(\abs{\pt{x}_1-\pt{x}_2}, \nu_1, \nu_2)$.
\end{proposition}
\begin{proof}
	We first notice that the potential $\mathcal{W}$ is radially symmetric. Therefore, as discussed in the previous section, translations have no effect on the order parameters.
	This is equivalent to the statement that the Lagrangian $\mathcal{L}$ is invariant under the group action of translations, which have $G=\left(1,1,0,0\right)$ as the infinitesimal generator.
	Thus, \eqref{eq:Noether} implies
	\begin{equation}
		\frac{d}{dt}\left(\vec{p}_1 + \vec{p}_2\right) = 0.
	\end{equation}
\end{proof}

The next result is crucial. It proves that microscopic interaction preserve a generalized notion of angular momentum, which is fully determined by the structure of the manifold.
	We will denote by $A_{\pt{x}}$ the infinitesimal generators of the canonical action of the group $\text{SO}(d)$ on the Euclidean space $\mathbb{E}^d$, evaluated at $\pt{x}$.
	Thus in three dimensions $A_{\pt{x}}$ is given by
	\begin{equation}
		A_{\pt{x}} : SO(3) \to T_\pt{x}\mathcal{M}, \qquad \vec{q}\mapsto \vec{q}\times \pt{x},
	\end{equation}
	where $\vec{q}$ is the rotation vector representing the element of $SO(d)$ that $A_{\pt{x}}$ takes as argument.
	In two dimensions $A_{\pt{x}}$ reduces to
	\begin{equation}
		A_{\pt{x}} : SO(2) \to T_\pt{x}\mathcal{M}, \qquad \vec{q}\mapsto \vec{q}\times
		\begin{bmatrix}
		x^{(x)} \\ x^{(y)} \\ 0
		\end{bmatrix}.
	\end{equation}
\begin{remark}
	In the particular case of \cref{ex:infinitesimalGenerator3D} and \cref{ex:infinitesimalGenerator2D} the infinitesimal generator of $\mathcal{A}$ and the infinitesimal generators of the canonical action of the group $\text{SO}(d)$ on the Euclidean space $\mathbb{E}^d$ are the same.
	Notice that in general $A_{\pt{x}}$ and $A_\nu$ are not the same operator, since the action $\mathcal{A}$ might differ from the canonical action of $\text{SO}(d)$ on $\mathbb{E}^d$ as in \cref{ex:infinitesimalGeneratorHeadTail}.
\end{remark}
\begin{proposition}
	\label{prop:angular}
	The total angular momentum of the system composed of two constituents, defined as
	\begin{equation}
		\label{eq:conservation_angular}
		\rev{
		    A_{\pt{x}} \mat{Q}\cdot \vec{p}_1
		  + A_{\pt{x}} \mat{Q}\cdot \vec{p}_2
		  + A_\nu \mat{Q}\cdot \varsigma_1
		  + A_\nu \mat{Q}\cdot \varsigma_2, \qquad \varsigma_i \coloneqq \mat{B}_i(\nu_i)\vec{\dot{\nu}}_i,
		}
	\end{equation}
	is preserved by any two body interaction whose potential $\mathcal{W}(\abs{\pt{x}_1-\pt{x}_2}, \nu_1, \nu_2)$ is frame indifferent.
\end{proposition}
\begin{proof}
	The frame indifference constraint \eqref{eq:frameIndifference} can be recast as \eqref{eq:Noether} where 
	\begin{equation}
		G = \left(A_{\pt{x}}\mat{Q}, A_{\pt{x}}\mat{Q}, A_\nu\mat{Q}, A_\nu\mat{Q}\right),
	\end{equation} 
	for any element of $\mat{Q}\in SO(d)$.
	Thus, \eqref{eq:Noether} implies	
	\begin{equation}
		\frac{d}{dt}\left(
		    A_{\pt{x}} \mat{Q}\cdot \vec{p}_1
		  + A_{\pt{x}} \mat{Q}\cdot \vec{p}_2
		  + A_\nu \mat{Q}\cdot \mat{B}_1(\nu_1)\vec{\dot{\nu}}_1
		  + A_\nu \mat{Q}\cdot \mat{B}_2(\nu_2)\vec{\dot{\nu}}_2\right) = 0.
	\end{equation}
	\rev{
	Equivalently, we can introduce the conjugate momentum to the order parameter $\nu_i$ as $\vec{\varsigma}_i \coloneqq \mat{B}_i(\nu_i)\vec{\dot{\nu}}_i$, and rewrite the previous expression as
	\begin{equation}
		\frac{d}{dt}\left(
		    A_{\pt{x}} \mat{Q}\cdot \vec{p}_1
		  + A_{\pt{x}} \mat{Q}\cdot \vec{p}_2
		  + A_\nu \mat{Q}\cdot \varsigma_1
		  + A_\nu \mat{Q}\cdot \varsigma_2\right) = 0.
	\end{equation}
	}
\end{proof}
\begin{exA}[Non-diffusing gas bubbles: Angular momentum]
	\label{ex:angular_bubbles}
	In the context of \cref{ex:bubbles} we can easily see that since the group action of $\text{SO}(d)$ on the manifold $\mathcal{M}$ is trivial, Proposition \ref{prop:angular} applies with generator $G = \left(A_{\pt{x}}, A_{\pt{x}}, 0, 0\right)$.
	Thus, Noether's theorem does not provide any additional invariants and we can ignore any angular momentum associated with the order parameter $\nu$.
	In fact, the conservation associated with the generator $A_{\pt{x}}$ is already captured by the conservation of the linear momentum \eqref{eq:conservation_linear}.
\end{exA}
\begin{exD}[Calamitic molecules in three dimensions: Angular momentum]
	\label{ex:angular}
	In the setting of \cref{ex:3D} we can further simplify Proposition \ref{prop:angular}.
	In fact, as discussed in \cref{ex:infinitesimalGenerator3D}, the infinitesimal generators of the action $\mathcal{A}$ under the group $\text{SO}(3)$ evaluated at $\vec{\nu}$ are the matrices $A_{\vec{\nu}}$ representing the cross-product with the vector $\vec{\nu}$.
	Furthermore, in this scenario the action $\mathcal{A}$ coincides with the canonical group action of $\text{SO}(3)$ on $\mathbb{E}^3$, and hence $A_{\pt{x}}$ is the matrix representing the cross-product with the position vector $\pt{x}$.
	Adopting \eqref{eq:rodrigues}, we can use a vector $\vec{q}$ to represent an element of $\text{SO}(3)$ and recast the infinitesimal generator $G$ in Proposition \ref{prop:angular} as \\
	\begin{equation}
		G = \left(\vec{q}\times \pt{x}_1, \vec{q}\times \pt{x}_2, \vec{q}\times \vec{\nu}_1, \vec{q}\times \vec{\nu}_2\right).
	\end{equation}
	Thus, following the same steps as in the proof of Proposition \ref{prop:angular}, we can conclude that
	\begin{equation}
		\frac{d}{dt}\left(
		    \vec{q}\times \pt{x}_1\cdot \vec{p}_1
		  + \vec{q}\times \pt{x}_2\cdot \vec{p}_2
		  + \vec{q}\times \vec{\nu}_1\cdot \mat{B}_1(\vec{\nu}_1)\vec{\dot{\nu}}_1
		  + \vec{q}\times \vec{\nu}_2\cdot \mat{B}_2(\vec{\nu}_2)\vec{\dot{\nu}}_2\right) = 0.
	\end{equation}
	Using the triple product identity we can rewrite the previous expression as
	\begin{equation}
		\frac{d}{dt}\left(\vec{q}\cdot \left(\pt{x}_1\times \vec{p}_1 + \pt{x}_2\times \vec{p}_2 + \vec{\nu}_1\times \mat{B}_1(\vec{\nu}_1)\vec{\dot{\nu}}_1 + \vec{\nu}_2\times \mat{B}_2(\vec{\nu}_2)\vec{\dot{\nu}}_2\right) \right) = \vec{0},
	\end{equation}
	which, since we picked $\vec{q}$ to be arbitrary, implies the conservation of
	\begin{equation}
		\label{eq:angular_momentum_3D_gen}
		\vec{\eta} \coloneqq \pt{x}_1\times \vec{p}_1 + \pt{x}_2\times \vec{p}_2 + \vec{\nu}_1\times \mat{B}_1(\vec{\nu}_1)\vec{\dot{\nu}}_1 + \vec{\nu}_2\times \mat{B}_2(\vec{\nu}_2)\vec{\dot{\nu}}_2.
	\end{equation}
	While this quantity might feel unfamiliar at first, we can observe that in the vanishing girth limit considered that the quantity \eqref{eq:angular_momentum_3D_gen} is the classical angular momentum of the segments representing the microscopic constituents.
	In fact, for calamitic molecules the classical Lagrangian is \eqref{eq:decomposition} with $\mat{B}(\vec{\nu}_i) = \mat{I}$, where $\mat{I}$ is the identity.
	Thus we can rewrite \eqref{eq:angular_momentum_3D_gen} as
	\begin{equation}
		\vec{\eta} = \pt{x}_1\times\vec{p}_1 + \vec{\nu}_1 \times \vec{\dot{\nu}}_1 + \pt{x}_2\times\vec{p}_2 + \vec{\nu}_2 \times \vec{\dot{\nu}}_2.
	\end{equation}
	Let $\vec{\omega}$ be the angular velocity of the rodlike molecule. Using the triple product identity, together with the well-known property of segment-like rigid bodies that $\dot{\vec{\nu}_i} = \vec{\omega}\times \vec{\nu}_i$ we can rewrite one term of the previous expression as
	\begin{equation}
		\vec{\nu}_i\times \dot{\vec{\nu}_i}= \vec{\nu}_i\times \vec{\omega}_i\times\vec{\nu}_i = (\vec{\nu}_i\cdot \vec{\nu}_i)\vec{\omega} - (\vec{\nu}_i\cdot \vec{\omega}_i)\vec{\nu}_i = \vec{\omega}_i - (\vec{\nu}_i\cdot \vec{\omega}_i)\vec{\nu}_i = \mat{\mathbb{I}}_i\vec{\omega},
	\end{equation}
	where we have used the fact that the inertia tensor of a segment is $\mat{\mathbb{I}}_i\coloneqq \mat{I}-\vec{\nu}_i\otimes \vec{\nu}_i$.
	This calculation recovers the classical definition of angular momentum, i.e.
	\begin{equation}
		\label{eq:calssic_angular_momentum}
		\vec{\eta} = \pt{x}_1\times \vec{p}_1 + \mat{\mathbb{I}}_1\vec{\omega}_1 + \pt{x}_2\times \vec{p}_2 + \mat{\mathbb{I}}_2\vec{\omega}_2.
	\end{equation}
\end{exD}
\begin{exB}[Calamitic molecules in two dimensions: Angular momentum]
	\label{ex:angular2D}
	As discussed in \cref{ex:infinitesimalGenerator2D}, we can consider the action of $\text{SO}(2)$ on the unit circle $\mathbb{S}^1$ embedded in $\mathbb{R}^3$ as the action of a subgroup of $\text{SO}(3)$ on $\mathbb{R}^3$.
	We can then observe that the same procedure used in \cref{ex:angular} can be used to show that the total angular momentum of the system composed of two calamitic molecules in two dimensions is given once again by \eqref{eq:calssic_angular_momentum}.
	In particular, since the $\pt{x}_i$ and $\vec{p}_i$ are two-dimensional vectors, they can be represented as vectors in $\mathbb{R}^3$ with zero third component and thus
	\begin{equation}
		\begin{bmatrix}x_i^{(x)}\\x_i^{(y)}\\0\end{bmatrix}
		\times
		\begin{bmatrix}p_i^{(x)}\\p_i^{(y)}\\0\end{bmatrix}
		= \begin{bmatrix}0\\0\\x_i^{(x)}p_i^{(y)}-x_i^{(y)}p_i^{(x)}\end{bmatrix},
	\end{equation}
	where $x_i^{(x)}$ and $x_i^{(y)}$ are the $x$ and $y$ components of $\pt{x}_i$ respectively.
	Likewise, since rotations only happen in the $xy$-plane, the angular velocity $\vec{\omega}_i$ has zero first and second component, and thus
	\begin{equation}
		\label{eq:inertia_tensor_2D}
		\mathbb{I}_i\vec{\omega}_i = 
		\begin{bmatrix}
			*&* & 0 \\
			*&* & 0 \\
			0 & 0 & I_i
		\end{bmatrix}\begin{bmatrix}0\\0\\\omega_i\end{bmatrix} = \begin{bmatrix}0\\0\\I_i\omega_i\end{bmatrix},
	\end{equation}
	where the matrix $\mathbb{I}_i$ is the inertia tensor of a segment.
	The structure of this matrix is independent of rotations, since it is obtained by changing the basis in which the inertia tensor is expressed via a rotation matrix on the $xy$-plane $Q(\theta_i)$
	\begin{equation}
		\mathbb{I}_i = \begin{bmatrix}\cos(\theta_i) & -\sin(\theta_i) & 0 \\ \sin(\theta_i) & \cos(\theta_i) & 0\\ 0 & 0 & 1\end{bmatrix}\begin{bmatrix}I_i^{(x)} & 0 \\ 0 & I_i^{(y)} \\ 0 & 0 & I_i^{(z)}\end{bmatrix}\begin{bmatrix}\cos(\theta_i) & \sin(\theta_i)&0\\ -\sin(\theta_i) & \cos(\theta_i)&0\\0 & 0 & 1\end{bmatrix}.
	\end{equation}
	Thus, we can see that the total angular momentum \eqref{eq:calssic_angular_momentum} only has a non-zero third component, and can be expressed for a two body interaction using the single scalar quantity
	\begin{equation}
		\label{eq:calssic_angular_momentum_2d}
		\eta \coloneqq x_1^{(x)}p_1^{(y)}-x_1^{(y)}p_1^{(x)} + I_1\omega_1 + x_1^{(x)}p_2^{(y)}-x_1^{(y)}p_2^{(x)} + I_2\omega_2,
	\end{equation}
	which is the classical definition of angular momentum in two dimensions.
\end{exB}
\begin{remark}
	In \eqref{eq:decomposition} we have described the Lagrangian in general form, yet it is important to notice that in the context of \cref{ex:head-tail} the elements of the tangent space $T_\nu\mathcal{M}$ are linear maps from $\mathbb{R}^2$ to $\mathbb{R}^2$. The matrix $\mat{B}(\nu)$ is thus a 4th order tensor and $\vec{\dot{\nu}}$ is an element of $\mathbb{R}^{2\times 2}$.
	We retain the same notation for the sake of simplicity and because we can think of $\vec{\dot{\nu}}$ as a vector in the vector space $\mathbb{R}^{2\times 2}$ and $\mat{B}(\nu)$ as a linear map from $\mathbb{R}^{2\times 2}$ to $\mathbb{R}^{2\times 2}$.
	In this setting the dot product $\cdot$ has to be understood as double contraction, i.e.~$C: D = \sum_{i,j} C_{ij} D_{ij}$.
\end{remark}
\begin{exC}[Head-tail symmetric molecules in two dimensions: Angular momentum]
	\label{ex:angular_head_tail}
	As discussed in \cref{ex:infinitesimalGeneratorHeadTail} we know that in the setting of \cref{ex:head-tail} the infinitesimal generators of the action $\mathcal{A}$ under the group $\text{SO}(2)$ can be computed by embedding the group $\text{SO}(2)$ in $\text{SO}(3)$.
	This process yields that the infinitesimal generators of the action $\mathcal{A}$ are the third order tensors $A_{\nu}^{\text{ht}}$, which associate to a rotation vector $\vec{q}$ the matrix $(\vec{q}\times \vec{\nu})\otimes \vec{\nu} + \vec{\nu}\otimes (\vec{q}\times \vec{\nu})$.
	Thus applying Proposition \ref{prop:angular} we can conclude that the total angular momentum of the system composed of two head-tail molecules in two dimensions is given by
	\begin{equation}
		  \vec{\eta} \coloneqq \vec{q}\times \pt{x}_1\cdot \vec{p}_1 + \vec{q}\times \pt{x}_2\cdot \vec{p}_2
		  +\sum_{i=1}^{2}(\vec{q}\times \vec{\nu}_i)\otimes \vec{\nu}_i:\mat{B}_i(\nu_i)\dot{\vec{\nu}}_i + \vec{\nu}_i\otimes (\vec{q}\times \vec{\nu}_i):\mat{B}_i(\nu_i)\vec{\dot{\nu}}_i
	\end{equation}
	and it is preserved by any two body interaction whose potential $\mathcal{W}(\abs{\pt{x}_1-\pt{x}_2}, \nu_1, \nu_2)$ is frame indifferent.
	Notice that we have also represented the position vector $\pt{x}_i$ as an element of three-dimensional Euclidean space $\mathbb{E}^3$, assuming $\pt{x}_i$ has zero third component.
\end{exC}
\begin{proposition}
	\label{prop:energy}
	Any two body interaction whose potential is of the form of $\mathcal{W}(\abs{\pt{x}_1-\pt{x}_2}, \nu_1, \nu_2)$ preserves the total energy of the system composed of two constituents, which is equal to $\mathcal{H}$, and it is defined as \\
	\begin{equation}
		\mathcal{H}\coloneqq\mathcal{W}(\abs{\pt{x}_1-\pt{x}_2}, \nu_1, \nu_2) + \sum_{i=1}^2 \left( \frac{1}{2}m_i \vec{v}_i\cdot \vec{v}_i + \frac{1}{2}\vec{\dot{\nu}}_i\cdot \mat{B}_i(\nu_i)\vec{\dot{\nu}}_i \right),
	\end{equation}
\end{proposition}
\begin{proof}
	This statement follows from the fact that the Hamiltonian $\mathcal{H}$ does not depend on time, the system is subject only to holonomic constraints, and the kinetic energy is a homogeneous positive definite quadratic form of the generalized velocities, see \cite[Theorem 8.2]{fasanoMarmi}.
\end{proof}
\begin{remark}
	If the potential $\mathcal{W}$ vanishes on a time interval, by Proposition \ref{prop:energy} the total kinetic energy
	\begin{equation}
		\label{eq:conservation_energy}
		\sum_{i=1}^2 \frac{1}{2}m_i \vec{v}_i\cdot \vec{v}_i + \frac{1}{2}\vec{\dot{\nu}}_i\cdot \mat{B}_i(\nu_i)\vec{\dot{\nu}}_i,
	\end{equation}
	is preserved over that time interval.
	This observation will be useful in Section \ref{sec:BBGKY} where a kinetic theory for ordered fluids is developed.
\end{remark}
\begin{remark}
	\label{rmk:homo}
	It is common to work under the hypothesis that the fluid under consideration is homogeneous, i.e. $\mat{B}_1 = \mat{B}_2$ and $m_1 = m_2$.
	In fact a homogeneous fluid is characterised by the fact that all the microscopic constituents are identical.
	For the rest of the paper we assume that the fluid under consideration is homogeneous.
\end{remark}
\subsection{Microscopic collision}
\label{sec:microscopic_collision}
In the following examples we discuss the dynamics of interacting particles described by an order parameter manifold.
This discussion will be useful when developing the kinetic theory for ordered fluids in Section \ref{sec:BBGKY}. As we will see,
the dimension of the manifold and the symmetry of the group action fully determine another manifold to which the post-collisional variables are confined.
\begin{example}[Hard-sphere: Collision]
	\label{ex:spherical}
	Let us consider for the moment the case of simple spherical molecules of radius $R$ and mass $m$, not endowed with any order parameter.
	We would describe the dynamics of the system composed of two spherical molecules with the Lagrangian 
	\begin{equation}
		\mathcal{L}(\pt{x}_1, \pt{x}_2, \vec{v}_1, \vec{v}_2) = L_{0,1}(\vec{v}_1) + L_{0,2}(\vec{v}_2) - \mathcal{W}(\abs{\pt{x}_1-\pt{x}_2}),
	\end{equation}
	where $L_{0,i}(\vec{v}_i) = \frac{1}{2}m\vec{v}_i\cdot \vec{v}_i$ are defined as in \eqref{eq:afterlagrangian} and $\mathcal{W}(\abs{\pt{x}_1-\pt{x}_2})$ is defined as
	\begin{equation}
		\mathcal{W}(\abs{\pt{x}_1-\pt{x}_2}) = \begin{cases}
			0 & \text{if } \abs{\pt{x}_1-\pt{x}_2} > 2R,\\
			\infty & \text{if } \abs{\pt{x}_1-\pt{x}_2} \leq 2R.
		\end{cases}
	\end{equation}
	Applying Noether's theorem, in particular \Cref{prop:linear,prop:energy}, we can conclude that the total linear momentum and total energy of the system are conserved.
	Thus, for fixed $\vec{v}_1$ and $\vec{v}_2$ we know that the system is confined to a manifold of $\mathbb{R}^6$ defined as 
	\begin{equation}
		\mathfrak{M} = \Big\{
			(\vec{v}_1^\prime, \vec{v}_2^\prime) \in \mathbb{R}^6 \; \text{such that \eqref{eq:conservation_linear} and \eqref{eq:conservation_energy} are conserved}
			\Big\},
	\end{equation}
    where $^\prime$ denotes a quantity after collision.
	This manifold is a two-dimensional submanifold of $\mathbb{R}^6$, which can be parametrized by $\vec{n}\in \mathbb{S}^2$. The resulting collision rule is the classical one
	\begin{equation}
		\vec{v}_1^\prime = \vec{v}_1 - \vec{n}\left[\vec{n}\cdot (\vec{v}_1-\vec{v}_2)\right],\qquad \vec{v}_2^\prime = \vec{v}_2 + \vec{n}\left[\vec{n}\cdot (\vec{v}_1-\vec{v}_2)\right].
	\end{equation}
	\begin{figure}
	\centering
	\scalebox{1.0}{\input{Figures/collision_sphere.tikz}}
	\caption{Collision of two hard spheres, moving with velocities $v_1$ and $v_2$, respectively, and with centers of mass $\pt{x}_1$ and $\pt{x}_2$.}
	\end{figure}
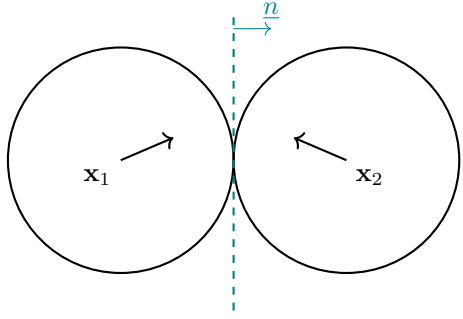
\end{example}
\begin{exA}[Non-diffusing gas bubbles: Collision]
	\label{ex:bubbles_collision}
	Since as discussed in \cref{ex:angular_bubbles} we have no additional invariants in the setting of \cref{ex:bubbles}, the classical collision rule for spherical molecules described here can be used in the context of nondiffusing gas bubbles,
	together with the additional assumption that the volume fraction $\nu$ is exchanged during the collision, i.e.
	\begin{equation}
		\label{eq:volume_fraction_bubbles}
		\nu_1^\prime + \nu_2^\prime = \nu_1 + \nu_2.
	\end{equation}
	The binary rule for the exchange of the volume fraction can thus be rewritten as
	\begin{equation}
		\label{eq:binary_rule_bubbles}
		\vec{\nu}_1^\prime = (1-q)\nu_1 + q \nu_2, \qquad \vec{\nu}_2^\prime = (1-q)\nu_2 + q \nu_1,
	\end{equation}
	where $q\in [0,1]$ is a parameter describing how much volume is exchanged.
	\begin{figure}
		\centering
		\scalebox{0.7}{\input{Figures/collision_bubbles.tikz}}
		\caption{Collision of two gas bubbles. As a result of the collision the volume fraction of each bubble is exchanged, with~$\nu_1^\prime + \nu_2^\prime = \nu_1 + \nu_2$.}
	\end{figure}
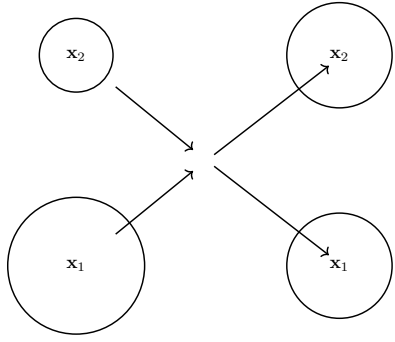
\end{exA}
\begin{exD}[Calamitic molecules in three dimensions: Collision]
	\label{ex:collision_response}
	In the setting of \cref{ex:3D} we can consider a {\it collision response model} for convex body collision \cite{stronge}.
	In particular, we assume that the segment-like molecule $2$ collides on the side of the segment-like molecule $1$, and induces a collision impulse $J$, defined as 
	\begin{equation}
	\label{eq:collision_impulse}
	J = -\dfrac{\vec{V} \cdot \vec{n}}{\frac{2}{m} + \left[ \mathbb{I}^{-1}_1 (\vec{r}_1\times \vec{n})\times \vec{r}_1 + \mathbb{I}^{-1}_2 (\vec{r}_2\times \vec{n})\times \vec{r}_2\right]\cdot \vec{n}},
	\end{equation}
	where $\mathbb{I}_1$ and $\mathbb{I}_2$ are the inertia tensors, $\vec{r}_{1,2}$ are the distances between the centers of mass and the contact point of the collision, $\vec{V}=\vec{v}_2-\vec{v}_1$ is the relative velocity of the center of mass of the two molecules, and $\vec{n}$ is the normal vector to the plane of the collision.
	With this, the post-collisional velocities and angular velocities are given by the following binary rule
    \begin{subequations}
		\label{eq:binary_rule}
	\begin{alignat}{2}
		\vec{v}'_1 &= \vec{v}_1 - (1+e_{{v}})\frac{J}{m} \vec{n},\qquad && \vec{v}'_2  = \vec{v}_2 + (1+e_{{v}})\frac{J}{m} \vec{n}, \\
		\vec{\omega}'_1  &= \vec{\omega}_1 - (1+e_{{\omega}})J \mathbb{I}^{-1}_1 (r_1\times \vec{n}), \qquad &&\vec{\omega}'_2  = \vec{\omega}_2 + (1+e_{{\omega}})J \mathbb{I}^{-1}_2 (r_2\times \vec{n}),
	\end{alignat}
    \end{subequations}
	where $e_{{v}}$ and $e_\omega$ are the coefficients of restitution for the linear and angular velocities, respectively. We assume $e_v = e_\omega = 1$, so that 
	this model of microscopic interaction always preserves linear momentum, angular momentum, and the total energy.

	Analogously to \cref{ex:spherical}, assuming the colliding segments have the same inertia tensor, for fixed $\vec{v}_1, \vec{v}_2, \vec{\omega}_1, \vec{\omega}_2$ we know that the system is confined to a manifold of $\mathbb{R}^{12}$ defined as
	\begin{equation}
		\mathfrak{M} = \Big\{
			(\vec{v}_1^\prime, \vec{v}_2^\prime, \vec{\omega}_1^\prime, \vec{\omega}_2^\prime) \in \mathbb{R}^{12} \; \text{such that \eqref{eq:conservation_linear}, \eqref{eq:conservation_angular} and \eqref{eq:conservation_energy} are conserved}.
			\Big\}
	\end{equation}
	Notice that this is a five dimensional submanifold of $\mathbb{R}^{12}$, which can be parametrized by $(\vec{n}, \vec{r}_{1,2})\in \mathbb{S}^2\times \mathbb{R}^3$ via \eqref{eq:binary_rule} {\cite[p.~37]{stronge}}.
	Notice that if instead of using the angular velocities $\vec{\omega}_i$ we use the total time derivative of the order parameter $\vec{\dot{\nu}}_i$, we can decrease the dimension of the manifold by one, in fact in this case
	\begin{equation}
		\mathfrak{M} = \Big\{
			(\vec{v}_1^\prime, \vec{v}_2^\prime, \vec{\dot{\nu}}_1^\prime, \vec{\dot{\nu}}_2^\prime) \in \mathbb{R}^{10} \; \text{such that \eqref{eq:conservation_linear}, \eqref{eq:conservation_angular} and \eqref{eq:conservation_energy} are conserved}
			\Big\},
	\end{equation}
	which is a four-dimensional submanifold of $\mathbb{R}^{10}$.
\end{exD}
\begin{exB}[Calamitic molecules in two dimensions: Collision]
	\label{ex:response2D}
	Once again in the setting of \cref{ex:2D} we consider the group action of $\text{SO}(2)$ on the unit circle $\mathbb{S}^1$ embedded in $\mathbb{R}^3$ as the action of a subgroup of $\text{SO}(3)$ on $\mathbb{R}^3$.
	Thus, similar to what was done in \cref{ex:angular2D} we can compute the collision impulse $J$ in the 2D setting.
    We begin by computing the cross product of the distances between the centers of mass and the contact point of the collision with the normal vector to the plane of the collision:
    \begin{equation}
    (r_i\times \vec{n}) = \begin{bmatrix}
    0 \\ 0 \\ \vec{r}_i^{(x)}n^{(y)} - \vec{r}_i^{(y)}n^{(x)}
    \end{bmatrix}.
    \end{equation}
    We can then apply the pseudo inverse of the inertia tensors to the cross products, yielding
    \begin{equation}
    	\label{eq:angular_velocity_update_2d}
    \mathbb{I}^{\dagger}_i (\vec{r}_i\times \vec{n}) = -I_3^{-1}\begin{bmatrix}
    0 \\ 0 \\ \vec{r}_i^{(x)}n^{(y)} + \vec{r}_i^{(y)}n^{(x)}
    \end{bmatrix},
    \end{equation}
    where we used the fact that the inertia tensor $\mathbb{I}_i$ is the one of \eqref{eq:inertia_tensor_2D} and we employ the pseudo inverse (denoted with $^\dagger$) to compute the minimum-norm solution.
    Finally, we compute
    \begin{equation}
    	\mathbb{I}^{\dagger}_i (\vec{r}_i\times \vec{n})\times \vec{r}_i = I_3^{-1}\!\!\begin{bmatrix}
    		-\vec{r}_i^{(x)}\vec{r}_i^{(x)}\nu_1^{(y)} - \vec{r}_i^{(x)}\vec{r}_i^{(y)}\nu_1^{(x)} \\
			\;\,\vec{r}_i^{(y)}\vec{r}_i^{(x)}\nu_1^{(x)} + \vec{r}_i^{(y)}\vec{r}_i^{(y)}\nu_1^{(y)} \\ 0
    	\end{bmatrix},
    \end{equation}
    which we employ in \eqref{eq:collision_impulse} to compute $J$.
    The post-collisional velocities and angular velocities are given by the same binary rule \eqref{eq:binary_rule}.
    In particular, since the angular velocities $\vec{\omega}_i$ have zero first and second components, we can represent the angular velocities using a single scalar quantity.
    Furthermore we notice that the same holds for the update of the angular velocities in \eqref{eq:binary_rule} since we see from \eqref{eq:angular_velocity_update_2d} that $\mathbb{I}_i^{\dagger}(r_i\times \vec{n})$ has zero first and second component.

	Thus, in the two-dimensional setting, for fixed $\vec{v}_1, \vec{v}_2, \omega_1, \omega_2$ we know that the system is confined to a manifold of $\mathbb{R}^6$ defined as
	\begin{equation}
		\mathfrak{M} = \Big\{
			(\vec{v}_1^\prime, \vec{v}_2^\prime, \omega_1^\prime, \omega_2^\prime) \in \mathbb{R}^6 \; \text{such that \eqref{eq:conservation_linear}, \eqref{eq:conservation_angular} and \eqref{eq:conservation_energy} are conserved}
			\Big\}.
	\end{equation}
	Notice that this is a two-dimensional submanifold of $\mathbb{R}^6$. This manifold can be parametrized by $(\ell,\tau)\in\mathbb{R}_{\geq 0}\times[0,2\pi)$.
	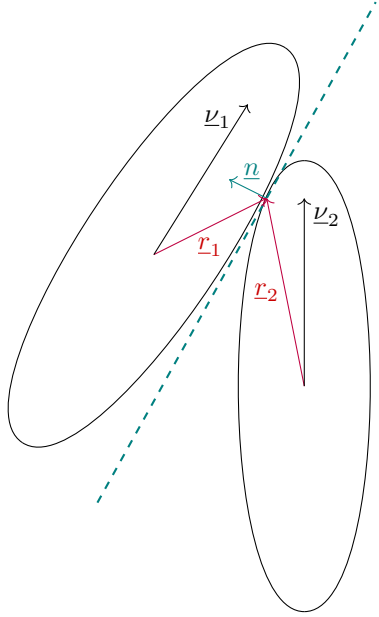
\begin{figure}
		\centering
		\scalebox{1}{\input{Figures/collision_calamitic.tikz}}
		\caption{Collision of calamitic molecules in two dimensions.}
	\end{figure}
\end{exB}
\begin{exC}[Head-tail symmetric molecules in two dimensions: Collision]
	\label{ex:collision_head_tail}
	The collision response model presented in \cref{ex:response2D} can also be used to describe the collision of head-tail symmetric molecules in two dimensions.
	However, to describe a collision event, as well as specifying $\vec{n}$, $\vec{r}_{1,2}$, we must also specify whether the head-tail symmetric molecules are co-rotating or counter-rotating. This is done by specifying $q \in \{0, 1\}$ and setting the angular velocity of the first molecule to be $(-1)^{q}\omega_1$.
\end{exC}

These examples motivate the following proposition describing in general the manifold to which the post-collisional variables are confined.
\begin{proposition}
If the Lagrangian is invariant under the group action $\mathcal{A}$, then the post-collision variables are constrained to a manifold $\mathfrak{M}$ of dimension $d + \iota - 1$, where we recall $\iota = \dim{\mathcal{M}}$. The manifold is given by
\begin{equation}
\mathfrak{M} = \Big\{ (\vec{v}_1^\prime, \vec{v}_2^\prime, \varsigma_1^\prime, \varsigma_2^\prime) \in \mathbb{R}^d \times \mathbb{R}^d \times T_{\nu_1} \mathcal{M} \times T_{\nu_2} \mathcal{M} : \text{\eqref{eq:conservation_linear}, \eqref{eq:conservation_angular}, \eqref{eq:conservation_energy} are conserved} \Big\}.
\end{equation}
\end{proposition}
\begin{proof}
	The proof follows from the fact the velocities $\vec{v}_1^\prime, \vec{v}_2^\prime$ and the conjugate momenta $\vec{\varsigma}_1^\prime, \vec{\varsigma}_2^\prime$ live in $\mathbb{R}^d \times \mathbb{R}^d \times T_{\nu_1} \mathcal{M} \times T_{\nu_2} \mathcal{M}$ and since $\dim(T_{\nu_i}\mathcal{M}) = \iota$ we have that the dimension of the manifold $\mathfrak{M}$ at most is $d + d + \iota + \iota = 2d + 2\iota$.
	The conservation law of linear momentum \eqref{eq:conservation_linear} gives us $d$ constraints, reducing the dimension of the manifold to $d + \iota + \iota = d + 2\iota$.
	The conservation law of angular momentum \eqref{eq:conservation_angular} gives us $\iota$ additional constraints, reducing the dimension of the manifold to $d + \iota$.
	Finally, the conservation law of energy \eqref{eq:conservation_energy} gives us one additional constraint, reducing the dimension of the manifold to $d + \iota - 1$.
	Thus, the dimension of the manifold $\mathfrak{M}$ is $d + \iota - 1$.
\end{proof}
\begin{table}[h]
    \centering
    \caption{Dimension of the manifold $\mathfrak{M}$ to which the post-collisional variables are confined, for the examples presented in this section. The dimension $d$ is the dimension of the space in which the particles move, and $\iota$ is the dimension of the order parameter manifold $\mathcal{M}$.}
    \label{tab:horizontal}
    \begin{tabular}{l|c|c|c|c|c}
        \toprule
        Example & \ref{ex:spherical} & \ref{ex:bubbles_collision} & \ref{ex:collision_response} & \ref{ex:response2D} & \ref{ex:collision_head_tail} \\
        \midrule
        $d$ & 3 & 3 & 3 & 2 & 2 \\
        $\iota$ & 0 & 1 & 2 & 1 & 2 \\
        $d+\iota - 1$ & 2 & 3 & 4 & 2 & 3 \\
        \bottomrule
    \end{tabular}
\end{table}
\section{The BBGKY hierarchy for ordered fluids}
In this section we derive a Bogoliubov--Born--Green--Kirkwood--Yvon (BBGKY) hierarchy to study a system of $N$ pairwise-interacting constituents endowed with an order from the point of view of statistical mechanics.
We will then derive a Vlasov--Boltzmann type kinetic equation for the one particle distribution function, under the hypothesis that the interaction between the ordering of the constituents is of a weak nature, and that no external forces act on the system.
In particular, the Vlasov--Boltzmann equation will be of the form
\begin{equation} 
	\label{eq:boltzmann}
	\frac{\partial f}{\partial t} + m^{-1}\vec{p} \cdot \nabla_{\pt{x}}\, f + \mat{B}(\nu)^{-1}\vec{\varsigma}\cdot \nabla_{\nu}\, f + \mathcal{V}\cdot \nabla_{\vec{\varsigma}} f
	= \mathcal{C}[f, f],
\end{equation}
where {$f = f(\pt{x},\nu,\vec{p}, \vec{\varsigma}, t)$} is the one particle distribution function, representing the probability of finding a particle in the phase space point $\left(\pt{x}, \nu, \vec{p}, \vec{\varsigma}\right)$, $\mathcal{C}$ is the collision operator, representing the effect of binary interactions, and $\mathcal{V}$ is a Vlasov-type force acting only on the order parameters.

We will study the collision operator and its collision invariants for such kinetic equations. We will derive a Boltzmann inequality following the approach of Cercignani \& Lampis \cite{cercignaniLampis}.
This will guarantee that the solution of the kinetic equation will thermalize to a Maxwellian distribution.
	
\subsection{Related literature}
Kinetic equations of the form \eqref{eq:boltzmann} are a particular instance of the wider class of kinetic equations for molecules with internal degrees of freedom.
This class of equations includes Boltzmann-type equations for loaded spheres \cite{jeans,dahlerSatherI,dahlerSandlerII}, and Boltzmann-type equations for non-spherical molecules.
In \cite{taxman} Taxman proposes the use of the standard Boltzmann equation for spherical molecules, but with a modified collision operator, to account for the non-spherical nature of the molecules.
A first classical generalization of the Boltzmann equation for non-spherical molecules known as  Boltzmann--Curtiss equation is used in \cite{curtissI,curtissII,curtissIII,curtissIV, curtissV,curtissVI}. This treats convex symmetric top molecules (one whose moment of inertia has two equal principal moments).
In our work we generalize the Boltzmann--Curtiss equation to include treat generic molecules endowed with an order parameter manifold $\mathcal{M}$.
In particular, in \cref{ex:2D} and \cref{ex:3D} we discuss the case of calamitic molecules in two and three dimensions, which is the direct generalization of the Boltzmann--Curtiss equation to molecules that might not be convex symmetric top molecules.
An extension of the Boltzmann--Curtiss equation for non-spherical molecules that include chattering, i.e.~collisions between two rigid convex molecules involving two or more impulsive hits, is studied in \cite{coleEtAll,allenEtAll, mori}.
The Wang--Chang--Uhlenbeck--de~Boer equation \cite{deBoerEtAll} is a generalization of the Boltzmann equation for non-spherical molecules, where the rotational degrees of freedom are quantized, i.e.~different rotational degrees of freedom are represented as a gas mixture, and molecules with different but close orientations are considered to be of the same species.
Waldmann~\cite{waldmann} and Snider~\cite{snider} presented a generalization of the Wang--Chang--Uhlenbeck--de~Boer equation to include quantum mechanical effects. Inelastic effects including rotational degrees of freedom have also been considered in kinetic theory of granular gases, see for instance \cite{JR85,Goldhirsch2003,oxfordbook}. 

\subsection{The hierarchy}
\label{sec:BBGKY}
We construct a BBGKY hierarchy by adapting the derivation presented in prior work \cite[Appendix A]{farrellEtAll}, which was inspired by those presented in \cite{harris,kardar,huangKerson}.
We consider a system of $N$ constituents endowed with an order parameter manifold $\mathcal{M}$, interacting in pairs, and assumed to be indistinguishable.
To this aim, we introduce the Hamiltonian formalism associated to the Lagrangian $\mathcal{L}$ introduced in the previous section.
As usual, we introduce the conjugate momenta to the generalized coordinates, i.e.
\begin{equation}
	\vec{p_i} = \frac{\partial \mathcal{L}}{\partial \vec{v}_i} = m\vec{v}_i, \qquad \vec{\varsigma_i} \coloneqq \frac{\partial \mathcal{L}}{\partial \vec{\dot{\nu}_i}} = \mat{B}(\nu)\,\vec{\dot{\nu}_i}.
\end{equation}
We then introduce the Hamiltonian $\mathcal{H}$ of the full system of $N$ constituents, only interacting in pairs, as
\begin{equation}
	\mathcal{H} \coloneqq \sum_{i=1}^N \frac{1}{2m}\vec{p_i}\cdot \vec{p_i} + \frac{1}{2}\vec{\varsigma_i}\cdot \mat{B}(\nu)^{-1}\,\vec{\varsigma_i}+\sum_{1\leq i<j\leq N}\mathcal{W}(\abs{\pt{x}_i-\pt{x}_j}, \nu_i, \nu_j).
\end{equation}
Let {$\Gamma^*_i(t) = \left(\pt{x}_i(t), \nu_i(t), \vec{p}_i(t), \vec{\varsigma}_i(t)\right)$} denote the phase space point of constituent $i$ at time $t$, and let \rev{$f_N\left(\{\Gamma_i\}_{i=1}^N\right)\coloneqq \prod_{i=1}^{N}\delta\left(\Gamma_i - \Gamma_i^*(t)\right)$} denote the \textit{empirical measure} or \textit{joint N-particles distribution function}.
Let $f_s$ denote the marginal distributions of the joint N-particles distribution function, with respect to $\Gamma^{(s)} = \left(\Gamma_{s+1}, \dots, \Gamma_N\right)$.
Notice that as a consequence of the indistinguishability of the particles, the marginals $f_s$ are symmetric with respect to the exchange of the particles.

\begin{remark}
	Technically, the Hamiltonian $\mathcal{H}$ is the Legendre transform of the Lagrangian $\mathcal{L}$, which in this case is always well-defined since we assumed that $\mat{B}(\nu)$ is symmetric and positive definite for all $\nu\in \mathcal{M}$.
	In fact, the symmetry and positive definiteness of $\mat{B}(\nu)$ implies that the Hessian of the Lagrangian $\mathcal{L}$ with respect to $\vec{\dot{\nu}}$ is also symmetric and positive definite, which in turn implies that the Legendre transform is well-defined.
	We point out that the Legendre transform of $\mathcal{L}_{1,i}$ with respect to $\vec{\dot{\nu}}_i$ is exactly the co-energy associated with the ordering introduced by Capriz \cite{carpiz}.
\end{remark}
\begin{remark}
	\rev{
		Given the order parameter manifold $\mathcal{M}$, the tuple $(\nu, \vec{\varsigma})$ lives in the cotangent bundle of $\mathcal{M}$.
		The Hamilton description of the dynamics here adopted can thus be expressed using the canonical Poisson brackets \cite[Section 1.3]{marsden}.
		This choice comes both with advantages and disadvantages; a larger phase space has to be used, but any smooth manifold can be used as an order-parameter manifold, in particular not requiring the symplecticity of the manifold.
		An alternative and conceptually appealing perspective is to encode these symmetries directly into the geometric structure of the kinetic theory.
		Since the Lagrangian considered here is invariant under rotations, the generalized angular momentum may be interpreted as a momentum map associated with such symmetry.
		This observation naturally motivates the introduction of reduced Lie--Poisson structure obtained via symmetry reduction, in which generalized angular momentum conservation is built into the formulation at the most fundamental level \cite[Section 1.3]{marsden}.
		The use of the reduced Lie--Poisson brackets is typical in the so-called GENERIC framework \cite{GENERIC}. In particular, systems of rod-like molecules similar to the one presented in \cref{ex:2D} and \cref{ex:3D} have been studied in the GENERIC framework in \cite{grmela}.
	}
\end{remark}

We decompose the Hamiltonian into three terms: two containing respectively contributions from particles $1,\dots,s$ and $s+1,\dots, N$, and a third term containing mixed terms:
\begin{alignat}{3}
  \mathcal{H}_s &= \left(\sum_{i=1}^{s} \frac{\abs{\vec{p_i}}^2}{2m}+\frac{1}{2}\vec{\varsigma_i}\cdot\mat{B}(\nu)^{-1}\vec{\varsigma_i}\right) &&+\sum_{1\leq i < j \leq s}\!\!\!\!\mathcal{W}(\abs{\pt{x}_i-\pt{x}_j},\nu_i,\nu_j),\\
  \mathcal{H}_{N-s} &= \left(\sum_{i=s+1}^{N} \frac{\abs{\vec{p_i}}^2}{2m}+\frac{1}{2}\vec{\varsigma_i}\cdot\mat{B}(\nu)^{-1}\vec{\varsigma_i}\right) &&+\sum_{s+1\leq i < j \leq N}\!\!\!\!\!\!\!\mathcal{W}(\abs{\pt{x}_i-\pt{x}_j},\nu_i,\nu_j),\\
  \hat{\mathcal{Z}}_{s} &= && \sum_{i=1}^{s} \sum_{j=s+1}^{N} \mathcal{W}(\abs{\pt{x}_i-\pt{x}_j},\nu_i,\nu_j),
\end{alignat}
so that $\mathcal{H}=\mathcal{H}_s+\mathcal{H}_{N-s}+\hat{\mathcal{Z}}_s$.
Considering Liouville's equation for the joint N-particles distribution function $f_N$, i.e.
\begin{equation}
  \frac{\partial f_N}{\partial t} = -\{f_N,\mathcal{H}\},
\end{equation}
where $\{\cdot,\cdot\}$ are the Poisson brackets, we can integrate with respect to $\Gamma^{(s)}$ to obtain the integro-differential equation describing the time evolution of the marginal distribution function $f_s$, i.e.
\begin{equation}
  \frac{\partial f_s}{\partial t} 
  = \int \frac{\partial f_N}{\partial t}\, d\Gamma^{(s)}
  =-\int \left(\{f_N,\mathcal{H}_s\}+\{f_N,\mathcal{H}_{N-s}\}+\{f_N,\hat{\mathcal{Z}}_s\}\right)\,d\Gamma^{(s)}.
\end{equation}
We notice that, since $\mathcal{H}_s$ is independent of $\Gamma_{s+1},\dots,\Gamma_{n}$, we can bring the integral inside the Poisson brackets in that term to obtain
\begin{equation}
  \frac{\partial f_s}{\partial t} 
  = \int \frac{\partial f_N}{\partial t}\, d\Gamma^{(s)}
  = - \{f_s,\mathcal{H}_s\}
    - \int \{f_N,\mathcal{H}_{N-s}\} \,d\Gamma^{(s)}
    - \int\{f_N,\hat{\mathcal{Z}}_s\}\,d\Gamma^{(s)}.
    \label{eq:Liouville}
\end{equation}
We then notice that the second term on the right-hand side vanishes since it is an exact divergence, and we assume appropriate decay conditions at infinity:
\begin{align}
  \int\!\!\sum_{i=1}^{N}\!\!\Big(\nabla_{\pt{x}_i}f_N\!\cdot\!\nabla_{\vec{p}_i}\!\mathcal{H}_{N-s}\!\!&+\!\nabla_{\nu_{i}}f_N\!\cdot\!\nabla_{\vec{\varsigma}_i}\!\mathcal{H}_{N-s}\!\!-\!\nabla_{\vec{p}_i}f_N\!\cdot\!\nabla_{\pt{x}_i}\!\mathcal{H}_{N-s}\!\!-\!\nabla_{\vec{\varsigma}_i}f_N\!\cdot\!\nabla_{\nu_i}\!\mathcal{H}_{N-s}\!\Big)d\Gamma^{(s)}\\
  =&\int\!\!\sum_{i=s+1}^{N}\Big(\nabla_{\pt{x}_i}f_N\cdot \frac{\vec{p}_i}{m}+\nabla_{\nu_i}f_N\cdot\mat{B}(\nu_i)^{-1}\vec{\varsigma}_i\Big)\,d\Gamma^{(s)}\nonumber\\
  &-\int\!\!\sum_{i=s+1}^{N}\!\!\nabla_{\vec{p}_i}f_N\cdot\sum_{j=i+1}^{N}\!\!\nabla_{\pt{x}_i}\cdot \mathcal{W}(\abs{\pt{x}_i-\pt{x}_j},\nu_i,\nu_j)d\Gamma^{(s)}\nonumber\\
  &-\int\!\!\sum_{i=s+1}^{N}\nabla_{\vec{\varsigma}_i}f_N\cdot \sum_{j=i+1}^{N}\nabla_{\nu_i}\mathcal{W}(\abs{\pt{x}_i-\pt{x}_j},\nu_i,\nu_j)\,d\Gamma^{(s)}\nonumber\\
  =&\int\!\!\sum_{i=s+1}^{N}\Bigg(\nabla_{\pt{x}_i}\cdot\Big(f_N\frac{\vec{p}_i}{m}\Big)+\nabla_{\nu_i}\cdot\Big(f_N\mat{B}(\nu_i)^{-1}\vec{\varsigma}_i\Big)\Bigg)\,d\Gamma^{(s)}\nonumber\\
  &-\int\!\!\sum_{i=s+1}^{N}\nabla_{\vec{p}_i}\cdot\Big(f_N\sum_{j=i+1}^{N}\nabla_{\pt{x}_i}\mathcal{W}(\abs{\pt{x}_i-\pt{x}_j},\nu_i,\nu_j)\Big)\,d\Gamma^{(s)}\nonumber\\
  &-\int\!\!\sum_{i=s+1}^{N}\nabla_{\vec{\varsigma}_i}\cdot\Big(f_N\sum_{j=i+1}^{N}\nabla_{\nu_i}\mathcal{W}(\abs{\pt{x}_i-\pt{x}_j},\nu_i,\nu_j)\Big)\,d\Gamma^{(s)} \\
  &=0 \text{ by decay conditions at infinity.}\nonumber
\end{align}
This argument relied on the calculation
\begin{equation}
\nabla_{\nu_i}f_N\cdot\mat{B}(\nu_i)^{-1}\vec{\varsigma}_i = \nabla_{\nu_i}\cdot\Big(f_N\mat{B}(\nu_i)^{-1}\vec{\varsigma}_i\Big),
\end{equation}
which is not immediately obvious, since $\mat{B}$ depends on $\nu_i$.
\rev{The subtlety is that while $\mat{B}$ depends on $\nu_i$, the term $\mat{B}(\nu_i)^{-1}\vec{\varsigma}_i$ is equivalent to $\vec{\dot{\nu}}_i$, which being a generalised velocity, is independent of $\nu_i$.
Matters are greatly simplified if the group action $\mathcal{A}$ is transitive, since in this case $\mat{B}(\nu_i)$ is independent of $\nu_i$.}

It remains to study the last term in \eqref{eq:Liouville}, which expands to
\begin{align}
  &\int\!\!\sum_{i=1}^{N}\!\Big(\!\nabla_{\pt{x}_i}f_N\!\cdot\!\nabla_{\vec{p}_i}\hat{\mathcal{Z}}_{s}
  \!+\!\nabla_{\nu_i}f_N\!\cdot\!\nabla_{\vec{\varsigma}_i}\hat{\mathcal{Z}}_{s}
  \!-\!\nabla_{\vec{p}_i}f_N\!\cdot\!\nabla_{\pt{x}_i}\hat{\mathcal{Z}}_{s}
  \!-\!\nabla_{\vec{\varsigma}_i}f_N\!\cdot\!\nabla_{\nu_i}\hat{\mathcal{Z}}_{s}\!\Big)\!\,d\Gamma^{(s)}\\
  = -&\int\sum_{i=1}^{s}\nabla_{\vec{p}_i}f_N\cdot\sum_{j=s+1}^{N}\nabla_{\pt{x}_i}\mathcal{W}(\abs{\pt{x}_i-\pt{x}_j},\nu_i,\nu_j)\,d\Gamma^{(s)}\nonumber\\
  -&\int\sum_{j=s+1}^{N}\nabla_{\vec{p}_j}f_N\cdot\sum_{i=1}^{s}\nabla_{\pt{x}_j}\mathcal{W}(\abs{\pt{x}_i-\pt{x}_j},\nu_i,\nu_j)\,d\Gamma^{(s)}\nonumber\\
  -&\int\sum_{i=1}^{s}\nabla_{\vec{\varsigma}_i}f_N\cdot\sum_{j=s+1}^{N}\nabla_{\nu_i}\mathcal{W}(\abs{\pt{x}_i-\pt{x}_j},\nu_i,\nu_j)\,d\Gamma^{(s)}\nonumber\\
  -&\int\sum_{j=s+1}^{N}\nabla_{\vec{\varsigma}_j}f_N\cdot\sum_{i=1}^{s}\nabla_{\nu_j}\mathcal{W}(\abs{\pt{x}_i-\pt{x}_j},\nu_i,\nu_j)\,d\Gamma^{(s)}.\nonumber
\end{align}
We notice that the second and last terms in the previous equation vanish since they are exact divergences.
Since all the particles are indistinguishable, we can simplify the sum over $j$ as
\begin{align}
  -\int\!\{f_N,\hat{\mathcal{Z}}_{s}\} d\Gamma^{(s)}=& 
  \int\sum_{i=1}^{s}\nabla_{\vec{p}_i} f_N \cdot \sum_{j=s+1}^{N} \nabla_{\pt{x}_i} \mathcal{W}(\abs{\pt{x}_i-\pt{x}_j},\nu_i,\nu_j) \, d\Gamma^{(s)}\\
  &+\int\sum_{i=1}^{s} \nabla_{\vec{\varsigma}_i} f_N \cdot \sum_{j=s+1}^{N} \nabla_{\nu_i} \mathcal{W}(\abs{\pt{x}_i-\pt{x}_j},\nu_i,\nu_j) \, d\Gamma^{(s)}\nonumber\\
  =&\quad\,(N-s)\int\sum_{i=1}^{s} \nabla_{\vec{p}_i} f_N \cdot \nabla_{\pt{x}_i} \mathcal{W}(\abs{\pt{x}_i-\pt{x}_{s+1}},\nu_i,\nu_{s+1}) \, d\Gamma^{(s)}\nonumber\\
  &+(N-s)\int\sum_{i=1}^{s} \nabla_{\vec{\varsigma}_i} f_N \cdot \nabla_{\nu_i} \mathcal{W}(\abs{\pt{x}_i-\pt{x}_{s+1}},\nu_i,\nu_{s+1}) \, d\Gamma^{(s)}\nonumber\\
  =&\quad\,(N-s)\int\sum_{i=1}^{s} \nabla_{\vec{p}_i} f_{s+1} \cdot \nabla_{\pt{x}_i} \mathcal{W}(\abs{\pt{x}_i-\pt{x}_{s+1}},\nu_i,\nu_{s+1}) \, d\Gamma_{s+1}\nonumber\\
  &+(N-s)\int\sum_{i=1}^{s} \nabla_{\vec{\varsigma}_i} f_{s+1} \cdot \nabla_{\nu_i} \mathcal{W}(\abs{\pt{x}_i-\pt{x}_{s+1}},\nu_i,\nu_{s+1}) \, d\Gamma_{s+1},\nonumber
\end{align}
where the last identity has been obtained by integrating $f_N$ over $\Gamma^{(s)}$ to obtain the marginal $f_{s+1}$.
The $(N-s)$ factor comes from the fact that we are summing over $j=s+1,\dots,N$ and this amounts to $N-s$ identical contributions, due to the permutation symmetry of $f_N$.

We have now obtained the BBGKY hierarchy for the Boltzmann--type equations here considered, i.e.~
\begin{align}
  \frac{\partial f_s}{\partial t} + \{f_s,\mathcal{H}_s\} =& \quad\,(N-s)\int\sum_{i=1}^{s}\nabla_{\vec{p}_i}f_{s+1}\cdot\nabla_{\pt{x}_i}\mathcal{W}(\abs{\pt{x}_i-\pt{x}_{s+1}},\nu_i,\nu_{s+1})\,d\Gamma_{s+1}\\
  &+(N-s)\int\sum_{i=1}^{s}\nabla_{\vec{\varsigma}_i}f_{s+1}\cdot\nabla_{\nu_i}\mathcal{W}(\abs{\pt{x}_i-\pt{x}_{s+1}},\nu_i,\nu_{s+1})\,d\Gamma_{s+1}.\nonumber
\end{align}
We abuse notation and denote by $f_s$ the normalized joint $N$-particles distribution. This yields the following expression:
\begin{align}
  \frac{\partial f_s}{\partial t} + \{f_s,\mathcal{H}_s\} =&\quad\, \int\sum_{i=1}^{s}\nabla_{\vec{p}_i} f_{s+1}\cdot\nabla_{\pt{x}_i}\mathcal{W}(\abs{\pt{x}_i-\pt{x}_{s+1}},\nu_i,\nu_{s+1})\,d\Gamma_{s+1}\\
  &+\int\sum_{i=1}^{s}\nabla_{\vec{\varsigma}_i}f_{s+1}\cdot\nabla_{\nu_i}\mathcal{W}(\abs{\pt{x}_i-\pt{x}_{s+1}},\nu_i,\nu_{s+1})\,d\Gamma_{s+1}.\nonumber
\end{align}

We now focus our attention on the first two terms of the BBGKY hierarchy, the terms associated to the one and two particles distribution functions:
\begin{align}
  \label{eq:BBGKY1_pre}
  \frac{\partial f_1}{\partial t}\!+\!\frac{\vec{p}_1}{m}\!\cdot\!\nabla_{\pt{x}_1}\!f_1\!+\!\mat{B}(\nu_1)^{-1}\vec{\varsigma}_1\!\cdot\!\nabla_{\nu_1}\!f_1\!=&\quad\,\!\!\int\!\! \nabla_{\pt{x}_1}\mathcal{W}(\abs{\pt{x}_1-\pt{x}_2},\nu_1,\nu_2)\!\cdot\!\nabla_{\vec{p}_1}\!f_2\!\,d\Gamma_2\\
  &+\!\!\int\!\!\nabla_{\nu_1}\mathcal{W}(\abs{\pt{x}_1-\pt{x}_2},\nu_1,\nu_2)\cdot\nabla_{\vec{\varsigma}_1}f_2\,d\Gamma_2,\nonumber
\end{align}
\begin{align}
  \label{eq:BBGKY2}
  \frac{\partial f_2}{\partial t} + \frac{\vec{p}_1}{m}\cdot \nabla_{\pt{x}_1}f_2&+\mat{B}(\nu_1)^{-1}\vec{\varsigma}_1\cdot\nabla_{\nu_1} f_2 + \frac{\vec{p}_2}{m}\cdot\nabla_{\pt{x}_2} f_2+\mat{B}(\nu_2)^{-1}\vec{\varsigma}_2\cdot\nabla_{\nu_2} f_2\\
  &- \nabla_{\pt{x}_1}\mathcal{W}(\abs{\pt{x}_1-\pt{x}_2},\nu_1,\nu_2)\cdot \Big(\nabla_{\vec{p}_1}f_2-\nabla_{\vec{p}_2}f_2\Big) \nonumber\\
  &- \nabla_{\nu_1}\mathcal{W}(\abs{\pt{x}_1-\pt{x}_2},\nu_1,\nu_2)\cdot\Big(\nabla_{\vec{\varsigma}_1} f_2-\nabla_{\vec{\varsigma}_2} f_2\Big)=0\nonumber.
\end{align}
The right-hand side of the last equation vanishes because we have assumed $f_3\equiv 0$ since we are working under the assumption that we only have binary collisions.

The right-hand side of \eqref{eq:BBGKY1_pre} will be simplified below to yield the familiar Boltzmann collision term. It is more natural to write this collision term with relative velocities. We therefore add a vanishing term of the form $\nabla_{\pt{x}_2} \mathcal{W}(\abs{\pt{x}_1-\pt{x}_2},\nu_1,\nu_2)\cdot\nabla_{\vec{p}_2}f_2$ to \eqref{eq:BBGKY1_pre} to yield:
\begin{align}
  \label{eq:BBGKY1}
  \frac{\partial f_1}{\partial t}\!+\!\frac{\vec{p}_1}{m}\!\cdot\!\nabla_{\pt{x}_1}\!f_1\!+\!\mat{B}(\nu_1)^{-1}\vec{\varsigma}_1\!\cdot\!\nabla_{\nu_1}\!f_1\!=&\quad\,\!\!\int\!\! \nabla_{\pt{x}_1}\mathcal{W}(\abs{\pt{x}_1-\pt{x}_2},\nu_1,\nu_2)\!\cdot\!\Big(\nabla_{\vec{p}_1}\!f_2\!-\!\nabla_{\vec{p}_2}\!f_2\Big)\,d\Gamma_2\\
  &+\!\!\int\!\!\nabla_{\nu_1}\mathcal{W}(\abs{\pt{x}_1-\pt{x}_2},\nu_1,\nu_2)\cdot\nabla_{\vec{\varsigma}_1}f_2\,d\Gamma_2.\nonumber
\end{align}

At this point it is convenient to observe that there at least two timescales in the previous equations: the timescale associated to the macroscopic length scale and the timescale associated to the microscopic length scale.
For example, the left-hand side of \eqref{eq:BBGKY1} acts on the slow timescale while the right-hand side acts on the fast timescale, via the potential $\mathcal{W}$.
To highlight the same timescale separation in \eqref{eq:BBGKY2} we introduce fast and slowly varying coordinates, respectively\\
\begin{equation}
	\label{eq:fast_slow_coordinates}  
	\hat{\pt{x}} = \pt{x}_2-\pt{x}_1, \qquad \hat{\pt{X}} = \frac{1}{2}\left(\pt{x}_2+\pt{x}_1\right).
\end{equation}
We can compute $\nabla_{\hat{\pt{X}}} f_2$ and $\nabla_{\hat{\pt{x}}}f_2$ using the chain rule, i.e.
\begin{align}
	\nabla_{\hat{\pt{X}}} f_2\!&=\! \nabla_{\pt{x}_1}f_2 \cdot \frac{\partial \pt{x}_1}{\partial \hat{\pt{X}}} + \nabla_{\pt{x}_2} f_2\cdot \frac{\partial \pt{x}_2}{\partial \hat{\pt{X}}} \!=\! 2\left(\nabla_{\pt{x}_1} f_2\!+\! \nabla_{\pt{x}_2} f_2\right),\\
	\nabla_{\hat{\pt{x}}} f_2\!&=\! \nabla_{\pt{x}_1} f_2\cdot \frac{\partial \pt{x}_1}{\partial \hat{\pt{x}}} + \nabla_{\pt{x}_2}f_2 \cdot \frac{\partial \pt{x}_2}{\partial \hat{\pt{x}}} \!=\! - \nabla_{\pt{x}_1} f_2 \!+\! \nabla_{\pt{x}_2} f_2,
\end{align}
where we have denoted $\displaystyle\frac{\partial \pt{x}_i}{\partial \hat{\pt{X}}}$ and $\displaystyle\frac{\partial \pt{x}_i}{\partial \hat{\pt{x}}}$ the Jacobian matrices of the transformations \eqref{eq:fast_slow_coordinates}.
Using the previous identities we can rewrite \eqref{eq:BBGKY2} as
\begin{align}
  \frac{\partial f_2}{\partial t}\!+ \frac{1}{2}\frac{\vec{p}_2\!+\!\vec{p}_1}{m}\cdot\nabla_{\hat{\pt{X}}}f_2
  &+\mat{B}(\nu_1)^{-1}\vec{\varsigma}_1\cdot \nabla_{\nu_1} f_2+\mat{B}(\nu_2)^{-1}\vec{\varsigma}_2\cdot  \nabla_{\nu_2}f_2\;
  +\boxed{\!\frac{\vec{p}_2\!-\!\vec{p}_1}{m}\cdot\nabla_{\hat{\pt{x}}} f_2}&\label{eq:BBGKY2Separated}\\
  &-\;\boxed{\nabla_{\pt{x}_1} \mathcal{W}(\abs{\pt{x}_1-\pt{x}_2},\nu_1,\nu_2)\cdot\Big(\nabla_{\vec{p}_1}f_2-\nabla_{\vec{p}_2}f_2\Big)}&\nonumber\\
  &-\;\nabla_{\nu_1} \mathcal{W}(\abs{\pt{x}_1-\pt{x}_2},\nu_1,\nu_2)\cdot\Big(\nabla_{\vec{\varsigma}_1} f_2 - \nabla_{\vec{\varsigma}_2} f_2\Big)=0\nonumber
\end{align}
where we have boxed all the terms acting on the fast timescale.
\begin{remark}
	\rev{
	The standard way of separating the fast and slow timescales is via the following asymptotic expansion.
	Under the hypothesis that the two particle distribution function $f_2$ is non-vanishing only on the collisional length scale, we can introduce a small parameter $\varepsilon$ such that 
	\begin{equation}
		f_2(\pt{x}_1,\nu_1,\vec{p}_1,\vec{\varsigma}_1,\pt{x}_2,\nu_2,\vec{p}_2,\vec{\varsigma}_2,t) \ne 0 \Leftrightarrow \abs{\hat{\pt{x}}} = \pt{x}_2-\pt{x}_1 \geq \varepsilon.
	\end{equation}
	Thus the gradient of $f_2$ with respect to the fast variable $\hat{\pt{x}}$ does not vanish if and only if $\hat{\pt{x}}= \mathcal{O}(\varepsilon)$ and $\hat{\pt{X}} = \mathcal{O}(\frac{1}{2}(\pt{x}_1+\pt{x}_2))$.
	Letting $\varepsilon\to 0$ we obtain the separation of timescales highlighted in \eqref{eq:BBGKY2Separated}, more details can be found in \cite{kardar}.
	}
\end{remark}
\subsection{Weak-order interaction}
\label{sec:weak_order_interactions}
We might be tempted to introduce fast and slowly varying coordinates also for the order parameters $\nu_1$ and $\nu_2$.
Indeed, this is done in \cite{farrellEtAll} to derive the Boltzmann--Curtiss equation for non-spherical molecules from the BBGKY hierarchy.
Unfortunately, this is not always possible. For example, if the manifold $\mathcal{M}$ is not endowed with the structure of a vector space, as often happens, then a definition of fast and slowly varying coordinates similar to \eqref{eq:fast_slow_coordinates} would not be well-defined because the sum of two points in $\mathcal{M}$ might not belong to $\mathcal{M}$.
In \cite{farrellEtAll} to resolve this issue, the manifold $\mathcal{M}$ was implicitly embedded in a vector space, and fast and slowly varying coordinates were then defined in the embedded space.
This approach yields the Boltzmann--Curtiss equation,
but is suboptimal because we are unnecessarily increasing the dimension of the phase space: in the worst case this requires doubling the dimension of the manifold by the Whitney embedding theorem. This worst case arises in examples of practical interest, such as the head-tail symmetric calamitic molecules of \cref{ex:head-tail}.

Ideally, we would like to keep the low dimensionality of the manifold, and still be able to define fast and slowly varying coordinates.
The most straightforward approach is to observe that if the fluid is ordered then  we can assume $\nu_1$ and $\nu_2$ are slowly varying variables.
Thus, we can regard the slow timescale as stationary with respect to the fast timescale to obtain the following identity:
\begin{equation}
	\frac{\vec{p}_2-\vec{p}_1}{m}\cdot \nabla_{\hat{\pt{x}}} f_2 = \nabla_{\pt{x}_1}\mathcal{W}\left(\abs{\pt{x}_1-\pt{x}_2},\nu_1,\nu_2\right)\cdot\left(\nabla_{\vec{p}_1}f_2-\nabla_{\vec{p}_2}f_2\right).
\end{equation}
Substituting the previous identity in \eqref{eq:BBGKY1}, we obtain
\begin{align}
	\label{eq:BBGKY1_weak}
	\frac{\partial f_1}{\partial t} + \frac{\vec{p}_1}{m}\cdot\nabla_{\pt{x}_1}f_1+\mat{B}(\nu_1)^{-1}\vec{\varsigma}_1\cdot\nabla_{\nu_1} f_1 =&\quad\,\int \frac{\vec{p}_2-\vec{p}_1}{m}\cdot \nabla_{\pt{x}} f_2\,d\Gamma_2\\
  	&+\int\nabla_{\nu_1}\mathcal{W}(\abs{\pt{x}_1-\pt{x}_2},\nu_1,\nu_2)\cdot\nabla_{\vec{\varsigma}_1} f_2\,d\Gamma_2\nonumber.
\end{align}
The right-hand side of the previous equation still depends on the two-particle distribution function $f_2$, and thus is not closed.
Unfortunately, without more specific knowledge of the exact interactions taking place in the system, it is extremely difficult to close \eqref{eq:BBGKY1_weak}.
Once again the fact that the order parameters are subject to partial ordering comes to our rescue.
We will assume that the order parameters obey a mean-field description, which corresponds to the assumption that the conjugate moments of the order parameters of any two constituents are independent.
\begin{remark}
	The hypothesis that the conjugate moments of the order parameters obey a mean-field description is crucial to obtain an equation of Vlasov--Boltzmann type.
    This hypothesis is reasonable in many ordered fluids.
	For example liquid crystals phases are characterised by the behavior of the correlation functions of the order parameters only, not of the correlation functions of conjugate moments of the order parameters as well \cite{degennes}.
\end{remark}
\begin{definition}
	We say that a system is governed by weak-order interactions if the derivatives of the two-particle distribution function factorise as
	\begin{align} \label{eq:weak-order-interactions}
		\nabla_{\vec{\varsigma_i}} f_2(\Gamma_i,\Gamma_j, t) &= f_1(\Gamma_j,t)\nabla_{\vec{\varsigma_i}} f_1(\Gamma_i, t),
	\end{align}
	for $i\neq j$ and $i,j=1,2$.
\end{definition}
\begin{remark}
	A common modelling assumption is to assume that also the positions, the momenta, and the order parameters of any two constituents are independent, which is equivalent to assuming that the two-particle distribution function factorizes as
	\begin{equation} \label{eq:weak-interaction}
		f_2(\Gamma_1,\Gamma_2, t) = f_1(\Gamma_1, t)f_1(\Gamma_2, t).
	\end{equation}
	This would imply that all interaction are of weak nature: \eqref{eq:weak-interaction} implies \eqref{eq:weak-order-interactions}, but not vice versa.
	Such weak interactions are undesirable in the context of the present work because they lead to the derivation of Vlasov-type kinetic equations \cite{Klimontovich,lashmoreDavies}.
	Vlasov-type kinetic equations are of reversible nature, as we know from the fact that they are compatible with Loschmidt's paradox.
	\rev{The common resolution of Loschmidt's paradox is to observe that \eqref{eq:weak-interaction} only holds in the pre-collisional phase, while in the post-collisional phase the two-particle distribution function does not factorise because of the correlations generated by the collision.
	Thus to split the collision operator into a gain and a loss term, we need to transform the two-particle distribution function expressed in terms of the post-collisional variables into the two-particle distribution function expressed in terms of the pre-collisional variables.
	In the current section we will follow a similar approach \cite{cercignaniLampis}, but we will always assume that the two-particle distribution function factorises as in \eqref{eq:weak-order-interactions} both in the pre-collisional and post-collisional phase.
	} 
\end{remark}

\begin{remark}
\label{rmk:mesoscale}
	We assumed that $\nu_1$ and $\nu_2$ are varying on a different scale with respect to the fast coordinates. A natural question is whether this assumption is reasonable in physical systems of interest.
	For example let us consider the case of a nematic liquid crystals, the timescales associated with molecular collisions and the director reorientation are key to understanding the validity of the previous assumption. 

	The collision time, $ \tau_{\text{col}} $, is the characteristic time between molecular collisions. It is given by the expression
	\begin{equation}
	\tau_{\text{col}} \sim \frac{\ell_{\text{mfp}}}{v}
	\end{equation}
	where $ \ell_{\text{mfp}} $ is the mean free path and $ v $ is the thermal velocity of the molecules. For a typical liquid crystal, we take $ \ell_{\text{mfp}} \sim 10^{-7} \, \text{m} $ and $ v \sim 100 \, \text{m/s}$ (for this experimental data, see \cite{degennes}), yielding a collision time of
	\begin{equation}
		\tau_{\text{col}} \sim \frac{10^{-7} \, \text{m}}{10^{2} \, \text{m/s}} = 10^{-9} \, \text{seconds}.
	\end{equation}

	The director reorientation time, $ \tau_{\text{director}} $, represents the timescale for the collective reorientation of the director in response to external perturbations.
	The director reorientation time can be approximated by
	\begin{equation}
	\tau_{\text{director}} \sim \frac{K}{\mu},
	\end{equation}
	where $ K $ is the Frank elastic constant and $ \mu $ is the viscosity. Using typical values $ K \sim 10^{-12} \, \text{N} \cdot \text{m} $ and $ \mu \sim 10^{-3} \, \text{Pa} \cdot \text{s}$ (for this experimental data, see \cite{degennes}), we obtain
	\begin{equation}
	\tau_{\text{director}} \sim \frac{10^{-12}}{10^{-3}} = 10^{-9} \, \text{seconds}.
	\end{equation}

	Yet, in nematic liquid crystals, the effective viscosity can be significantly higher than the bulk viscosity due to the molecular alignment inherent in the nematic phase. In this phase, the molecules are anisotropic and tend to align along a preferred direction. Unlike isotropic liquids, where molecules can move freely in all directions, the molecular movement in nematic liquid crystals is constrained by the alignment, which increases resistance to flow. As a result, the effective viscosity in the direction of the director's reorientation is higher.
	The director reorientation time, $ \tau_{\text{director}} $, can be expressed in terms of the effective viscosity $ \mu_{\text{eff}} $ as
	\begin{equation}
	\tau_{\text{director}} \sim \frac{K}{\mu_{\text{eff}}},
	\end{equation}
	where $ K $ is the Frank elastic constant and $ \mu_{\text{eff}} $ is the effective viscosity. Assuming a typical value for the Frank elastic constant $ K \sim 10^{-12} \, \text{N} \cdot \text{m} $ and considering a higher effective viscosity $\mu_{\text{eff}} \sim 10^{-2} \, \text{Pa} \cdot \text{s}$ due to the molecular alignment (for this experimental data, see \cite{OzakiEtAll}), we can compute the director reorientation time as
	\begin{equation}
	\tau_{\text{director}} \sim \frac{10^{-12}}{10^{-2}} = 10^{-10} \, \text{seconds}.
	\end{equation}

	This suggests that the director reorientation time can be an order of magnitude smaller than the collision time, which supports the assumption that $\nu_1$ and $\nu_2$ are varying on a different scale with respect to the fast coordinates.
	In particular, this observation suggest that the nematic ordering correlates molecules on larger length scales than the collisional length scale. The assumption of long-range correlations is typical of Vlasov-type kinetic equations, and justifies the weak-order interaction hypothesis.
\end{remark}
 
Under the assumption of weak-order interactions, we can rewrite \eqref{eq:BBGKY1_weak} as 
\begin{align}
	\frac{\partial f_1}{\partial t} + \frac{\vec{p}_1}{m}\cdot\nabla_{\pt{x}_1}f_1&+\mat{B}(\nu_1)^{-1}\vec{\varsigma}_1\cdot\nabla_{\nu_1}f_1 = \int \frac{\vec{p}_2-\vec{p}_1}{m}\cdot \nabla_{\hat{\pt{x}}}f_2\,d\Gamma_2\\
  	&+\int\nabla_{\nu_1}\mathcal{W}(\abs{\pt{x}_1-\pt{x}_2},\nu_1,\nu_2)f_1(\Gamma_2,t)\cdot\nabla_{\vec{\varsigma_1}}f_1(\Gamma_1,t)\,d\Gamma_2\nonumber.
\end{align}
Introducing the mean-field force,
\begin{equation}
	\label{eq:vlasov_Force}
	\mathcal{V}(\pt{x}_1,\nu_1) = -\int \nabla_{\nu_1}\mathcal{W}(\abs{\pt{x}_1-\pt{x}_2},\nu_1,\nu_2)f_1(\Gamma_2,t)\,d\Gamma_2,
\end{equation}
we can rewrite the previous equation as
\begin{equation}
	\label{eq:BBGKY1_weak_order}
	\frac{\partial f_1}{\partial t} + \frac{\vec{p}_1}{m}\cdot\nabla_{\pt{x}_1}f_1+\mat{B}(\nu_1)^{-1}\vec{\varsigma}_1\cdot\nabla_{\nu_1}f_1 + \mathcal{V}(\pt{x}_1,\nu_1)\cdot\nabla_{\vec{\varsigma}_1}f_1 = \int \frac{\vec{p}_2-\vec{p}_1}{m}\cdot \nabla_{\hat{\pt{x}}}f_2\,d\Gamma_2. 
\end{equation}

Following the classical derivation of the Boltzmann collision operator \cite{harris,cercignani} we assume that the interactions occurring between the fluid constituents are spatially of short range nature. We consider a ball of radius $2R$ centered at $\pt{x}_i$, denoted $B_{2R}(\pt{x}_i)$, where $R>0$ is greater than the interaction diameter, and smaller than the mean-free path. We then assume that any two fluid constituents whose distance is smaller than $2R$ interact with each other in a binary fashion.
Under this hypothesis we can reduce the spatial integral appearing on the right-hand side to an integral over the ball of radius $2R$ and via Gauss' theorem we can rewrite the integral as an integral over the boundary of the ball, 
\begin{equation}
	\int \frac{\vec{p}_2-\vec{p}_1}{m}\cdot \nabla_{\hat{\pt{x}}}f_2\,d\Gamma_2 = \frac{1}{m}\int d\vec{\varsigma}_2\, d\vec{p}_2\, d\nu_2 \,\int_{\partial B_{2R}(\pt{x}_1)} \left(\vec{p}\cdot \vec{n}\right) f_2(\Gamma_1,\Gamma_2,t)\, dA,
\end{equation}
where $\vec{n}$ is the outward normal to the boundary of the ball, $\vec{p} =  \vec{p}_2-\vec{p}_1$ is the relative linear momentum of the two fluid constituents and $dA$ is the surface area element over the boundary of the ball.
We can then split the integral over the boundary of the ball into two integrals, one over the portion of the sphere where $\vec{p}\cdot \vec{n}>0$ and one over the portion of the sphere where $\vec{p}\cdot \vec{n}<0$, respectively denoted as $\partial B_{2R}^+(\pt{x}_1)$ (post-collisional) and $\partial B_{2R}^-(\pt{x}_1)$ (pre-collisional):
\begin{align}
	\int \frac{\vec{p}_2-\vec{p}_1}{m}\cdot \nabla_{\hat{\pt{x}}}f_2\,d\Gamma_2 &= \frac{1}{m}\int d\vec{\varsigma}_2\, d\vec{p}_2\, d\nu_2 \,\int_{\partial B_{2R}^+(\pt{x}_1)} \abs{\vec{p}\cdot \vec{n}} f_2(\Gamma_1^\prime,\Gamma_2^\prime,t)\, dA \\
	&- \frac{1}{m}\int d\vec{\varsigma}_2\, d\vec{p}_2\, d\nu_2 \,\int_{\partial B_{2R}^-(\pt{x}_1)} \abs{\vec{p}\cdot \vec{n}} f_2(\Gamma_1,\Gamma_2,t)\, dA.\nonumber
\end{align}
Using the molecular chaos assumption for constituents about to interact, we rewrite the previous equation as
\begin{align}
	\int \frac{\vec{p}_2-\vec{p}_1}{m}\cdot \nabla_{\hat{\pt{x}}}f_2\,d\Gamma_2 &=\frac{1}{m} \int d\vec{\varsigma}_2\, d\vec{p}_2\, d\nu_2 \,\int_{\partial B_{2R}^+(\pt{x}_1)} \abs{\vec{p}\cdot \vec{n}} f_2(\Gamma_1^\prime,\Gamma_2^\prime,t)\, dA \\
	&-\frac{1}{m} \int d\vec{\varsigma}_2\, d\vec{p}_2\, d\nu_2 \,\int_{\partial B_{2R}^-(\pt{x}_2)} \abs{\vec{p}\cdot \vec{n}} f_1(\Gamma_1,t)f_1(\Gamma_2,t)\, dA.\nonumber
\end{align}
Transforming the area element $\abs{\vec{n}\cdot \vec{p}}dS$ into a double integral over $\theta_2\in [0,\frac{\pi}{2}]$ and $\varphi_2\in [0,2\pi]$, we can rewrite this as
\begin{align}
	\int \frac{\vec{p}_2-\vec{p}_1}{m}\cdot \nabla_{\hat{\pt{x}}}f_2\,d\Gamma_2 &= \int d\vec{\varsigma}_2\, d\vec{p}_2\, d\nu_2 \,\int_{0}^{\frac{\pi}{2}}\int_{0}^{2\pi} \mathcal{B}^{+}(\theta, \abs{\vec{p}}) f_2(\Gamma_1^\prime,\Gamma_2^\prime,t)\, d\theta_2\, d\varphi_2 \\
	&- \int d\vec{\varsigma}_2\, d\vec{p}_2\, d\nu_2 \,\int_{0}^{\frac{\pi}{2}}\int_{0}^{2\pi} \mathcal{B}^{-}(\theta, \abs{\vec{p}}) f_1(\Gamma_1,t)f_1(\Gamma_2,t)\, d\theta_2\, d\varphi_2,\nonumber
\end{align}
where $\mathcal{B}^{\pm}(\theta, \abs{\vec{p}})$ are known as the scattering kernels, which we scaled by the mass, and is the Jacobian of the transformation from the surface area element to the double integral \cite{harris}.
As usual, we can relate the scattering kernel to the differential cross-section of the interaction, i.e.~$\mathcal{B}^{\pm}(\theta, \abs{\vec{p}}) = \abs{\vec{p}}\Sigma^{\pm}(\abs{\vec{p}},\theta)$, to obtain
\begin{align}
	\int \frac{\vec{p}_2-\vec{p}_1}{m}\cdot \nabla_{\hat{\pt{x}}}f_2\,d\Gamma_2 &= \int d\vec{\varsigma}_2\, d\vec{p}_2\, d\nu_2 \,\int_{0}^{\frac{\pi}{2}}\!\!\int_{0}^{2\pi} \abs{\vec{p}}\Sigma^+(\abs{\vec{p}},\theta_2)f_2(\Gamma_1^\prime,\Gamma_2^\prime,t)\, d\theta_2\, d\varphi_2 \\
	&- \int d\vec{\varsigma}_2\, d\vec{p}_2\, d\nu_2 \,\int_{0}^{\frac{\pi}{2}}\!\!\int_{0}^{2\pi} \abs{\vec{p}}\Sigma^-(\abs{\vec{p}},\theta_2)f_1(\Gamma_1,t)f_1(\Gamma_2,t)\, d\theta_2\, d\varphi_2.\nonumber
\end{align}
We now trace back the collision coordinates $\Gamma_1^\prime$ and $\Gamma_2^\prime$ to the pre-collision coordinates $\Gamma_1$ and $\Gamma_2$, to apply the molecular chaos assumption also on the first term of the collision operator:
\begin{align}
	\int \frac{\vec{p}_2-\vec{p}_1}{m}\cdot \nabla_{\hat{\pt{x}}}f_2\,d\Gamma_2 &= \int d\vec{\varsigma}_2\, d\vec{p}_2\, d\nu_2 \,\int_{0}^{\frac{\pi}{2}}\!\!\int_{0}^{2\pi} \abs{\vec{p}}\Sigma^+(\abs{\vec{p}},\theta_2)f_1(\Gamma_1,t)f_1(\Gamma_2,t)\, d\theta_2\, d\varphi_2 \\
	&- \int d\vec{\varsigma}_2\, d\vec{p}_2\, d\nu_2 \,\int_{0}^{\frac{\pi}{2}}\!\!\int_{0}^{2\pi} \abs{\vec{p}}\Sigma^-(\abs{\vec{p}},\theta_2)f_1(\Gamma_1,t)f_1(\Gamma_2,t)\, d\theta_2\, d\varphi_2,\nonumber
\end{align}
Notice that in the previous equation we have only transformed the arguments of the integral, not the domain of integration, and no Jacobian is needed to transform the integral.

Following \cite{cercignaniLampis} we introduce the pre-collisional variables $\Xi_i\coloneqq \left(\vec{p}_i,\vec{\varsigma}_i,\nu_i\right)$ and the post-collisional variables $\Xi_i^\prime\coloneqq \left(\vec{p}_i^\prime,\vec{\varsigma}_i^\prime,\nu_i^\prime\right)$ and the transition probabilities, here denoted $W\!\left(\Xi_1^\prime,\Xi_2^\prime\mapsto \Xi_1,\Xi_2\right)$ and $W\!\left(\Xi_1,\Xi_2\mapsto \Xi_1^\prime,\Xi_2^\prime\right)$ defined as 
\begin{align}
	\int d\Xi_1^\prime\, d\Xi_2^\prime W\left( \Xi_1^\prime,\Xi_2^\prime\mapsto \Xi_1,\Xi_2\right) = \abs{\vec{p}} \Sigma^+(\abs{\vec{p}},\theta_2), \\
	\int d\Xi_1^\prime \, d\Xi_2^\prime W\left( \Xi_1,\Xi_2\mapsto \Xi_1^\prime,\Xi_2^\prime\right) = \abs{\vec{p}} \Sigma^-(\abs{\vec{p}},\theta_2).
\end{align}
We then rewrite the previous integral as the collision operator
\begin{align}
	\label{eq:collision_operator}
	C[f_1,f_1]  \coloneqq \!\! \int     \!d\Xi_1^\prime\, d\Xi_2^\prime \, d\Xi_2 &\int_{0}^{\frac{\pi}{2}}\!\!\int_{0}^{2\pi} \!\!W\left( \Xi_1^\prime,\Xi_2^\prime\mapsto \Xi_1,\Xi_2\right)f_1(\Gamma_1^\prime,t)f_1(\Gamma_2^\prime,t)\\
	- &W\left( \Xi_1,\Xi_2\mapsto \Xi_1^\prime,\Xi_2^\prime\right)f_1(\Gamma_1,t)f_1(\Gamma_2,t)\, d\theta_2\, d\varphi_2.\nonumber
\end{align}
\begin{remark}
	\label{rmk:cercignaniLampis}
	It is important to stress that we are only considering instantaneous binary interactions.
	This is equivalent to stating that only the conjugate momenta are exchanged during the interaction while the order parameters and the positions are left unchanged in the instant at which the interaction takes place.
	Furthermore, we point out that in the previous equation, not only the linear momentum of the two constituents is exchanged, but also the angular momentum associated to the order parameters.
	This observation will be crucial in characterizing the collision invariants of the operator $C[\cdot,\cdot]$. 
	\rev{
		It is well known that the hypothesis of instantaneous binary interactions is not always valid in physical systems of interest, and in particular leads to unphysical predictions in the context of the collision rule considered in \cref{ex:collision_response} \cite{impactParadox}.
		However, the assumption of instantaneous binary interactions is a common modelling assumption in kinetic theory. For this reason we opt to maintain it in this work, and plan to relax it in future works following approaches similar to those of \cite{kanzler}.
	}
\end{remark}

The transformation from the pre-collision coordinates to the post-collision coordinates $\Gamma_1\mapsto \Gamma_1^\prime$ and $\Gamma_2\mapsto \Gamma_2^\prime$ is volume preserving.
In fact, such transformation is the consequence of a Hamiltonian flow, which is known to be volume preserving \cite{fasanoMarmi}.
Furthermore, reversing the time flow of the Hamiltonian system we obtain that 
\begin{equation}
	\label{eq:cercignaniLampis}
	\int d\Xi_1\, d\Xi_2 \;W\left( \Xi_1^\prime,\Xi_2^\prime\mapsto \Xi_1,\Xi_2\right) = \int d\Xi_1\, d\Xi_2\; W\left( \Xi_1,\Xi_2\mapsto \Xi_1^\prime,\Xi_2^\prime\right).
\end{equation}
This last equation, often called the reciprocity principle, is the identity used in \cite{cercignaniLampis} to prove the H-theorem for polyatomic gases, without resorting to the Boltzmann ``closed collision cycle'' argument.
In particular, the previous equation states that the marginal transition probability with respect to the pre-collision conjugate momenta must be equal when transitioning from the pre-collision to the post-collision coordinates and vice-versa.

\subsection{Collision invariants and steady state}
\label{sec:collision_invariants}

With the collision operator in hand, we now turn to proving a H-theorem. As in the case of polyatomic gases, the lack of symmetry of the collision kernel (manifesting as the presence of two different transition probabilities) requires a careful treatment to prove the Boltzmann inequality. The argument in this section follows that of Cercignani \& Lampis~\cite{cercignaniLampis}.

We consider a generalization of the collision operator defined in \eqref{eq:collision_operator}:
\begin{align}
	\label{eq:collision_operator_general}
	C[f,g] =&\frac{1}{2} \! \int     \!d\Xi_1^\prime\, d\Xi_2^\prime \, d\Xi_2 \!\!\int_{0}^{\frac{\pi}{2}}\!\!\!\!\int_{0}^{2\pi} \!\!\!\!\!\!\!W\left( \Xi_1^\prime,\Xi_2^\prime\mapsto \Xi_1,\Xi_2\right)\Big(f(\Gamma_1^\prime,t)g(\Gamma_2^\prime,t)+f(\Gamma_2^\prime,t)g(\Gamma_1^\prime,t)\Big)\\
	&\!-W\left( \Xi_1,\Xi_2\mapsto \Xi_1^\prime,\Xi_2^\prime\right)\Big(f(\Gamma_1,t)g(\Gamma_2t)+f(\Gamma_2,t)g(\Gamma_1,t)\Big)\, d\theta_2\, d\varphi_2.\nonumber
\end{align}
We have dropped the subscript $1$ from the one particle distribution function to lighten the notation.
We are interested in studying integrals of a quantity $\psi(\Gamma)$ tested against the collision operator, i.e.
\begin{equation}
	\label{eq:weighted_integral}
	\int d\Xi_1 \psi(\Gamma_1) C[f,g].
\end{equation}
\begin{lemma}
\label{lem:weak-form}
The result of testing the collision operator $C[f, g]$ against a function $\psi(\Gamma)$ is given by
	\begin{align}
		\label{eq:collision_invariants}
		\int &d\Xi_1\, \psi(\Gamma_1)C[f,g]=\\
		+&\frac{1}{8}\! \int \!\!\!d\Xi_1^\prime\, d\Xi_2^\prime d\Xi_2\,d\Xi_1 \!\!\int_{0}^{\frac{\pi}{2}}\!\!\!\int_{0}^{2\pi} \!\!\!\!\!\!\!W\!\left( \Xi_1^\prime,\Xi_2^\prime\mapsto \Xi_1,\Xi_2 \right) \times\!\left(\psi(\Gamma_1)+\psi(\Gamma_2)-\psi(\Gamma_1^\prime)-\psi(\Gamma_2^\prime)\right) \nonumber\\
		&\times \Big(f(\Gamma_1^\prime,t)g(\Gamma_2^\prime,t)+f(\Gamma_2^\prime,t)g(\Gamma_1^\prime,t)\Big)\,d\theta_2\,d\varphi_2 \nonumber\\
		-&\frac{1}{8}\! \int \!\!\!d\Xi_1^\prime\, d\Xi_2^\prime d\Xi_2\,d\Xi_1 \!\!\int_{0}^{\frac{\pi}{2}}\!\!\!\int_{0}^{2\pi} \!\!\!\!\!\!\!W\!\left( \Xi_1,\Xi_2\mapsto \Xi_1^\prime,\Xi_2^\prime\right) \times\! \left(\psi(\Gamma_1)+\psi(\Gamma_2)-\psi(\Gamma_1^\prime)-\psi(\Gamma_2^\prime)\right) \nonumber\\
		&\times\Big(f(\Gamma_1,t)g(\Gamma_2,t)+f(\Gamma_2,t)g(\Gamma_1,t)\Big)\,d\theta_2\,d\varphi_2.\nonumber
	\end{align}
\end{lemma}
\begin{proof}
	The proof is a classical result \cite{cercignani}, a consequence of the fact that the transformation $\Gamma_1\mapsto \Gamma_2$ has unitary Jacobian and the fact that the transition $(\Gamma_1,\Gamma_2)\mapsto (\Gamma_1^\prime,\Gamma_2^\prime)$ is also volume preserving.
	Furthermore, one has to use the fact that the transition probability $W(\cdot\mapsto \cdot)$ is independent of which one is molecule one and which one is molecule two, since all molecules are indistinguishable \cite{cercignaniLampis}.
\end{proof}
\begin{remark}
	Notice that \eqref{eq:collision_invariants} differs from the classical weak formulation of the Boltzmann collision operator.
	This is because we are working here under \eqref{eq:cercignaniLampis}, which is referred to as the reciprocity assumption in \cite{cercignaniLampis}.
	Under the stronger assumption of detailed balance, i.e.~$W\left( \Xi_1^\prime,\Xi_2^\prime\mapsto \Xi_1,\Xi_2\right) = W\left( \Xi_1,\Xi_2\mapsto \Xi_1^\prime,\Xi_2^\prime\right)$, we can group together the first and second terms in \eqref{eq:collision_invariants} to obtain the classical weak form of the Boltzmann collision operator:
		\begin{align}
			\int &d\Xi_1\, \psi(\Gamma_1)C[f,g]=\\
			&\frac{1}{8}\! \int \!\!\!d\Xi_1^\prime\, d\Xi_2^\prime d\Xi_2\,d\Xi_1 \!\!\int_{0}^{\frac{\pi}{2}}\!\!\!\int_{0}^{2\pi} \!\!\!\!\!\!\!W\!\left( \Xi_1,\Xi_2\mapsto \Xi_1^\prime,\Xi_2^\prime\right) \times\! \left(\psi(\Gamma_1)+\psi(\Gamma_2)-\psi(\Gamma_1^\prime)-\psi(\Gamma_2^\prime)\right) \nonumber\\
			&\times\Big(f(\Gamma_1^\prime,t)g(\Gamma_2^\prime,t)+f(\Gamma_2^\prime,t)g(\Gamma_1^\prime,t)-f(\Gamma_1,t)g(\Gamma_2,t)-f(\Gamma_2,t)g(\Gamma_1,t)\Big)\,d\theta_2\,d\varphi_2. \nonumber
		\end{align}

\end{remark}
\begin{theorem}
The quantities $\psi_1(\Gamma_1)=1$, $\psi_2(\Gamma_1)=\vec{p}$, $\psi_3(\Gamma_1)=A_{\pt{x}} \mat{Q}\cdot \vec{p}+A_\nu \mat{Q}\cdot \mat{B}(\nu)\vec{\dot{\nu}}$, and $\psi_4(\Gamma_1)=\frac{1}{2m}\vec{p}\cdot\vec{p} + \frac{1}{2}\varsigma\cdot \mat{B}(\nu)^{-1}\varsigma$ are collision invariants, i.e
	\begin{equation}
		\label{eq:collision_invariants_2}
		\int d\Xi_1\, \psi_i(\Gamma_1)C[f,f]=0,\quad i=1,2,3,4.
	\end{equation}
\end{theorem}
\begin{proof}	
	From Lemma \ref{lem:weak-form} we can conclude that if the following identity holds, then the quantity $\psi(\Gamma_1)$ is a collision invariant:
	\begin{equation}
		\label{eq:collision_invariants_final}
		\psi(\Gamma_1,t)+\psi(\Gamma_2,t)-\psi(\Gamma_1^\prime,t)-\psi(\Gamma_2^\prime,t) = 0.
	\end{equation}
	From Section \ref{sec:microscopic_collision} we know that $\psi_i$, with $i=1,2,3,4,$ satisfies this identity.
\end{proof}
Of particular interest are one-particle distribution functions $\bar{f}(\Gamma,t)$ for which the collision operator $C[\bar{f},\bar{f}]$ vanishes.
The standard approach to identifying such functions is to prove a Boltzmann type inequality.
As pointed out by Lorentz~\cite[Volume 6, p.~74]{lorentz}, Boltzmann's initial approach to proving the H-theorem is not applicable to polyatomic gases.
The same issues arise in the context of ordered fluids.
\begin{theorem}
	\label{thm:boltzmannInequality}
	The following Boltzmann-type inequality holds:
	\begin{equation}
		\int d\Xi_1\, \log(f(\Gamma_1,t))C[f,f]\leq 0.
	\end{equation}
\end{theorem}
\begin{proof}
	The proof of this result is based on the weak form of the collision operator \eqref{eq:collision_invariants} and follows the same lines as the proof of the Boltzmann inequality for polyatomic gases \cite{cercignaniLampis}.
    A detailed proof is given in Appendix \ref{sec:mathaspects}.
\end{proof}
\begin{theorem}
	\label{thm:maxwellian}
	Let the order parameter manifold $(\mathcal{M},\mathcal{A})$ be such that the group action $\mathcal{A}$ is transitive.
	We define a Maxwellian distribution to be any function of the following form
	\rev{
	\begin{align}
			\bar{f}(\Gamma_1,t) = \exp\Big(c+&\vec{a}\cdot\vec{p}_1+\vec{b}\cdot \vec{\varsigma}_1+d(m^{-1}\vec{p}_1\cdot\vec{p}_1+\vec{\varsigma}_1\cdot \mat{B}(\nu_1)^{-1}\vec{\varsigma}_1)\Big)\nonumber, 
	\end{align}
	}
    where  $c, d \in \mathbb{R}$ and $\vec{a}, \vec{b} \in \mathbb{R}^d$.
	The Maxwellian distribution is the only stationary solution of \eqref{eq:boltzmann}, i.e.
	\begin{equation}
		 m^{-1} \vec{p}\cdot\nabla_{\pt{x}} \bar{f}+\mat{B}(\nu)^{-1}\vec{\varsigma} \cdot\nabla_{\nu}\bar{f} = 0.
	\end{equation}
\end{theorem}
\begin{proof}
	The proof of this result follows the classical approach pioneered by Boltzmann.
	The only differences are minor technicalities arising from the fact that we are working on a manifold.
	We deal with this issue in Appendix \ref{sec:mathaspects}, following the approach presented in \cite{arkerydCercignani}.
\end{proof}
\begin{exA}[Non-diffusing gas bubbles: Collision operator]
	\label{ex:collision_operator_bubbles}
	Continuing on the setting of nondiffusing gas bubbles introduced in \cref{ex:bubbles} and using the collision rule introduced in \cref{ex:bubbles_collision}, we give an explicit expression for the collision operator $C[f,f]$.
	In particular we notice that we only need to give an expression for the transition probability $W(\cdot\mapsto\cdot)$.
	Following the idea that gas bubbles behave like a mixture of ``mass-less'' hard spheres that interchange volume when they collide, we have assumed that the interaction rule between two bubbles is the same as the one between two hard spheres, with the additional exchange of the order parameter, see \cref{ex:bubbles_collision}.
	We have thus neglected the first peculiarity of gas mixtures, i.e. the fact that different molecules might have different masses. However, we still need to take into account the fact that different molecules might have different order parameters.
	To this end we assume the transition probability $W(\Xi_1,\Xi_2\mapsto \Xi_1^\prime,\Xi_2^\prime)$ is the same one presented in \cite[Chapter 6.2]{cercignani}, i.e.
	\begin{align}
		W(\Xi_1,\Xi_2\mapsto \Xi_1^\prime,\Xi_2^\prime) &= S(\nu_1,\nu_2)\delta(\vec{p}_1+\vec{p}_2-\vec{p}_1^\prime-\vec{p}_2^\prime)\\
		&\times \delta(\nu_1+\nu_2-\nu_1^\prime - \nu_2^\prime)\delta(\norm{\vec{p}_1}^2+\norm{\vec{p}_2}^2-\norm{\vec{p}_1^\prime}^2-\norm{\vec{p}_2^\prime}^2)\nonumber,
	\end{align}
	where $\delta$ is the Dirac delta measure and $S(\nu_1,\nu_2)$ encodes that two bubbles with larger volume are more likely to collide. For example, assuming that the bubbles behave like Maxwellian molecules, we can choose as $S(\nu_1,\nu_2)=\frac{1}{2}(\nu_1+\nu_2)$.
\end{exA}
\begin{exB}[Calamitic molecules in two dimensions: Collision operator]
	\label{ex:collision_operator_calamitic}
	Continuing on the setting of two-dimensional calamitic molecules introduced in \cref{ex:2D}, we give an explicit expression for the transition probability $W(\Xi_1,\Xi_2\mapsto \Xi_1^\prime,\Xi_2^\prime)$.
	We follow \cite[Appendix]{curtissV}, where Curtiss \& Dahler prove that the transition probability for the collision of two rigid convex bodies is given by
	\begin{equation}
		W(\Xi_1,\Xi_2\mapsto \Xi_1^\prime,\Xi_2^\prime) = (\vec{\mathfrak{g}}\cdot \vec{n})S(\vec{\nu}_1,\vec{\nu}_2),
	\end{equation}
	with $S(\vec{\nu}_1,\vec{\nu}_2)$ being the Jacobian of the transformation from the excluded volume integral between the two calamitic molecules to the sphere, $\vec{\mathfrak{g}}$ is the relative velocity of the point of contact between the two molecules, and $\vec{n}(\theta_2,\varphi_2)$ is the unit normal to $\mathbb{S}^2$.
	Under the vanishing girth assumption, we can neglect any rotation of the molecules around the axis $\vec{\nu}_i$, i.e.~$\vec{\omega}_i\cdot \vec{\nu}_i = 0$, thus expressing the angular velocity $\vec{\omega}_i$ in terms of the conjugate momenta to $\vec{\nu}_i$ and $\vec{\nu}_i$ itself, i.e.~$\vec{\omega}_i = \vec{\nu}_i\times \vec{\varsigma}_i$.
	Thus, the effective velocity of the point of contact connected to the center of mass of molecule $i$ by a vector $\vec{r}_i$ is given by
	\begin{equation}
		\label{eq:effective_velocity_calamitic}
		\vec{\mathfrak{g}} = \frac{1}{m}\left(\vec{p}_1-\vec{p}_2\right) + \Big((\vec{\varsigma}_1\cdot\vec{r}_i)\vec{\nu}_1-(\vec{\varsigma}_2\cdot\vec{r}_i)\vec{\nu}_2-(\vec{\nu}_1\cdot \vec{r}_i)\vec{\varsigma}_1+(\vec{\nu}_2\cdot \vec{r}_i)\vec{\varsigma}_2\Big).
	\end{equation}
\end{exB}
\begin{remark}
Curtiss \& Dahler further assumed that the molecules are hard top-symmetric bodies to simplify the collision operator to obtain a standard gain-loss splitting, typical of the Boltzmann collision operator. They do this to prove an H-theorem. We can retain the full complexity of this collision operator, thanks to \Cref{thm:boltzmannInequality}.
\end{remark}
\begin{exC}[Head-tail symmetric calamitic molecules: Collision operator]
	\label{ex:collision_operator_head_tail}
	Continuing on the setting of head-tail symmetric calamitic molecules introduced in \cref{ex:head-tail}, we can give an explicit expression for the transition probabilities.
	As remarked in \cref{ex:collision_head_tail}, the binary interaction between two head-tail symmetric calamitic molecules at a microscopic level can be described by the same transition probability $W(\Xi_1,\Xi_2\mapsto \Xi_1^\prime,\Xi_2^\prime)$ as the one presented in \cref{ex:collision_operator_calamitic}.
	The only difference is that in this context the effective velocity of the point of contact between the two molecules is given by
	\begin{equation}
		\vec{\mathfrak{g}} = \frac{1}{m}\left(\vec{p}_1-\vec{p}_2\right) +
		\begin{pmatrix}
			\omega_1r_{1}^{(y)}-\omega_2r_{2}^{(y)}\\
			\omega_2r_{2}^{(x)}-\omega_1r_{1}^{(x)}\\
		\end{pmatrix},
	\end{equation}
	since in \cref{ex:head-tail} we have limited ourselves to the two-dimensional case, and considered as conjugate momenta the angular velocities $\omega_i$.
\end{exC}
\subsection{Mean-field potential}
\label{sec:meanField}
We now turn our attention to the transport part of \eqref{eq:boltzmann}.
We devote particular care to understanding the role of the mean-field potential $\mathcal{W}$ and of the Vlasov-type force $\mathcal{V}$.
For the moment, let us consider the space-homogeneous case, i.e.~$\mathcal{W}=\mathcal{W}(\nu,\vec{\varsigma})$ and $\mathcal{V}=\mathcal{V}(\nu,\vec{\varsigma})$, and assume that the collision operator is ignored.
Thus \eqref{eq:boltzmann} reads as
\begin{equation}
	\label{eq:boltzmannHomogeneous}
	\frac{\partial f}{\partial t} +\mat{B}(\nu)^{-1}\vec{\varsigma}\cdot \nabla_{\nu}\, f + \mathcal{V}\cdot \nabla_{\vec{\varsigma}} f
 = 0.
\end{equation}
To study the dynamics of this equation, we assume that the one-particle distribution function $f$ is of the form of the empirical one-particle distribution, i.e.
\begin{equation}
	\label{eq:empirical_distribution}
	f(\nu,\vec{\varsigma},t)\approx f^N(\nu,\vec{\varsigma},t) = \frac{1}{N} \sum_{i=1}^{N} \delta(\nu-\nu_i(t))\otimes\delta(\vec{\varsigma}-\vec{\varsigma}_i(t)),
\end{equation}
where $\nu_i(t)$ and $\vec{\varsigma}_i(t)$ are the order parameters of the $i$-th particle.
This ansatz is equivalent to assuming that the particles are non-interacting and that the dynamics of each constituent is described by the following system of ordinary differential equations (ODEs),
\begin{equation}
	\label{eq:ODEs}
	\frac{d \nu_i}{d t}  = \vec{\varsigma}_i,\qquad \frac{d \vec{\varsigma}_i}{d t}  = \mathcal{V}(\nu_i,\vec{\varsigma}_i).
\end{equation}
Thus, to describe the long term dynamics of the system we need to study the stability of the fixed points of the system of ODEs \eqref{eq:ODEs}.
This is written in a general form, which might be impractical to work with, since it would require operations on the manifold $\mathcal{M}$. Instead, we here prefer to work on the domain of the parametrization of the manifold, and we therefore rewrite the characteristics \eqref{eq:ODEs} in terms of the parametrization of $\mathcal{M}$.
\begin{exA}[Non-diffusing gas bubbles: Transport]
	\label{ex:transport_bubbles}
	In the setting of nondiffusing gas bubbles, the manifold $\mathcal{M}$ is the unit interval. This can be trivially parametrized, and the transport term is the classical one, where $\frac{\partial f}{\partial \nu}$ is computed in the standard way.
\end{exA}
\begin{exB}[Calamitic molecules in two dimensions: Transport]
	\label{ex:transport_calamitic}
	From \eqref{eq:BBGKY1_weak_order} we notice that the Vlasov-type term and the transport appearing on the left-hand side are non-standard, in the sense that the gradient determining the transport direction has to be computed on the order parameter manifold $\mathcal{M}$.
	In the setting of two-dimensional calamitic molecules, which have been the focus of \cref{ex:2D}, \cref{ex:infinitesimalGenerator2D} and \cref{ex:angular2D}, we begin by finding an expression for $\mat{B}(\vec{\nu})^{-1}\vec{\varsigma}\cdot \nabla_{\vec{\nu}}f$ in terms of $\theta$ and $\omega$.
	We observe that $\mat{B}(\vec{\nu})^{-1}\vec{\varsigma}$ is equivalent to $\dot{\vec{\nu}}$ and that the following identity holds for segment like molecules $\dot{\vec{\nu}} = \vec{\omega}\times \vec{\nu}$, i.e.
	\begin{equation}
		\dot{\vec{\nu}} = \begin{bmatrix}  0 \\ 0 \\ \omega \end{bmatrix}\times \begin{bmatrix} \cos(\theta) \\ \sin(\theta) \\ 0 \end{bmatrix} = \omega \begin{bmatrix} -\sin(\theta)\\ \cos(\theta)\\0\end{bmatrix}.
	\end{equation}
	We can then construct the projector onto the tangent space to the order parameter manifold $\mathcal{M}$, i.e.
	\begin{equation}
		\Pi_{\mathcal{M}} = \begin{bmatrix} -\sin(\theta) \\ \cos(\theta) \\ 0 \end{bmatrix}\otimes \begin{bmatrix} -\sin(\theta) \\ \cos(\theta) \\ 0 \end{bmatrix} = \begin{bmatrix} \sin^2(\theta) & -\sin(\theta)\cos(\theta) & 0 \\ -\sin(\theta)\cos(\theta) & \cos^2(\theta) & 0 \\ 0 & 0 & 0 \end{bmatrix}.
	\end{equation}
	We then compute $\nabla_{\vec{\nu}}f$ in terms of $\theta$ and $\omega$ using the chain rule, i.e.
	\begin{align}
		&\partial_{\vec{\nu}^{(x)}}f = \frac{-1}{\sqrt{1-\left(\nu^{(x)}\right)^2}}\partial_\theta f=\frac{-1}{\sin(\theta)}\partial_\theta f,\\
		&\partial_{\vec{\nu}^{(y)}}f = \frac{1}{\sqrt{1-\left(\nu^{(y)}\right)^2}}\partial_\theta f=\frac{1}{\cos(\theta)}\partial_\theta f.
	\end{align}
	Thus, we can compute $\mat{B}(\nu)^{-1}\vec{\varsigma}\cdot \nabla_{\vec{\nu}}f$ as
	\begin{align}
		\mat{B}(\nu)^{-1}\vec{\varsigma}\cdot \nabla_{\vec{\nu}}f &= \omega \begin{bmatrix} -\sin(\theta) \\ \cos(\theta) \\ 0 \end{bmatrix}\cdot \Pi_{\mathcal{M}}\begin{bmatrix} -\sin(\theta)^{-1} \\ \cos(\theta)^{-1}\\ 0 \end{bmatrix}\partial_\theta f \\
		& = 2\omega(\sin^2(\theta) + \cos^2(\theta))\partial_\theta f = 2\omega\partial_\theta f.
	\end{align}

This recovers the transport term of the Boltzmann--Curtiss equation.
\end{exB}

\begin{exC}[Head-tail symmetric calamitic molecules: Transport]
	Even if $\mathbb{RP}^1$ is homeomorphic to the circle $\mathbb{S}^1$, the order parameter manifold in \cref{ex:head-tail} is different to the one introduced in \cref{ex:2D}, by virtue of the fact that the group action $\mathcal{A}^{\text{ht}}$, introduced in \cref{ex:infinitesimalGeneratorHeadTail}, is different from the group action $\mathcal{A}$ introduced in \cref{ex:infinitesimalGenerator2D}.
	Thus, the gradient determining the transport direction in the order parameter manifold $\mathcal{M}$ in the context of head-tail symmetric calamitic molecules is different from the one in the context of simple calamitic molecules.
	In particular, considering the parametrization \eqref{eq:head-tail-parametrization} of $\mathbb{RP}^1$ we can compute $\mat{B}(\nu)^{-1}\vec{\varsigma}\cdot \nabla_{\nu}f$. In fact, it is easy to see that 
	\begin{equation}
		\mat{B}(\nu)^{-1}\vec{\varsigma} = \omega \begin{bmatrix}
			-2\sin(\theta) \cos(\theta) & \sin^2(\theta)-\cos^2(\theta) \\
			\sin^2(\theta)-\cos^2(\theta) & 2\sin(\theta)\cos(\theta)
		\end{bmatrix}.
	\end{equation}
	Following the same procedure as above we can compute the projector onto the tangent space to the order parameter manifold $\mathcal{M}$ to be the following fourth order tensor,
	\begin{equation}
		\Pi_\mathcal{M}\!\! =\!\! \begin{bmatrix}
			-2\sin(\theta) \cos(\theta) & \sin^2(\theta)-\cos^2(\theta) \\
			\sin^2(\theta)-\cos^2(\theta) & 2\sin(\theta)\cos(\theta)
		\end{bmatrix}\otimes \begin{bmatrix}
			-2\sin(\theta) \cos(\theta) & \sin^2(\theta)-\cos^2(\theta) \\
			\sin^2(\theta)-\cos^2(\theta) & 2\sin(\theta)\cos(\theta)
		\end{bmatrix}\!.
	\end{equation}
	Lastly, using the chain rule as in the previous example, we can compute $\nabla_{\nu}f$ in terms of $\theta$, i.e.
	\begin{equation}
		\nabla_{\nu}f = \begin{bmatrix}
			-\sin(\theta)^{-1} &  \frac{1}{2}\cos(2\theta)^{-1}\\
			\frac{1}{2}\cos(2\theta)^{-1} & \cos(\theta)^{-1}
		\end{bmatrix}.
	\end{equation}
	Thus we conclude that 
	\begin{align}
\mat{B}(\nu)^{-1}\vec{\varsigma}\cdot \nabla_{\nu}f &= \\ 
		 \omega\begin{bmatrix}
			-2\sin(\theta) \cos(\theta) & \sin^2(\theta)-\cos^2(\theta) \\
			\sin^2(\theta)-\cos^2(\theta) & 2\sin(\theta)\cos(\theta)
		\end{bmatrix}:&\Pi_{\mathcal{M}}:\begin{bmatrix}
			-\sin(\theta)^{-1} &  \frac{1}{2}\cos(2\theta)^{-1}\\
			\frac{1}{2}\cos(2\theta)^{-1} & \cos(\theta)^{-1}
		\end{bmatrix}\partial_\theta f , \nonumber
	\end{align}
	where $:$ denotes the double tensor contraction, i.e. $\mat{A}:\mat{B} = \sum_{i,j} A_{ij}B_{ij}=\tr(\mat{A}^T\mat{B})$.
\end{exC}
\begin{exB}[Calamitic molecules in two dimensions: Mean-field potential]
	\label{ex:Vlasov2D_calamitic_quadratic}
	Let us consider the setting of \cref{ex:2D} and assume that the mean-field potential is given by
	\begin{equation}
		\label{eq:quadratic_potential}
		\mathcal{W}(\theta,\omega) = \frac{1}{2}\alpha\abs{\theta-\hat{\theta}}^2 + \beta \omega\theta,\qquad \hat{\theta}=\arctan(\hat{\nu}^{(y)},\hat{\nu}^{(x)}), \qquad \hat{\vec{\nu}}=\frac{1}{N} \sum_{i=1}^N \vec{\nu}_i.
	\end{equation}
	Under this hypothesis the Vlasov-type force can be computed from \eqref{eq:vlasov_Force} to be
	\begin{equation}
		\mathcal{V}(\theta,\omega) = -\alpha\left(\theta-\hat{\theta}\right) - \beta\omega.
	\end{equation}
	Thus the system of ODEs \eqref{eq:ODEs} can be recast as a linear system of ODEs, i.e.
	\begin{equation}
		\frac{d}{dt}
		\begin{bmatrix}
			\theta_i \\
			\omega_i
		\end{bmatrix} = \begin{bmatrix}
			0 & 1 \\ -\alpha & -\beta
		\end{bmatrix}\begin{bmatrix}
			\theta_i\\
			\omega_i
		\end{bmatrix}+\alpha\begin{bmatrix}
			0 \\ \hat{\theta}
		\end{bmatrix}.
	\end{equation}
	We can immediately see that the fixed point of this system is unique and it is given by $\vec{\nu}=\hat{\vec{\nu}}$ and $\vec{\varsigma}=0$.
	It remains to study the stability of the fixed point. The eigenvalues of the Jacobian of the system are given by
	\begin{equation}
		\lambda_{1,2} = \frac{-\beta\pm \sqrt{\beta^2-4\alpha}}{2}.
	\end{equation}
	If $\beta$ vanishes and $\alpha>0$, the equilibrium point $(\hat{\vec{\nu}}, 0)$ is a center.
	If $\alpha<0$, the equilibrium point $(\hat{\vec{\nu}},0)$ is a saddle point, and we do not expect any asymptotic alignment phenomena.
	A different scenario arises if $\beta$ is non-vanishing. If $\beta>0$ the fixed point $(\hat{\vec{\nu}},0)$ is stable provided that also $\alpha>0$, hence we should expect the system of calamitic molecules to exhibit alignment along a preferred direction $\hat{\vec{\nu}}$.
	In particular, if $\beta>0$ and $\alpha>\frac{\beta^2}{4}$ the fixed point $(\hat{\vec{\nu}},0)$ is a stable spiral, thus we should expect the system to exhibit alignment along a preferred direction $\hat{\vec{\nu}}$ in the long term, while there might be some oscillations around the fixed point before the system reaches alignment.
	On the other hand, if $\beta>0$ and $\alpha\in(0,\frac{\beta^2}{4})$ the fixed point $(\hat{\vec{\nu}},0)$ is a stable node, thus we should expect the system to exhibit alignment along a spontaneously emerging direction $\hat{\vec{\nu}}$ in the long term, without any oscillations around the fixed point.
	Figure \ref{fig:ODEs} gives numerical examples of the dynamics of the system of ODEs \eqref{eq:ODEs} for different values of $\alpha$ and $\beta$.
	\begin{figure}[h]
		\begin{center}
		\includegraphics[width=0.3\textwidth]{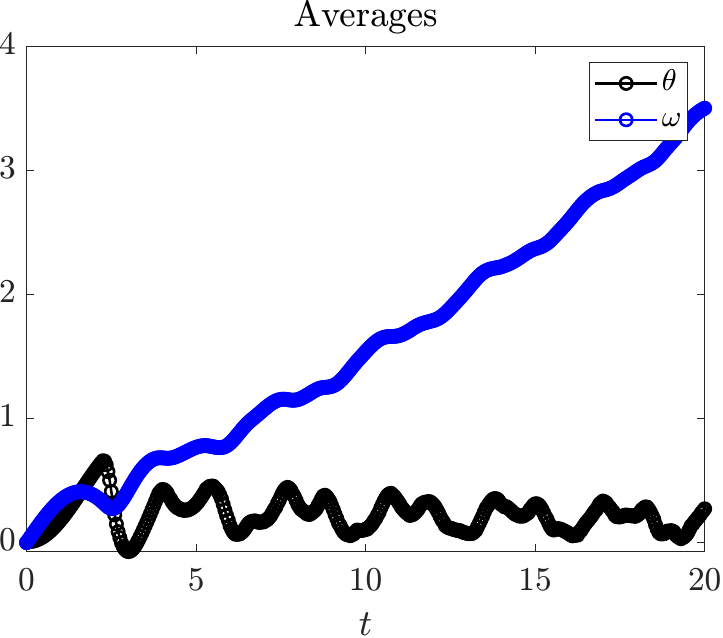}
		\includegraphics[width=0.3\textwidth]{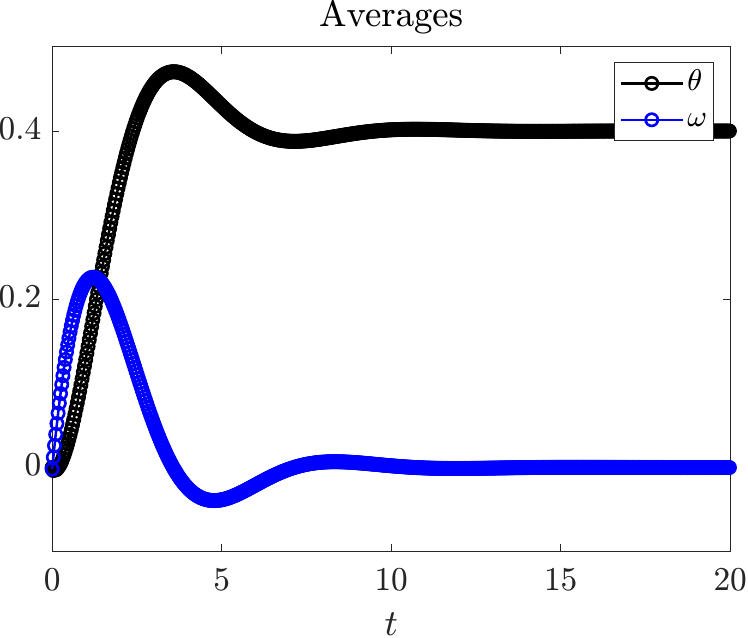}
		\includegraphics[width=0.3\textwidth]{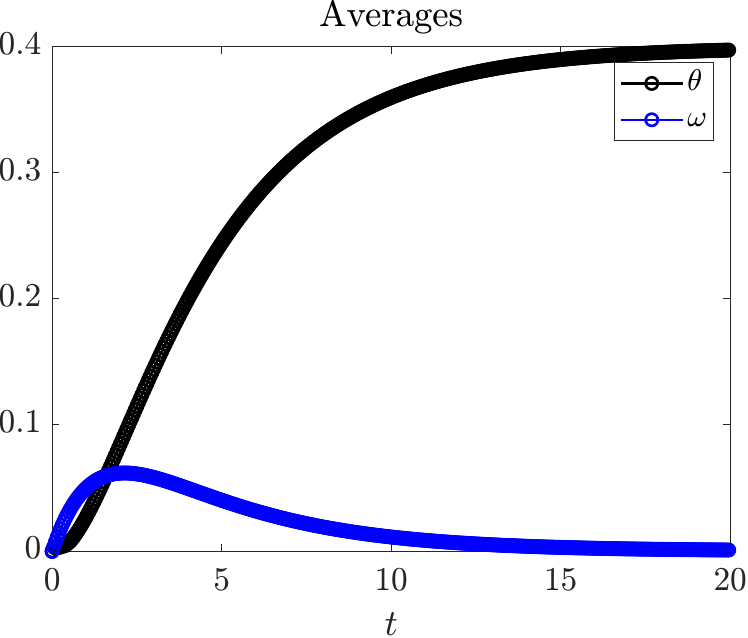}
		\includegraphics[width=0.3\textwidth]{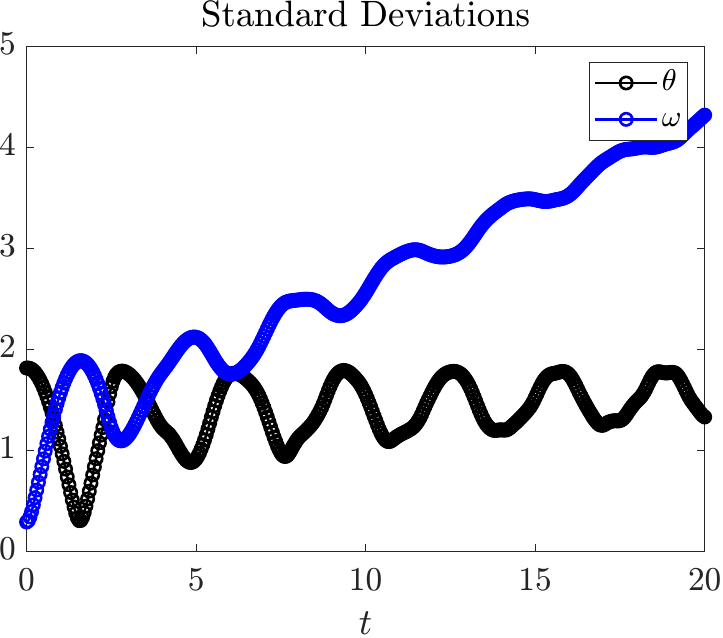}
		\includegraphics[width=0.3\textwidth]{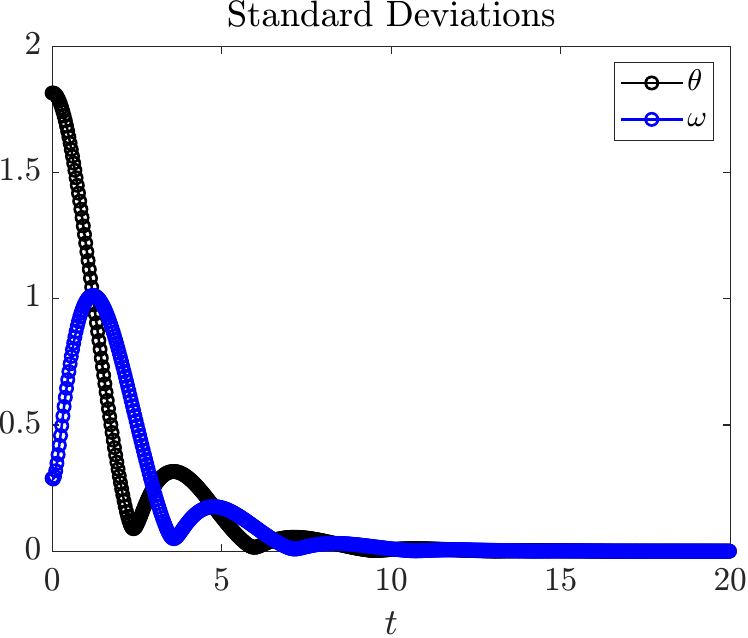}
		\includegraphics[width=0.3\textwidth]{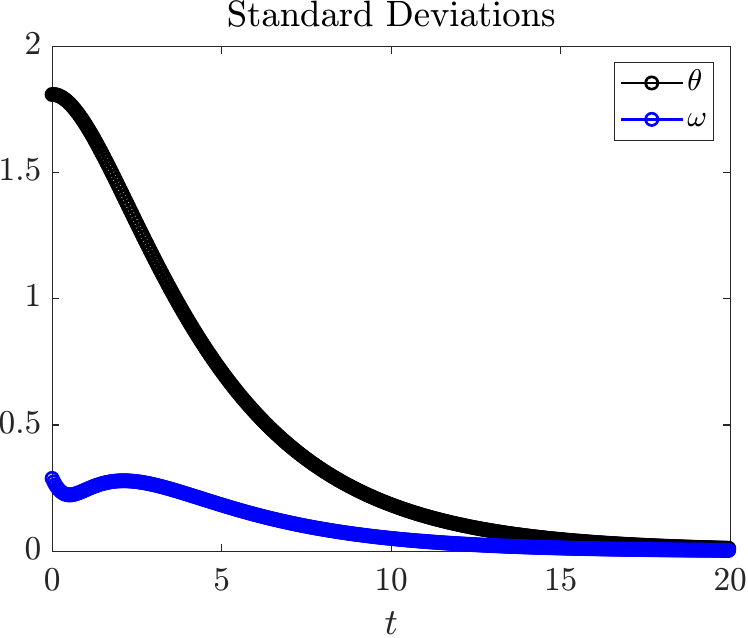}
		\end{center}
		\caption{Mean (top) and standard deviation (bottom) of a system of particles evolving according to \eqref{eq:ODEs} for $\beta=0$ and $\alpha=1$ (saddle point, left), $\beta=1$ and $\alpha=1$ (stable spiral, middle), and $\beta=1$ and $\alpha=0.25$ (stable node, right).
		The initial conditions for $\vec{\nu}$ and $\vec{\varsigma}$ are drawn from a uniform distribution in $[-\pi,\pi]$ and $[-1,1]$ respectively. Lastly we fixed $\vec{\hat{\nu}}=\left(\sin(\hat{\theta}),\cos(\hat{\theta})\right)$ with $\hat{\theta}=0.4$.
		}
		\label{fig:ODEs}
	\end{figure}
\end{exB}
\begin{exB}
	The Vlasov force considered in the previous equation can be obtained as the linearisation of the following non-linear Vlasov force
	\begin{equation}
		\mathcal{V}(\theta) = -\sin\left(\theta-\hat{\theta}\right), \qquad \hat{\theta}=\arctan(\hat{\nu}^{(y)},\hat{\nu}^{(x)}), \qquad \hat{\vec{\nu}}=\frac{1}{N} \sum_{i=1}^N \vec{\nu}_i.
	\end{equation}
	This is inspired by the torque imposed on the calamitic constituents by an external electric field.
	Thus we can observe alignment phenomena in the system of ODEs \eqref{eq:ODEs} also considering the microscopic potential
	\begin{equation}
		\label{eq:cosine_potential}
		\mathcal{W}(\theta) = \alpha\cos\left(\theta-\hat{\theta}\right), \qquad \hat{\theta}=\arctan(\hat{\nu}^{(y)},\hat{\nu}^{(x)}), \qquad \hat{\vec{\nu}}=\frac{1}{N}\sum_{i=1}^N \vec{\nu}_i.
	\end{equation}
\end{exB}
\section{Conclusion}
In this work, we have proposed a unified kinetic theory framework for ordered fluids, inspired by the geometric perspective pioneered by Capriz.
By introducing the notion of an order parameter manifold and characterizing its symmetries via Lie group actions, we have developed a Lagrangian formulation that naturally reflects the internal structure of microscopic constituents.
Through Noether's theorem, we identified key conserved quantities---such as linear momentum, angular momentum, and total energy---that are preserved under two-body interactions, and derived the explicit form of these invariants in both two- and three-dimensional settings.

Building on these foundations, we developed a BBGKY hierarchy adapted to the ordered fluid setting. We then derived a Vlasov--Boltzmann-type kinetic equation that captures both translational and orientational degrees of freedom.
This derivation crucially relied on a weak-order interaction assumption. This assumption is physically justified in a range of systems, including nematic and smectic liquid crystals, that exhibit long-range orientational order.

The resulting kinetic equation provides a flexible and consistent starting point for analyzing the collective dynamics of ordered fluids from first principles.
Its structure accommodates general forms of microscopic interactions, respects the geometric constraints of the order parameter manifold, and complies with classical thermodynamic properties such as the H-theorem via a suitable generalization of the Boltzmann entropy inequality.

Throughout the paper we have considered several examples to illustrate the applicability of our framework. We now briefly summarize the full kinetic equation derived in each of these examples.
\begin{exA}[Non-diffusing gas bubbles: Kinetic equation]
	\label{ex:KT_bubbles}
	In the context of non-diffusing gas bubbles, we derived a kinetic equation that describes the evolution of bubbles in a fluid medium, where the order parameter is the volume fraction $\nu$ associated with the gas bubbles and the momentum $\vec{\varsigma}$ is the rate of change of the volume fraction.
	We assumed that the interaction between bubbles is governed by a collision rule that conserves the total volume fraction and momentum, leading to a Boltzmann-type equation with a collision operator that accounts for the exchange of volume fraction during collisions.
	The resulting kinetic equation is of the form
	\begin{align}
		\label{eq:boltzmann_bubbles}
		\frac{\partial f}{\partial t} &+ \frac{1}{m}\vec{p}_1\cdot\nabla_{\pt{x}}f + \vec{\varsigma}_1\cdot\nabla_{\nu}f =\\
	    &\frac{1}{2} \! \int \!\!\!d\Xi_1^\prime\, d\Xi_2^\prime d\Xi_2 \!\!\int_{0}^{\frac{\pi}{2}}\!\!\!\!\int_{0}^{2\pi} \!\!\!\!\!\!\!W\left( \Xi_1^\prime,\Xi_2^\prime\mapsto \Xi_1,\Xi_2\right)\Big(f(\Gamma_1^\prime,t)g(\Gamma_2^\prime,t)+f(\Gamma_2^\prime,t)g(\Gamma_1^\prime,t)\Big)\!\nonumber \\
	    &\qquad\qquad\qquad\quad-W\!\left( \Xi_1,\Xi_2\mapsto \Xi_1^\prime,\Xi_2^\prime\right)\!\Big(f(\Gamma_1,t)g(\Gamma_2t)+f(\Gamma_2,t)g(\Gamma_1,t)\Big)\, d\theta_2\, d\varphi_2,\nonumber
	\end{align}
	with the transition probability $W(\Xi_1,\Xi_2\mapsto \Xi_1^\prime,\Xi_2^\prime)$ being given by
	\begin{equation}
		W(\Xi_1,\Xi_2\mapsto \Xi_1^\prime,\Xi_2^\prime) = S(\nu_1,\nu_2)\delta(\vec{p}_1+\vec{p}_2-\vec{p}_1^\prime-\vec{p}_2^\prime)\delta(\norm{\vec{p}_1}^2+\norm{\vec{p}_2}^2-\norm{\vec{p}_1^\prime}^2-\norm{\vec{p}_2^\prime}^2),
	\end{equation}
	where $\delta$ is the Dirac delta measure and $S(\nu_1,\nu_2)$ encodes that two bubbles with larger volume are more likely to collide.
\end{exA}
\begin{exB}[Calamitic molecules in two dimensions: Kinetic equation]
	\label{ex:KT_calamitic}
	In the context of two-dimensional calamitic molecules, we derived a kinetic equation that describes the evolution of the distribution function of calamitic molecules, where the order parameter is given by the angle $\theta$ and the angular velocity $\omega$.
	The interaction between calamitic molecules is governed by a collision rule that conserves both linear momentum and angular momentum, leading to a Boltzmann-type equation with a collision operator that accounts for the exchange of orientation during collisions.
	The resulting kinetic equation is of the form
	\begin{align}
		\label{eq:boltzmann_calamitic}
		\frac{\partial f}{\partial t} & +2\omega_1 \nabla_{\theta_1}f + \mathcal{V}\cdot\nabla_{\omega_1}f =\\ 
	    &\frac{1}{2} \! \int \!\!\!d\Xi_1^\prime\, d\Xi_2^\prime d\Xi_2 \!\!\int_{0}^{\frac{\pi}{2}}\!\!\!\!\int_{0}^{2\pi} \!\!\!\!\!\!\!W\left( \Xi_1^\prime,\Xi_2^\prime\mapsto \Xi_1,\Xi_2\right)\Big(f(\Gamma_1^\prime,t)g(\Gamma_2^\prime,t)+f(\Gamma_2^\prime,t)g(\Gamma_1^\prime,t)\Big)\nonumber \\
	    &\qquad\qquad\qquad-W\left( \Xi_1,\Xi_2\mapsto \Xi_1^\prime,\Xi_2^\prime\right)\Big(f(\Gamma_1,t)g(\Gamma_2,t)+f(\Gamma_2,t)g(\Gamma_1,t)\Big)\, d\theta_2\, d\varphi_2,\nonumber
	\end{align}
	with the transition probability $W(\Xi_1,\Xi_2\mapsto \Xi_1^\prime,\Xi_2^\prime)$ being given by
	\begin{equation}
		W(\Xi_1,\Xi_2\mapsto \Xi_1^\prime,\Xi_2^\prime) = (\vec{\mathfrak{g}}\cdot \vec{n})S(\vec{\nu}_1,\vec{\nu}_2),
	\end{equation}
	where $\vec{\mathfrak{g}}$ is the effective velocity of the point of contact between the two molecules, $\vec{n}$ is the normal vector to the contact surface, and $S(\vec{\nu}_1,\vec{\nu}_2)$ is the Jacobian of the transformation from the excluded volume between two molecules with orientation $\vec{\nu}_1$ and $\vec{\nu}_2$ to $\mathbb{S}^2$.
	The effective velocity $\vec{\mathfrak{g}}$ is once again given by
	\begin{equation}
		\vec{\mathfrak{g}} = \frac{1}{m}\left(\vec{p}_1-\vec{p}_2\right) +
		\begin{pmatrix}
			\omega_1r_{1}^{(y)}-\omega_2r_{2}^{(y)}\\
			\omega_2r_{2}^{(x)}-\omega_1r_{1}^{(x)}\\
		\end{pmatrix}.
	\end{equation}
	As Vlasov force we considered two different cases, a linear Vlasov force $\mathcal{V}(\theta,\omega) = -\alpha\left(\theta-\hat{\theta}\right) - \beta\omega$, where $\hat{\theta}$ is the average orientation of the system, and a non-linear Vlasov force $\mathcal{V}(\theta) = -\sin\left(\theta-\hat{\theta}\right)$.
	These Vlasov forces can lead to different dynamics, including alignment along a preferential direction corresponding the emergence of a nematic alignment, and random oscillations of each molecule corresponding to the isotropic phase.
\end{exB}
\begin{exC}[Head-tail symmetric calamitic molecules: Kinetic equation]
	\label{ex:KT_head_tail}
	In the context of head-tail symmetric calamitic molecules, we derived a kinetic equation that describes the evolution of the distribution function, where the order parameter is given by the angle $\theta$ and the angular velocity $\omega$.
	The interaction between head-tail symmetric calamitic molecules is governed by a collision rule that conserves both linear momentum and angular momentum, leading to a Boltzmann-type equation with a collision operator that accounts for the exchange of orientation during collisions.
	The resulting kinetic equation has the following form,
	\begin{align}
		\label{eq:boltzmann_head_tail}
		\frac{\partial f}{\partial t} & +\omega_1\begin{bmatrix}
			-2\sin(\theta_1) \cos(\theta_1) & \sin^2(\theta_1)-\cos^2(\theta_1) \\
			\sin^2(\theta_1)-\cos^2(\theta_1) & 2\sin(\theta_1)\cos(\theta_1)
		\end{bmatrix}\!:\!\Pi_{\mathcal{M}}\!:\!\begin{bmatrix}
			-\sin(\theta_1)^{-1} &  \frac{1}{2}\cos(2\theta_1)^{-1}\\
			\frac{1}{2}\cos(2\theta_1)^{-1} & \cos(\theta_1)^{-1}
		\end{bmatrix}\nabla_{\theta_1}f\\ 
	    &+ \mathcal{V}\cdot\nabla_{\omega_1}f\nonumber\\
		&=\frac{1}{2} \! \int \!\!\!d\Xi_1^\prime\, d\Xi_2^\prime d\Xi_2 \!\!\int_{0}^{\frac{\pi}{2}}\!\!\!\!\int_{0}^{2\pi} \!\!\!\!\!\!\!W\left( \Xi_1^\prime,\Xi_2^\prime\mapsto \Xi_1,\Xi_2\right)\Big(f(\Gamma_1^\prime,t)g(\Gamma_2^\prime,t)+f(\Gamma_2^\prime,t)g(\Gamma_1^\prime,t)\Big)\nonumber \\
	    &\qquad\qquad\qquad-W\left( \Xi_1,\Xi_2\mapsto \Xi_1^\prime,\Xi_2^\prime\right)\Big(f(\Gamma_1,t)g(\Gamma_2,t)+f(\Gamma_2,t)g(\Gamma_1,t)\Big)\, d\theta_2\, d\varphi_2,\nonumber
	\end{align}
	with $\Pi_{\mathcal{M}}$ defined as 
	\begin{equation}
		\Pi_\mathcal{M}\!\! =\!\! \begin{bmatrix}
			-2\sin(\theta_1) \cos(\theta_1) & \sin^2(\theta_1)-\cos^2(\theta_1) \\
			\sin^2(\theta_1)-\cos^2(\theta_1) & 2\sin(\theta_1)\cos(\theta_1)
		\end{bmatrix}\otimes \begin{bmatrix}
			-2\sin(\theta_1) \cos(\theta_1) & \sin^2(\theta_1)-\cos^2(\theta_1) \\
			\sin^2(\theta_1)-\cos^2(\theta_1) & 2\sin(\theta_1)\cos(\theta_1)
		\end{bmatrix}\!.
	\end{equation}
	The transition probability $W(\Xi_1,\Xi_2\mapsto \Xi_1^\prime,\Xi_2^\prime)$ is given by
	\begin{equation}
		W(\Xi_1,\Xi_2\mapsto \Xi_1^\prime,\Xi_2^\prime) = (\vec{\mathfrak{g}}\cdot \vec{n})S(\vec{\nu}_1,\vec{\nu}_2),
	\end{equation}
	where $\vec{\mathfrak{g}}$ is the effective velocity of the point of contact between the two molecules, $\vec{n}$ is the normal vector to the contact surface, and $S(\vec{\nu}_1,\vec{\nu}_2)$ is the Jacobian of the transformation from the excluded volume between two molecules with orientation $\vec{\nu}_1$ and $\vec{\nu}_2$ to $\mathbb{S}^2$.
	The effective velocity $\vec{\mathfrak{g}}$ is given by
	\begin{equation}
		\vec{\mathfrak{g}} = \frac{1}{m}\left(\vec{p}_1-\vec{p}_2\right) +
		\begin{pmatrix}
			\omega_1r_{1}^{(y)}-\omega	_2r_{2}^{(y)}\\
			\omega_2r_{2}^{(x)}-\omega_1r_{1}^{(x)}\\
		\end{pmatrix}.
	\end{equation}
	We suggest the same Vlasov forces as in \cref{ex:KT_calamitic}, to mimic different phases of the system.	
\end{exC}

This work lays the foundation for future developments in the kinetic modelling of ordered fluids. Many open questions remain, including numerical simulations of the derived kinetic equations and studying their hydrodynamic limits, which will likely require more refined closure assumptions~\cite{farrellEtAll, farrellMalek}.

\section*{Acknowledgments}
The authors would like to acknowledge A.~Janigro and G.\ Russo for fruitful discussions on this work.
\\
\appendix

\section{Frame indifference in terms of group actions}
\label{sec:frameIndifference}
We here discuss the frame indifference hypothesis and its relation to the group action $\mathcal{A}$ introduced in Section \ref{sec:symmetry}.
Let us consider a function $\mathcal{L}:\mathcal{M}\to \mathbb{R}$, where $(\mathcal{M},\mathcal{A})$ is an order parameter manifold introduced in \cref{def:orderParameterManifold} with \gpt{$\iota$ being} the dimension of the manifold $\mathcal{M}$.
The differential of $\mathcal{L}$ evaluated at $\nu$ is a map of the form
\begin{equation}
	D\mathcal{L}\Big \vert_\nu: T_\nu \mathcal{M} \simeq \mathbb{R}^\iota \to T_{\mathcal{L}(\nu)}\mathbb{R}\simeq \mathbb{R}.
\end{equation}
Fixing $\nu \in \mathcal{M}$, we define an auxiliary map describing the effect of rotations on the Lagrangian:
\begin{equation}
	\Phi:SO(d)\to \mathbb{R},\qquad \Phi(\mat{Q}) = \mathcal{L}(\mathcal{A}(\mat{Q},\nu)),
\end{equation}
and observe that the differential of $\Phi$ is a map of the form
\begin{equation}
	D\Phi\Big \vert_Q: T_{\mat{Q}=\mat{\text{Id}}} SO(d) \simeq \mathfrak{so}(d)\to T_{\Phi(\mat{Q})}\mathbb{R}\simeq \mathbb{R}.
\end{equation}
In particular, we can compute, using the chain rule, the value of $D\Phi\Big \vert_{\mat{Q}=\mat{\text{Id}}}$ by making use of the infinitesimal generator of the group action $\mathcal{A}$, i.e. $A_\nu\Big\lvert_{Q=Id}: \mathfrak{so}(d)\to T_\nu\mathcal{M}$.
\begin{equation}
	\begin{tikzcd}
		\mathfrak{so}(d)\arrow[r,"A_\nu"]\arrow[d, phantom, sloped, "\simeq"] & T_{\nu}\mathcal{M}\arrow[d, phantom, sloped, "\simeq"]\arrow[r,"D\mathcal{L}\lvert_{\nu}"]&T_{\mathcal{L}(\nu)}\mathbb{R}\arrow[d, phantom, sloped, "\simeq"]\\
		\mathbb{R}^{\frac{d(d-1)}{2}}\arrow[r, "C"] & \mathbb{R}^{\iota}\arrow[r,"b"]&\mathbb{R}
	\end{tikzcd}
\end{equation}
Thus, we can write the differential of $\Phi$ at $\mat{Q}=\mat{\text{Id}}$ as $C^T b$, where $C$ is the representation of $A_\nu$ in $\mathbb{R}^{\frac{d(d-1)}{2}\times \iota}$ and $b$ is the \gpt{representation} of $D\mathcal{L}\Big\lvert_{\nu}$ in $\mathbb{R}^{\iota\times 1}$.
Hence, we can express the frame indifference hypothesis for the Lagrangian $\mathcal{L}=\mathcal{L}(\nu)$ as the condition:
\begin{equation}
	A_\nu^T \frac{\partial\mathcal{L}}{\partial \nu}(\nu) = 0.
\end{equation}
Similarly, we consider a Lagrangian $\mathcal{L}:\mathcal{M}\times \mathbb{R}^\iota\to \mathbb{R}$, and define the differential of $\mathcal{L}$ as a map of the form
\begin{equation}
	D\mathcal{L}\Big \vert_{(\nu,\vec{\dot{\nu}})}: T_\nu \mathcal{M} \times T_{\vec{\dot{\nu}}}\mathbb{R}^\iota\to T_{\mathcal{L}(\nu,\vec{\dot{\nu}})}\mathbb{R}.
\end{equation}
In this case, we define the auxiliary map for a fixed $\nu \in \mathcal{M}$ and $\vec{\dot{\nu}}\in \mathbb{R}^\iota$ as
\begin{equation}
	\Phi:SO(d)\to \mathbb{R},\qquad \Phi(\mat{Q}) = \mathcal{L}(\mathcal{A}(\mat{Q},\nu),\vec{\dot{\nu}} + D(A_\nu \vec{q})\Big\lvert_{\nu}\vec{\dot{\nu}} ).
\end{equation}
\gpt{Similarly to} before, we compute the differential of $\Phi$ at $\mat{Q}=\mat{\text{Id}}$ using the chain rule, first observing that
\begin{equation}
	\begin{tikzcd}
		\mathfrak{so}(d)\arrow[r,"DJ\lvert_{\mat{Q}=\mat{\text{Id}}}"{yshift=2pt}]\arrow[d, phantom, sloped, "\simeq"] & T_{\vec{\dot{\nu}}}\mathbb{R}^\iota\arrow[d, phantom, sloped, "\simeq"]\arrow[r,"D\mathcal{L}\lvert_{(\cdot,\vec{\dot{\nu}})}"yshift=3pt]&T_{\mathcal{L}(\nu,\vec{\dot{\nu}})}\mathbb{R}\arrow[d, phantom, sloped, "\simeq"]\\
		\mathbb{R}^{\frac{d(d-1)}{2}}\arrow[r, "\mathfrak{C}"] & \mathbb{R}^{\iota}\arrow[r,"\mathfrak{b}"]&\mathbb{R}
	\end{tikzcd}
\end{equation}
where $J(\mat{Q})=\vec{\dot{\nu}} + D(A_\nu \vec{q})\Big\lvert_{\nu}\vec{\dot{\nu}}$, and $\mathfrak{C}$ is the representation of $DJ\lvert_{\mat{Q}=\mat{\text{Id}}}$ in $\mathbb{R}^{\frac{d(d-1)}{2}\times \iota}$, i.e. $\mathfrak{C}=\frac{\partial A_\nu}{\partial \nu}\vec{\dot{\nu}}$, and $\mathfrak{b}$ is the representation of $D\mathcal{L}\lvert_{(\cdot,\vec{\dot{\nu}})}$ in $\mathbb{R}^{\iota\times 1}$.
Thus, we can express the frame indifference hypothesis for the Lagrangian $\mathcal{L}=\mathcal{L}(\nu,\vec{\dot{\nu}})$ as:
\begin{equation}
	\frac{\partial\mathcal{L}_1(\nu, \vec{\dot{\nu}})}{\partial \mat{Q}}\coloneqq {A^T_{\nu}}{\frac{\partial \mathcal{L}_1}{\partial \nu}} + {\left(\frac{\partial A_{\nu}}{\partial \nu}\vec{\dot{\nu}}\right)^T}{\frac{\partial \mathcal{L}_1}{\partial \vec{\dot{\nu}}}} = \vec{0}.
\end{equation}
Notice that the above condition is equivalent to \eqref{eq:frameIndifferenceLagrangian}.
\section{Mathematical aspects of the collision operator}
\label{sec:mathaspects}
In this appendix we will discuss some mathematical aspects of the collision operator \eqref{eq:collision_operator_general}, in particular we will here present some of the more tedious proofs that we have omitted in the main text, for the sake of clarity.
In particular, we will here report the proof of Theorem \ref{thm:boltzmannInequality} which follows from the work of \cite{cercignaniLampis}.
We will also discuss some technicalities related to the proof of \eqref{thm:maxwellian}.
\begin{proof}[Proof of Theorem \ref{thm:boltzmannInequality}]
The idea behind the approach of \cite{cercignaniLampis} is to start from \eqref{eq:collision_invariants} and obtain the following equation, after suitable changes of variables \cite{cercignaniLampis},
\begin{equation}
	\begin{alignedat}{2}
		&\int d\Xi_1\, \psi(\Gamma_1)C[f,f]=\\
		&\frac{1}{2} \!\int \!\!\!d\Xi_1^\prime\, d\Xi_2^\prime d\Xi_2\,d\Xi_1 \!\!\int_{0}^{\frac{\pi}{2}}\!\!\!\int_{0}^{2\pi} \!\!\!\!\!\!\!W\!\left( \Xi_1,\Xi_2\!\mapsto\! \Xi_1^\prime,\Xi_2^\prime\right)&\!\times\! \left(\psi(\Gamma_1,t)\!+\!\psi(\Gamma_2,t)\!-\!\psi(\Gamma_1^\prime,t)\!-\!\psi(\Gamma_2^\prime,t)\right)\\
		&\times\Big(f(\Gamma_1,t)f(\Gamma_2,t)\Big)\,d\theta_2\,d\varphi_2\nonumber.\\
	\end{alignedat}
\end{equation} 
We then pick as $\psi(\Gamma_1)$ the quantity $\log(f(\Gamma_1,t))$ to obtain the following equation,
\begin{equation}
	\label{eq:collision_invariants_cercignaniLampis}
	\begin{alignedat}{2}
		\int &d\Xi_1\, \log(f(\Gamma_1,t))C[f,f]=\\
		&\frac{1}{2} \int \!\!\!d\Xi_1^\prime\, d\Xi_2^\prime d\Xi_2\,d\Xi_1 \!\!\int_{0}^{\frac{\pi}{2}}\!\!\!\int_{0}^{2\pi} \!\!\!\!\!\!\!W\left( \Xi_1,\Xi_2\mapsto \Xi_1^\prime,\Xi_2^\prime\right)& \log \left(\frac{f(\Gamma_1^\prime,t)f(\Gamma_2^\prime,t)}{f(\Gamma_1,t)f(\Gamma_2,t)}\,d\theta_2\,d\varphi_2\right)\\
		&\times \Big(f(\Gamma_1,t)f(\Gamma_2,t)\Big)\,d\theta_2\,d\varphi_2\nonumber
	\end{alignedat}
\end{equation}
Using \eqref{eq:cercignaniLampis} together with \eqref{eq:collision_invariants} with $\psi(\Gamma_1)=1$ we can obtain the following identity,
\begin{align}
	&\int d\Xi_1^\prime\, d\Xi_2^\prime d\Xi_2\,d\Xi_1 \int_{0}^{\frac{\pi}{2}}\!\!\!\int_{0}^{2\pi} W( \Xi_1,\Xi_2\mapsto \Xi^\prime,\Xi_2^\prime) f(\Gamma_1^\prime,t)f(\Gamma_2^\prime,t)\,d\theta_2\,d\varphi_2 \\
	=&\int d\Xi_1^\prime\, d\Xi_2^\prime d\Xi_2\,d\Xi_1\int_{0}^{\frac{\pi}{2}}\!\!\!\int_{0}^{2\pi}  W( \Xi_1,\Xi_2\mapsto \Xi^\prime,\Xi_2^\prime) f(\Gamma_1,t)f(\Gamma_2,t)\,d\theta_2\,d\varphi_2.\nonumber
\end{align}
Using the previous identity we can rewrite \eqref{eq:collision_invariants_cercignaniLampis} as
\begin{equation}
	\begin{alignedat}{2}
		&\int d\Xi_1\, \log(f(\Gamma_1,t))C[f,f]=\\
		&\frac{1}{2} \!\int \!\!\!d\Xi_1^\prime\, d\Xi_2^\prime d\Xi_2\,d\Xi_1 \!\!\int_{0}^{\frac{\pi}{2}}\!\!\!\int_{0}^{2\pi} \!\!\!\!\!\!\!W\!\left( \Xi_1,\Xi_2\!\mapsto\! \Xi^\prime,\Xi_2^\prime\right)\Big(f(\Gamma_1,t)f(\Gamma_2,t)\Big)\\
		&\times \left(\log \left(\frac{f(\Gamma_1^\prime,t)f(\Gamma_2^\prime,t)}{f(\Gamma_1,t)f(\Gamma_2,t)}\right)-\left(\frac{f(\Gamma_1^\prime,t)f(\Gamma_2^\prime,t)}{f(\Gamma_1,t)f(\Gamma_2,t)}-1\right)\right)\,d\theta_2\,d\varphi_2\\
	\end{alignedat}
\end{equation}
Using the fact that the one particle distribution function is non-negative as is $W(\cdot\mapsto \cdot)$ we know that
\begin{equation}
	\label{eq:boltzmannInequality}
	\int d\Xi_2\, \log(f(\Gamma_2,t))C[f,f]\leq 0,
\end{equation}
since $\log(x)\leq x-1$ for all $x\geq 0$. In particular, the previous inequality becomes an identity if and only if $x=1$, i.e. if and only if
\begin{equation}
	f(\Gamma_1^\prime,t)f(\Gamma_2^\prime,t)=f(\Gamma_1,t)f(\Gamma_2,t).
\end{equation}
\end{proof}
\begin{lemma}
	\label{lem:Cauchy}
	Let $f:\mathbb{R}^d \to \mathbb{R}$ be a continuous function such that $f(\vec{x}) + f(\vec{y}) = f(\vec{x}+\vec{y})$ for all $\vec{x},\vec{y}\in \mathbb{R}^d$, then $f$ is a linear function, i.e. $f(\vec{x})=a+b\cdot \vec{x}$ for some $a\in \mathbb{R}$ and $b\in \mathbb{R}^d$.
\end{lemma}
In order to prove that \cref{thm:maxwellian} we will have to prove that the collision invariants we obtained in Section \ref{sec:symmetry} are the only possible ones.
This problem has been dealt with by Boltzmann himself under the hypothesis that $\phi\in C^2(\mathbb{R}^3,\mathbb{R})$. Most modern proof of the fact there are only five linearly-independent invariants to the Boltzmann equation, see \cite{carleman, grad,cercignaniInvariants, arkeryd, arkerydCercignani}, are based on Gr\"onwall approach to reformulate \eqref{eq:collision_invariants} as Cauchy's additive functional equation, i.e.
$f(\vec{x}+\vec{y})=f(\vec{x})+f(\vec{y})$.
This reformulation clearly makes use of an additive structure in the phase space that might be lost when working with an order parameter manifold, for this reason we will now proceed proving \cref{thm:maxwellian} in the context of the present paper.
\begin{proof}[Proof of \cref{thm:maxwellian}]
	We begin taking a closer look at the collision invariants introduced in Section \ref{sec:symmetry}, i.e.
	\begin{enumerate}
		\item the microscopic linear momentum, i.e. $\vec{p_1} $,
		\item the microscopic angular momentum, i.e. $A_{\pt{x}_1} \mat{Q}\cdot \vec{p}_1+A_{\nu_1} \mat{Q}\cdot \mat{B}(\nu_1)\vec{\dot{\nu_1}}$,
		\item the microscopic kinetic energy, i.e. $m^{-1}\vec{p}_1\cdot\vec{p}_1+\vec{\varsigma}\cdot \mat{B}(\nu_1)^{-1}\vec{\varsigma_1}$,
	\end{enumerate}
	and we observe that these are not minimal, in the sense of linear independence. In fact, we know the conservation of linear momentum implies the conservation of $A_{\pt{x}_1} \mat{Q}\cdot \vec{p}_1$ and thus we are left considering the first and the last invariants, together with $A_{\nu_1} \mat{Q}\cdot \mat{B}(\nu_1)\vec{\dot{\nu_1}}$.
	Next we observe that under the hypothesis that $\mathcal{A}$ is transitive by \cref{prop:constant} we can further simply the invariant $A_{\nu_1} \mat{Q}\cdot \mat{B}(\nu_1)\vec{\dot{\nu_1}}$ to $A_{\nu_1}Q\cdot \mat{B}\,\vec{\dot{\nu_1}}$.
	Ultimately, given that the group action $\mathcal{A}$ is transitive, we can further simplify the invariant $A_{\nu_1}Q\cdot \mat{B}\,\vec{\dot{\nu_1}}$ to $\mat{B}\,\vec{\dot{\nu_1}}$ which is equivalent to $\vec{\varsigma}_1$.
	Hence the collision invariants can be expressed as functions of $(\vec{p}_i,\vec{\varsigma}_i)$ only, i.e. $\phi(\Gamma_i)=\phi(\vec{p}_i,\vec{\varsigma}_i)$.
	
	Notice that the portion of the phase space $(\vec{p}_i,\vec{\varsigma}_i)$ is endowed with the structure of a vector space we need.
	From this point on we can follow the same steps as in \cite{arkerydCercignani} to prove that the only possible linear independent invariants are the ones we have introduced in Section \ref{sec:symmetry}, i.e. the microscopic linear momentum, the microscopic angular momentum, and the microscopic kinetic energy.
	
	For the sake of completeness we will here report the main steps of the proof.
	We begin studying the even part of $\phi(\vec{p}_i,\vec{\varsigma}_i)$, i.e. $\phi^{E}(\vec{p}_i,\vec{\varsigma}_i)=\frac{1}{2}\left(\phi(\vec{p}_i,\vec{\varsigma}_i)+\phi(-\vec{p}_i,-\vec{\varsigma}_i)\right)$.
	In particular, we observe that $\phi(\vec{p}_1,\vec{\varsigma}_1)+\phi(\vec{p}_2,\vec{\varsigma}_2)$ is only a function of the sum $\vec{p}_1+\vec{p}_2$, $\vec{\varsigma}_1+\vec{\varsigma}_2$, and of $\vec{p}_1\cdot\vec{p}_1+\vec{p}_2\cdot\vec{p}_2 + \vec{\varsigma}_1\mat{B}^{-1}\cdot\vec{\varsigma}_1+ \vec{\varsigma}_2\mat{B}^{-1}\cdot\vec{\varsigma}_2$.
	Hence $\phi^{E}(\vec{p}_1,\vec{\varsigma}_1)$ is only a function of $\vec{p}_1\cdot\vec{p}_1+ \vec{\varsigma}_1\mat{B}^{-1}\cdot\vec{\varsigma}_1$ and furthermore $\phi^{E}(\vec{p}_1,\vec{\varsigma}_1)+ \phi^{E}(\vec{p}_2,\vec{\varsigma}_2)$ will only be a function of $\vec{p}_1\cdot\vec{p}_1+ \vec{p}_2\cdot\vec{p}_2 + \vec{\varsigma}_1\mat{B}^{-1}\cdot\vec{\varsigma}_1+ \vec{\varsigma}_2\mat{B}^{-1}\cdot\vec{\varsigma}_2$, thus we can apply \cref{lem:Cauchy} to conclude that $\phi^{E}(\vec{p}_1,\vec{\varsigma}_1)$ is a linear function of $\vec{p}_1\cdot \vec{p}_1+ \vec{\varsigma}_1\mat{B}^{-1}\cdot\vec{\varsigma}_1$, i.e. $\phi^{E}(\vec{p}_1,\vec{\varsigma}_1)=\vec{a}\cdot \vec{p}_1+\vec{b}\cdot \vec{\varsigma}_1$ for some $\vec{a},\vec{c}\in \mathbb{R}^3$.
	Considering $(\vec{p}_1,\vec{\varsigma}_1)$ and $(\vec{p}_2,\vec{\varsigma}_2)$ orthogonal to each other, we can prove that the odd part of $\phi$, i.e. $\phi^{O}(\vec{p}_1,\vec{\varsigma}_1)\coloneqq \frac{1}{2}\left(\phi(\vec{p}_1,\vec{\varsigma}_1)-\phi(-\vec{p}_1,-\vec{\varsigma}_1)\right)$ verifies the hypothesis of \cref{lem:Cauchy} thus it must a linear function of $(\vec{p}_1,\vec{\varsigma}_1)$, i.e. $\phi^{O}(\vec{p}_1,\vec{\varsigma}_1)=\vec{a}\cdot \vec{p}_1+\vec{b}\cdot \vec{\varsigma}_1$ for some $\vec{a},\vec{b}\in \mathbb{R}^3$.
	Finally, summing the \rev{odd} and even parts of $\phi$ we obtain the desired result.
\end{proof}
\bibliographystyle{siamplain}
\bibliography{references}
\end{document}

%% file: ex_shared.tex
\usepackage{lipsum}
\usepackage{amsfonts}
\usepackage{graphicx}
\usepackage{epstopdf}
\usepackage{algorithmic}
\usepackage{accents}
\usepackage{booktabs}
\usepackage{soul}
\ifpdf
  \DeclareGraphicsExtensions{.eps,.pdf,.png,.jpg}
\else
  \DeclareGraphicsExtensions{.eps}
\fi

\newsiamremark{remark}{Remark}
\newsiamremark{hypothesis}{Hypothesis}
\crefname{hypothesis}{Hypothesis}{Hypotheses}
\newsiamremark{claim}{Claim}
\newsiamremark{exA}{Example A}
\newsiamremark{exB}{Example B}
\newsiamremark{exC}{Example C}
\newsiamremark{exD}{Example D}
\newsiamremark{fact}{Fact}
\crefname{fact}{Fact}{Facts}

\headers{A Kinetic Theory Approach To Ordered Fluids}{J.~A.~Carrillo, P.~E.~Farrell, A.~Medaglia, U.~Zerbinati}

\title{A Kinetic Theory Approach To Ordered Fluids\thanks{Submitted to the editors DATE.
\funding{This work was funded by the Engineering and Physical Sciences Research Council
[grant numbers EP/R029423/1, EP/W026163/1, EP/T022132/1 and EP/V051121/1], by the Advanced Grant Nonlocal-CPD
(Nonlocal PDEs for Complex Particle Dynamics: Phase Transitions, Patterns and Synchronization)
 of the European Research Council Executive Agency (ERC) under the European Union Horizon 2020 research and innovation programme (grant agreement No.\ 883363),
by the ``Maria de Maeztu'' Excellence Unit IMAG, reference CEX2020-001105-M, funded
by MCIN/AEI/10.13039/501100011033/,
 by the Donatio Universitatis Carolinae Chair ``Mathematical modelling of multicomponent systems'', and the UKRI Digital Research Infrastructure Programme through the Science and Technology Facilities Council's Computational Science Centre for Research Communities (CoSeC).
}}}

\author{
José A. Carrillo\thanks{Mathematical Institute, University of Oxford, Oxford, UK(\email{carrillo@maths.ox.ac.uk}).} 
\and Patrick E.~Farrell\thanks{Mathematical Institute, University of Oxford, Oxford, UK and Mathematical Institute, Faculty of Mathematics and Physics, Charles University, Czechia (\email{patrick.farrell@maths.ox.ac.uk}).}
\and Andrea Medaglia\thanks{Department of Mathematics and Computer Science, University of Ferrara, Italy (\email{andrea.medaglia@unife.it}, \url{https://www.andreamedaglia.eu/}).}
\and Umberto Zerbinati\thanks{Mathematical Institute, University of Oxford, Oxford, UK (\email{zerbinati@maths.ox.ac.uk}, \url{https://www.uzerbinati.eu}).}
}

\usepackage{amsopn}

\DeclareMathOperator{\Forall}{\forall}

\usepackage{mathtools}
\usepackage{todonotes}
\makeatletter
\newcommand{\mat}[1]{{\mathpalette\mat@{#1}}}
\newcommand{\mat@}[2]{%
  \begingroup
  \sbox\z@{$\m@th#1\underline{#2}$}%
  \dimen@=\dp\z@ \advance\dimen@ -2\mat@dimen{#1}%
  \dp\z@=\dimen@
  \sbox\z@{$\m@th\underline{\box\z@}$}%
  \box\z@
  \endgroup
}
\newcommand\mat@dimen[1]{%
  \fontdimen8
  \ifx#1\displaystyle\textfont\else
  \ifx#1\textstyle\textfont\else
  \ifx#1\scriptstyle\scriptfont\else
  \scriptscriptfont\fi\fi\fi 3
}
\newcommand{\tensor}[1]{{\mathpalette\tensor@{#1}}}
\newcommand{\tensor@}[2]{%
  \begingroup
  \sbox\z@{$\m@th#1\underline{#2}$}%
  \dimen@=\dp\z@ \advance\dimen@ -2\tensor@dimen{#1}%
  \dp\z@=\dimen@
  \sbox\z@{$\m@th\underline{\box\z@}$}%
  \box\z@
  \endgroup
}
\newcommand\tensor@dimen[1]{%
  \fontdimen8
  \ifx#1\displaystyle\textfont\else
  \ifx#1\textstyle\textfont\else
  \ifx#1\scriptstyle\scriptfont\else
  \scriptscriptfont\fi\fi\fi 3
}
\def\contraction{\vbox{\baselineskip2.5\p@ \lineskiplimit\z@
  \kern\p@\hbox{.}\hbox{.}\hbox{.}\hbox{.}}}
\makeatother

\newcommand{\abs}[1]{\lvert#1\rvert}
\newcommand{\norm}[1]{\lVert #1\rVert}

\DeclareMathOperator{\tr}{tr}
\let \vec \underline
\let \pt \mathbf
\newsiamremark{example}{Example}
\makeatletter
\newcommand{\hathat}[1]{%
\begingroup%
  \let\macc@kerna\z@%
  \let\macc@kernb\z@%
  \let\macc@nucleus\@empty%
  \hat{\raisebox{.28ex}{\vphantom{\ensuremath{#1}}}\smash{\hat{#1}}}%
\endgroup%
}
\makeatother

\newcommand{\gpt}[1]{\textcolor{black}{#1}}
\newcommand{\rev}[1]{#1}

\usepackage{tikz}
\usepackage{tikzit}
\input{main.tikzstyles}
\usetikzlibrary{calc}
\usetikzlibrary{3d,perspective}
\tikzset
{
  axis/.style={thick,-latex},
  my view/.style={3d view={65}{20}},
  nutation/.style={rotate around x=\nut}
}
\usetikzlibrary{patterns,snakes}
\usetikzlibrary{arrows.meta} 
\tikzstyle{spring}=[line width=0.8,blue!7!black!80,snake=coil,segment amplitude=1,segment length=1.25,line cap=round]
\tikzstyle{largeSpring}=[line width=0.8,blue!7!black!80,snake=coil,segment amplitude=2,segment length=3.25,line cap=round]
\colorlet{my cyan}{cyan!50}
\definecolor{my magenta}{RGB}{0,128,128}
\definecolor{verde}{RGB}{0,128,128}
\definecolor{rosso}{RGB}{202,0,0}
\usepackage{pgfplots}
\pgfplotsset{compat = newest}
\usetikzlibrary{intersections, pgfplots.fillbetween}
\usepackage{tikz-cd}
\usetikzlibrary{babel}

%% file: main.tikzstyles
\tikzstyle{GreenNeuron}=[fill={rgb,255: red,0; green,128; blue,128}, draw={rgb,255: red,0; green,128; blue,128}, shape=circle]
\tikzstyle{BlueNeuron}=[fill={rgb,255: red,0; green,0; blue,173}, draw=none, shape=circle]
\tikzstyle{RedNeuron}=[fill={rgb,255: red,202; green,0; blue,0}, draw=none, shape=circle]
\tikzstyle{BlackCirc}=[fill=black, draw=black, shape=circle, font={\tiny}, scale=0.4]
\tikzstyle{tinyGreen}=[fill={rgb,255: red,0; green,128; blue,128}, draw={rgb,255: red,0; green,128; blue,128}, shape=circle, scale=0.4]
\tikzstyle{tinyRed}=[fill={rgb,255: red,202; green,0; blue,0}, draw={rgb,255: red,202; green,0; blue,0}, shape=circle, scale=0.4]
\tikzstyle{RedDisc}=[fill=none, draw={rgb,255: red,191; green,0; blue,64}, shape=circle, scale=0.75]
\tikzstyle{smallRedDisk}=[fill=none, draw={rgb,255: red,191; green,0; blue,64}, shape=circle, scale=0.6]
\tikzstyle{BigRedDisk}=[fill=none, draw={rgb,255: red,191; green,0; blue,64}, shape=circle, scale=0.9]
\tikzstyle{SuperBigDisckRed}=[fill=none, draw={rgb,255: red,191; green,0; blue,64}, shape=circle, scale=1.05]
\tikzstyle{tinyBlack}=[fill=black, draw=black, shape=circle, scale=0.4]

\tikzstyle{Arrow}=[->, draw=black]
\tikzstyle{Green}=[-, draw={rgb,255: red,0; green,128; blue,128}, thick]
\tikzstyle{GreenDashed}=[-, draw={rgb,255: red,0; green,128; blue,128}, thick, dashed]
\tikzstyle{Dashed}=[-, dashed]
\tikzstyle{RedDashed}=[-, thick, dashed, draw={rgb,255: red,202; green,0; blue,0}]
\tikzstyle{BlueDashed}=[-, draw={rgb,255: red,0; green,0; blue,173}, thick, dashed]
\tikzstyle{dotRed}=[-, draw={rgb,255: red,202; green,0; blue,0}, dotted]
\tikzstyle{Thikblue}=[-, draw={rgb,255: red,0; green,0; blue,173}, thick]
\tikzstyle{ArrowRed}=[draw={rgb,255: red,191; green,0; blue,64}, ->]
\tikzstyle{doubleRedArrow}=[->, draw={rgb,255: red,191; green,0; blue,64}, double]
\tikzstyle{blueArrow}=[draw=blue, ->]
\tikzstyle{spring}=[-, spring]
\tikzstyle{LargeSpring}=[-, largeSpring]
\tikzstyle{thikRedArrow}=[->, thick, draw={rgb,255: red,191; green,0; blue,64}]
\tikzstyle{empty}=[-, draw=none]
\tikzstyle{greenArrow}=[->, draw={rgb,255: red,0; green,128; blue,128}]
\tikzstyle{dubleArrow}=[<->]
\tikzstyle{dashedArrow}=[dashed, ->]
\tikzstyle{double Green Arrow}=[<->, draw={rgb,255: red,0; green,128; blue,128}]

%% file: Figures/collision_sphere.tikz
	\begin{tikzpicture}
    	\def\R{1.5}
    	\def\dx{2*\R}  
    	\coordinate (x1) at (-\R,0);
    	\coordinate (x2) at (\R,0);
    	\draw[thick] (x1) circle (\R);
    	\draw[thick] (x2) circle (\R);
    	\draw[style=GreenDashed] (0,-2) -- (0,2);
    	\node at (x1) [below left] {$\pt{x}_1$};
    	\node at (x2) [below right] {$\pt{x}_2$};
    	\draw[->, thick] (x1) -- ($(x1)+(0.7,0.3)$) ;
    	\draw[->, thick] (x2) -- ($(x2)+(-0.7,0.3)$) ;
    	\draw[style=greenArrow] (0,1.75) -- (0.5,1.75) node[above] {$\color{verde}\vec{n}$};
\end{tikzpicture}

%% file: Figures/collision_bubbles.tikz
\begin{tikzpicture}

\def\RlargeIn{1.3}
\def\RsmallIn{0.7}

\def\RlargeOut{1.0}
\def\RsmallOut{1.0}

\def\dx{2.5}
\def\dy{2}

\def\delta{0.3}

\def\centerGap{0.4}

\coordinate (xin1) at (-\dx,-\dy); 
\coordinate (xin2) at (-\dx,\dy);  

\coordinate (xout1) at (\dx,-\dy); 
\coordinate (xout2) at (\dx,\dy);  
\coordinate (xout11) at (\dx-0.2,-\dy+0.2); 
\coordinate (xout22) at (\dx-0.2,\dy-0.2);  

\draw[thick] (xin1) circle (\RlargeIn);
\draw[thick] (xin2) circle (\RsmallIn);

\draw[thick] (xout1) circle (\RlargeOut);
\draw[thick] (xout2) circle (\RsmallOut);

\node at (xin1) {$\pt{x}_1$};
\node at (xin2) {$\pt{x}_2$};
\node at (xout1) {$\pt{x}_1$};
\node at (xout2) {$\pt{x}_2$};

\coordinate (collisionIn1) at ($(1.2,1)!\centerGap!(xin1)$);
\coordinate (collisionIn2) at ($(1.2,-1)!\centerGap!(xin2)$);
\coordinate (collisionOut1) at ($(0.8,-0.65)!-\centerGap!(xout1)$);
\coordinate (collisionOut2) at ($(0.8,0.65)!-\centerGap!(xout2)$);

\draw[->, thick] ($(xin1)!\delta!(0,0)$) -- (collisionIn1);
\draw[->, thick] ($(xin2)!\delta!(0,0)$) -- (collisionIn2);

\draw[->, thick] (collisionOut1) -- (xout11);
\draw[->, thick] (collisionOut2) -- (xout22);

\end{tikzpicture}

%% file: Figures/collision_calamitic.tikz
\begin{tikzpicture}
	\begin{pgfonlayer}{nodelayer}
		\node [style=none] (0) at (1, 3.75) {};
		\node [style=none] (1) at (-2.5, -1.5) {};
		\node [style=none] (2) at (1.25, 2.25) {};
		\node [style=none] (3) at (1.25, -3.75) {};
		\node [style=none] (4) at (-0.75, 1) {};
		\node [style=none] (5) at (1.25, -0.75) {};
		\node [style=none] (6) at (1.25, -0.75) {};
		\node [style=none] (7) at (-1.5, -2.3) {};
		\node [style=none] (8) at (2.2, 4.36) {};
		\node [style=none] (9) at (0.75, 1.75) {};
		\node [style=none] (10) at (0.25, 2) {};
		\node [style=none] (11) at (0.55, 2.1) {$\color{verde}\vec{n}$};
		\node [style=none] (12) at (0.1, 2.8) {$\vec{\nu}_1$};
		\node [style=none] (13) at (1.55, 1.5) {$\vec{\nu}_2$};
		\node [style=none] (14) at (0.5, 3) {};
		\node [style=none] (15) at (1.25, 1.75) {};
		\node [style=none] (16) at (0, 1.1) {$\color{rosso}\vec{r}_1$};
		\node [style=none] (17) at (0.75, 0.5) {$\color{rosso}\vec{r}_2$};
	\end{pgfonlayer}
	\begin{pgfonlayer}{edgelayer}
		\draw [bend left=270, looseness=0.50] (0.center) to (1.center);
		\draw [bend left=90, looseness=0.50] (0.center) to (1.center);
		\draw [bend right=90, looseness=0.50] (2.center) to (3.center);
		\draw [bend left=90, looseness=0.50] (2.center) to (3.center);
		\draw [style=GreenDashed] (7.center) to (8.center);
		\draw [style=greenArrow] (9.center) to (10.center);
		\draw [style=ArrowRed] (4.center) to (9.center);
		\draw [style=ArrowRed] (6.center) to (9.center);
		\draw [style=Arrow] (4.center) to (14.center);
		\draw [style=Arrow] (6.center) to (15.center);
	\end{pgfonlayer}
\end{tikzpicture}

%% file: references.bib
@article{onsager,
author = {Onsager, Lars},
title = {THE EFFECTS OF SHAPE ON THE INTERACTION OF COLLOIDAL PARTICLES},
journal = {Annals of the New York Academy of Sciences},
volume = {51},
number = {4},
pages = {627-659},
year = {1949}
}

@article {ericksen,
    AUTHOR = {Ericksen, J. L.},
     TITLE = {Conservation laws for liquid crystals},
  JOURNAL = {Transactions of the Society of Rheology},
    VOLUME = {5},
      YEAR = {1961},
     PAGES = {23--34},
}

@article{ericksen2,
journal = {Colloid and Polymer Science},
pages = {117--122},
volume = {173},
number = {2},
year = {1960},
title = {Transversely isotropic fluids},
author = {Ericksen, J. L.},
}

@book{carpiz,
author = {Capriz, G.},
address = {New York},
isbn = {0387968865},
publisher = {Springer-Verlag},
series = {Springer Tracts in Natural Philosophy},
volume={35},
title = {Continua with Microstructure },
year = {1989},
}

@book{lee,
author = {Lee, John},
address = {New York, NY},
edition = {2nd ed. 2012.},
isbn = {9781441999825},
publisher = {Springer New York},
series = {Graduate Texts in Mathematics},
volume = {218},
title = {Introduction to Smooth Manifolds},
year = {2012},
}

@book{virga,
author = {Virga, Epifanio G.},
isbn = {9780367449063},
publisher = {CRC Press, Taylor \& Francis Group},
series = {Applied Mathematics and Mathematical Computation},
volume = {8},
title = {Variational Theories for Liquid Crystals},
year = {2019 - 1994},
}

@book{fasanoMarmi,
author = {Fasano, A. and Marmi, S.},
edition = {1st},
isbn = {1-383-02179-1},
lccn = {2005028822},
publisher = {Oxford University Press},
series = {Oxford graduate texts},
title = {Analytical Mechanics: An Introduction},
year = {2006},
}

@article{jeans,
  journal = {Proceedings of the Royal Society of London},
  number = {435},
  pages = {236-237},
  title = {The distribution of molecular energy},
  volume = {67},
  year = {1901},
  author = {Hopwood, J. J.},
}

@article{dahlerSatherI,
	author = {Dahler,J. S.  and Sather,N. F. },
	title = {Kinetic theory of loaded spheres. {I}},
	journal = {The Journal of Chemical Physics},
	volume = {38},
	number = {10},
	pages = {2363-2382},
	year = {1963},
}

@article{dahlerSandlerII,
    author = {Sandler, S. I. and Dahler, J. S.},
    title = {Kinetic theory of loaded spheres. {II}},
    journal = {The Journal of Chemical Physics},
    volume = {43},
    number = {5},
    pages = {1750-1759},
    year = {2004},
}

@article{taxman,
  issn = {0031-899X},
  journal = {Physical review},
  number = {6},
  pages = {1235-1239},
  title = {Classical theory of transport phenomena in dilute polyatomic gases},
  volume = {110},
  year = {1958},
  author = {Taxman, N.},
}

@book{deBoerEtAll,
	title={Studies in Statistical Mechanics},
	author={de\;Boer,J. and Uhlenbeck,G. E.},
	volume={2},
	year={1964},
	publisher={North-Holland}
}

@article{waldmann,
  title={Die Boltzmann-Gleichung f{\"u}r gase mit rotierenden molek{\"u}len},
  author={Waldmann, L.},
  journal={Zeitschrift f{\"u}r Naturforschung A},
  volume={12},
  number={8},
  pages={660--662},
  year={1957},
}

@article{snider,
  author = {Snider, R. F.},
  title = {Quantum‐mechanical modified Boltzmann equation for degenerate internal states},
  journal = {The Journal of Chemical Physics},
  volume = {32},
  number = {4},
  pages = {1051-1060},
  year = {2004},
  month = {07},
}

@article{curtissI,
	author = {Curtiss,C. F. },
	title = {Kinetic theory of nonspherical molecules},
	journal = {The Journal of Chemical Physics},
	volume = {24},
	number = {2},
	pages = {225-241},
	year = {1956},
}

@article{curtissII,
	author = {Curtiss,C. F.  and Muckenfuss,C.},
	title = {Kinetic theory of nonspherical molecules. {II}},
	journal = {The Journal of Chemical Physics},
	volume = {26},
	number = {6},
	pages = {1619-1636},
	year = {1957},
}

@article{curtissIII,
	author = {Muckenfuss,C.  and Curtiss,C. F. },
	title = {Kinetic theory of nonspherical molecules. {III}},
	journal = {The Journal of Chemical Physics},
	volume = {29},
	number = {6},
	pages = {1257-1272},
	year = {1958},
}

@article{curtissIV,
	author = {Livingston,P. M.  and Curtiss,C. F. },
	title = {Kinetic theory of nonspherical molecules. {IV}. {Angular} momentum transport coefficient},
	journal = {The Journal of Chemical Physics},
	volume = {31},
	number = {6},
	pages = {1643-1645},
	year = {1959},
}

@article{curtissV,
	author = {Curtiss,C. F.  and Dahler,J. S. },
	title = {Kinetic theory of nonspherical molecules. {V}},
	journal = {The Journal of Chemical Physics},
	volume = {38},
	number = {10},
	pages = {2352-2363},
	year = {1963},
}

@article{curtissVI,
  journal = {The Journal of Chemical Physics},
  number = {2},
  pages = {1416-1419},
  title = {The classical {Boltzmann} equation of a molecular gas},
  volume = {97},
  year = {1992},
  author = {Curtiss, C.F.},
}

@article{coleEtAll,
    author = {Cole, R. G. and Evans, D. R. and Hoffman, D. K.},
    title = {A renormalized kinetic theory of dilute molecular gases: Chattering},
    journal = {The Journal of Chemical Physics},
    volume = {82},
    number = {4},
    pages = {2061-2070},
    year = {1985},
}

@article{allenEtAll,
  journal = {Advances in Chemical Physics},
  pages = {1},
  volume = {86},
  year = {1993},
  title = {Hard convex body fluids},
  author = {Allen, M.P. and Evans, G.T. and Frenkel, D. and Mulder, B.M.},
}

@article{mori,
    author = {Mori, H.},
    title = {Transport, collective motion, and Brownian motion},
    journal = {Progress of Theoretical Physics},
    volume = {33},
    number = {3},
    pages = {423-455},
    year = {1965},
}

@book{harris,
	title={An Introduction to the Theory of the Boltzmann Equation},
	author={Harris, S.},
	year={2004},
	publisher={Courier Corporation}
}

@book{kardar,
title={Statistical Physics of Particles},
publisher={Cambridge University Press},
author={Kardar, M.},
year={2007}
}

@book{huangKerson,
publisher = {Wiley},
year = {1987},
title = {Statistical Mechanics},
edition = {2nd},
author = {Huang, Kerson},
}

@book{degennes,
author = {Gennes, Pierre-Gilles de. and Prost, J.},
address = {Oxford},
edition = {2nd},
isbn = {0198520247},
publisher = {Clarendon Press},
series = {Oxford Science Publications},
title = {The Physics of Liquid Crystals},
year = {1993},
}

@book{cercignani,
author = {Cercignani, Carlo},
address = {Cambridge},
title = {Rarefied Gas Dynamics: From Basic Concepts to Actual Calculations},
isbn = {0521650089},
lccn = {99015475},
publisher = {Cambridge University Press},
series = {Cambridge Texts in Applied Mathematics},
year = {2000},
}

@Article{cercignaniLampis,
author={Cercignani, Carlo
and Lampis, Maria},
title={On the {H-theorem} for polyatomic gases},
journal={Journal of Statistical Physics},
year={1981},
day={01},
volume={26},
number={4},
pages={795-801},
issn={1572-9613},
doi={10.1007/BF01010940},
}

@book{lorentz,
author = {Lorentz, H. A. and Zeeman, P. and Fokker, A. D.},
address = {The Hague},
lccn = {38001557 //R},
publisher = {M. Nijhoff},
title = {Collected Papers},
year = {1934 - 1939},
}

@article{farrellEtAll,
Author = {P. E. Farrell and G. Russo and U. Zerbinati},
Title = {{Kinetic derivation of an inviscid compressible {Leslie--Ericksen} equation for rarified calamitic gases}},
Year = {2024},
Journal = {Multiscale Modeling \& Simulation},
volume = {22},
issue = {4},
pages = {1585--1607},
doi = {10.1137/24M1630529},
}

@article{farrell,
  title = {Time-harmonic waves in {Korteweg} and nematic-{Korteweg} fluids},
  author = {Farrell, Patrick E. and Zerbinati, Umberto},
  journal = {Physical Review E},
  volume = {111},
  issue = {3},
  pages = {035413},
  numpages = {10},
  year = {2025},
  month = {Mar},
  publisher = {American Physical Society},
  doi = {10.1103/PhysRevE.111.035413},
}

@book{stronge,
  place={Cambridge},
  edition={2},
  title={Impact Mechanics},
  publisher={Cambridge University Press},
  author={Stronge, W. J.},
  year={2018}
}

@book{virgaSonnet,
author = {Sonnet, André M. and Virga, Epifanio G.},
address = {New York},
isbn = {9780387878140},
lccn = {2012930001},
title = {Dissipative Ordered Fluids: Theories for Liquid Crystals},
year = {2012},
publisher = {Springer},
}

@article{russoSmereka,
author = {Russo, Giovanni and Smereka, Peter},
title = {Kinetic Theory for Bubbly Flow {I}: Collisionless case},
journal = {SIAM Journal on Applied Mathematics},
volume = {56},
number = {2},
pages = {327-357},
year = {1996},
doi = {10.1137/S0036139993260563},
}

@article{russoSmereka2,
author = {Russo, Giovanni and Smereka, Peter},
title = {Kinetic Theory for Bubbly Flow {II}: Fluid Dynamic Limit},
journal  = {SIAM Journal on Applied Mathematics},
volume = {56},
number = {2},
pages = {358-371},
year = {1996},
doi = {10.1137/S0036139993260575},
}

@article{nochetto,
author = {Nochetto, Ricardo H. and Trivisa, Konstantina and Weber, Franziska},
title = {On the Dynamics of Ferrofluids: Global Weak Solutions to the {Rosensweig} System and Rigorous Convergence to Equilibrium},
journal = {SIAM Journal on Mathematical Analysis},
volume = {51},
number = {6},
pages = {4245-4286},
year = {2019},
doi = {10.1137/18M1224957},
}

@book{berisEdwards,
author = {Beris, Antony N. and Edwards, Brian J.},
address = {New York},
isbn = {9786610442577},
language = {eng},
lccn = {93020886},
publisher = {Oxford University Press},
series = {Oxford Engineering Science Series},
volume = {36},
title = {Thermodynamics of Flowing Systems with Internal Microstructure},
year = {2020 - 1994},
}

@book{GENERIC,
  title={Multiscale Thermodynamics: Introduction to GENERIC},
  author={Pavelka, Michal and Klika, V{\'a}clav and Grmela, Miroslav},
  year={2018},
  publisher={Walter de Gruyter GmbH \& Co KG}
}

@article{Pavelka,
    author = {Sýkora, Martin and Pavelka , Michal and La Mantia, Marco and Jou, David and Grmela, Miroslav},
    title = {On the relations between large-scale models of superfluid helium-4},
    journal = {Physics of Fluids},
    volume = {33},
    number = {12},
    pages = {127124},
    year = {2021},
    month = {12},
    issn = {1070-6631},
    doi = {10.1063/5.0070031},
}

@book{marsden,
author = {Marsden, Jerrold E. and Ratiu Tudor S.},
address = {New York, NY},
edition = {2nd},
isbn = {0-387-21792-4},
publisher = {Springer},
series = {Texts in Applied Mathematics},
volume = {17},
title = {Introduction to Mechanics and Symmetry: A Basic Exposition of Classical Mechanical Systems },
year = {1999},
}

@book{Klimontovich,
author = {Klimontovich, IU. L. and Massey, H. S. H. and Blunn, O. M. and Haar, D. ter},
address = {Oxford},
title = {The Statistical Theory of Non-Equilibrium Processes in a Plasma},
isbn = {1-4832-1462-1},
publisher = {Pergamon Press},
year = {1967 - 1967},
}

@report{lashmoreDavies,
  title={Kinetic theory and the Vlasov equation},
  author={Lashmore-Davies, CN},
  year={1987},
  publisher={CERN}
}

@article{cercignaniInvariants,
author = {Cercignani, Carlo},
address = {Heidelberg},
issn = {0022-4715},
journal = {Journal of Statistical Physics},
number = {5-6},
pages = {817-823},
publisher = {Springer},
title = {Are there more than five linearly-independent collision invariants for the {Boltzmann} equation?},
volume = {58},
year = {1990},
}

@article{degond,
title = {From kinetic to fluid models of liquid crystals by the moment method},
journal = {Kinetic and Related Models},
volume = {15},
number = {3},
pages = {417-465},
year = {2022},
issn = {1937-5093},
doi = {10.3934/krm.2021047},
author = {Pierre Degond and Amic Frouvelle and Jian-Guo Liu},
}

@article{suli,
author = {Barrett, John W. and S\"{u}li, Endre},
title = {Existence of global weak solutions to some regularized kinetic models for dilute polymers},
journal = {Multiscale Modeling \& Simulation},
volume = {6},
number = {2},
pages = {506-546},
year = {2007},
doi = {10.1137/060666810},
}

@book{goldstein,
publisher = {Pearson},
title = {Classical Mechanics},
year = {2014},
author = {Goldstein, Herbert and Poole, Charles P. and Safko, John L.},
address = {Harlow, Essex},
edition = {3rd},
isbn = {9781292038933},
}

@misc{carleman,
  title={Problemes math{\'e}matiques dans la th{\'e}orie cin{\'e}tique des gaz},
  author={Carleman, Torsten},
  year={1957}
}

@article{grad,
  title={On the kinetic theory of rarefied gases},
  author={Grad, Harold},
  journal={Communications on Pure and Applied Mathematics},
  volume={2},
  number={4},
  pages={331--407},
  year={1949},
  publisher={Wiley Online Library}
}

@Article{arkeryd,
author={Arkeryd, Leif},
title={On the {Boltzmann} equation part {II}: The full initial value problem},
journal={Archive for Rational Mechanics and Analysis},
year={1972},
volume={45},
number={1},
pages={17-34},
issn={1432-0673},
doi={10.1007/BF00253393},
}

@article{arkerydCercignani,
author = {Arkeryd L., Cercignani C.},
journal = {Rendiconti Lincei. Matematica e Applicazioni},
month = {5},
number = {2},
pages = {139-149},
publisher = {Accademia Nazionale dei Lincei},
title = {On a functional equation arising in the kinetic theory of gases},
url = {http://eudml.org/doc/244090},
volume = {1},
year = {1990},
}

@Article{JR85,
  author =      {Jenkins, J.T. and Richman, M. W.},
  title =       {Grad's $13$-moment system for a dense gas of inelastic spheres},
  journal =     {Archive for Rational Mechanics and Analysis},
  volume =      87,
  year =        1985,
  pages =       355
}

@Article{Goldhirsch2003,
  author =   {Goldhirsch, I.},
  title =    {Rapid Granular Flows},
  journal =      {Annual Review of Fluid Mechanics},
  year =     {2003},
  volume =   {35},
  pages =    {267}
}

@book{oxfordbook,
	Author = {Brilliantov, Nikolai V and P{\"o}schel, Thorsten},
	Publisher = {Oxford University Press},
	Title = {Kinetic Theory of Granular Gases},
	Year = {2010}
}

@article{grmela,
    author = {Grmela, Miroslav and Lafleur, Pierre G.},
    title = {Kinetic theory and hydrodynamics of rigid body fluids},
    journal = {The Journal of Chemical Physics},
    volume = {109},
    number = {16},
    pages = {6956-6972},
    year = {1998},
    month = {10},
    issn = {0021-9606},
    doi = {10.1063/1.477332},
}

@article{OzakiEtAll,
  title = {Effective viscosity for nematic-liquid-crystal viscosity measurement using a shear horizontal wave},
  author = {Ozaki, Ryotaro and Aoki, Masashi and Yoshino, Katsumi and Toda, Kohji and Moritake, Hiroshi},
  journal = {Physics Review E},
  volume = {81},
  issue = {6},
  pages = {061703},
  numpages = {8},
  year = {2010},
  publisher = {American Physical Society},
}

@article{impactParadox,
author = {Peter Palffy-Muhoray and Epifanio G Virga and Mark Wilkinson and Xiaoyu Zheng},
title ={On a paradox in the impact dynamics of smooth rigid bodies},
journal = {Mathematics and Mechanics of Solids},
volume = {24},
number = {3},
pages = {573-597},
year = {2019},
doi = {10.1177/1081286517751262},
}

@misc{kanzler,
      title={First order non-instantaneous corrections in collisional kinetic alignment models}, 
      author={Laura Kanzler and Carmela Moschella and Christian Schmeiser},
      year={2025},
      eprint={2503.05686},
      archivePrefix={arXiv},
      primaryClass={math.AP},
      url={https://arxiv.org/abs/2503.05686}, 
}

@paper{farrellMalek,
	title={A kinetic interpretation of thermomechanical restrictions of continua},
	author={Farrell, P. E. and Málek, J. and Souček, O. and Zerbinati, U.},
	journal={in preparation},
	year={2026},
}
